\documentclass[a4paper,notitlepage]{amsbook}
\usepackage{mathpazo}
\usepackage[T1]{fontenc}
\usepackage[latin9]{inputenc}
\usepackage{array}
\usepackage{calc}
\usepackage{amsthm}
\usepackage{amsbsy}
\usepackage{graphicx}
\usepackage{setspace}
\usepackage{amssymb}
\usepackage[numbers]{natbib}
\usepackage[dot]{bibtopic}
\usepackage[all]{xy}
\makeatletter

\providecommand{\tabularnewline}{\\}

\numberwithin{section}{chapter}
\usepackage{enumitem}		
\newlength{\lyxlistindent}      
\theoremstyle{plain}
\ifx\thechapter\undefined
\newtheorem{thm}{Theorem}
\else
\newtheorem{thm}{Theorem}[chapter]
\fi
  \theoremstyle{definition}
  \newtheorem{defn}[thm]{Definition}
 \theoremstyle{definition}
  \newtheorem{example}[thm]{Example}
  \theoremstyle{plain}
  \newtheorem{lem}[thm]{Lemma}
  \theoremstyle{remark}
  \newtheorem{rem}[thm]{Remark}
  \theoremstyle{plain}
  \newtheorem{cor}[thm]{Corollary}
  \theoremstyle{plain}
  \newtheorem{prop}[thm]{Proposition}

\usepackage[matha,mathb]{mathabx}
\usepackage{layouts}
\usepackage{bbm}
\usepackage{hyperref}
\usepackage{mathrsfs}
\usepackage{mathtools}
\usepackage[all]{xy}
\usepackage{feynmf}
\usepackage{datetime}
\usepackage{ifsym}

\DeclareMathAlphabet{\mathpzc}{OT1}{pzc}{m}{it}

\renewcommand{\thechapter}{\Roman{chapter}}

\numberwithin{equation}{chapter}

\theoremstyle{remark}

\newtheorem*{not*}{Notation}

\newcommand{\adv}{\mathrm{adv}}

\newcommand{\ap}{\mathcal{A}} 
\newcommand{\Aut}{\mathrm{Aut}}

\newcommand{\bbs}{\hspace{0.1em}\boxbackslash\hspace{0.1em}}

\newcommand{\cc}{\mathpzc{c}}
\newcommand{\C}{\mathscr{C}}
\newcommand{\Char}{\mathrm{Char}}
\newcommand{\coker}{\mathrm{coker}}

\newcommand{\dc}{\mathpzc{d}}
\newcommand{\D}{\mathscr{D}} 

\newcommand{\diag}{\mathrm{diag}}
\newcommand{\Diag}{\mathrm{Diag}}
\newcommand{\DIAG}{\mathrm{DIAG}}

\renewcommand{\div}{\mathrm{div}}

\newcommand{\E}{\mathscr{E}} 

\newcommand{\F}{\boldsymbol{\mathbb{F}}}
\newcommand{\feyn}{\mathtt{feyn}} 
\newcommand{\FT}{\mathfrak{F}} 

\newcommand{\h}{\mathpzc{h}}
\newcommand{\Halg}{\mathfrak{H}} 

\newcommand{\homog}{\mathrm{homog}}

\newcommand{\id}{\mathrm{id}}
\newcommand{\im}{\mathrm{im}}
\renewcommand{\Im}{\mathrm{Im}}

\newcommand{\Lrt}{\mathcal{L}}
\newcommand{\lc}{\mathcal{V}} 
\newcommand{\Lin}{\mathrm{Lin}}
\newcommand{\linc}{\mathtt{C}} 
\newcommand{\loc}{\mathrm{loc}}
\newcommand{\Log}{\mathrm{Log}}

\newcommand{\M}{\mathbb{M}}
\newcommand{\md}{\mathrm{md}} 

\newcommand{\MS}{\mathrm{MS}}

\newcommand{\oast}{\circledast}
\newcommand{\obs}{\obackslash}
\newcommand{\oc}{\mathpzc{o}}
\newcommand{\ocomp}{\copyright}
\newcommand{\openone}{\mathbbm{1}}

\newcommand{\Part}{\mathrm{Part}}
\newcommand{\Perm}{\mathrm{Perm}}
\DeclareMathOperator*{\pos}{pos}
\newcommand{\pp}{\mathrm{pp}} 
\newcommand{\prep}{\mathrm{prep}}
\newcommand{\Ptn}{\mathcal{P}} 
\newcommand{\Pto}{\mathbb{P}} 

\newcommand{\R}{\mathcal{R}} 

\newcommand{\ran}{\mathrm{ran}}
\renewcommand{\Re}{\mathrm{Re}}

\newcommand{\ren}{\mathrm{ren}}
\newcommand{\ret}{\mathrm{ret}}
\newcommand{\Res}{\mathrm{Res}}
\newcommand{\rp}{\mathrm{rp}} 

\newcommand{\sd}{\mathrm{sd}}
\newcommand{\Sm}{\boldsymbol{\mathcal{S}}} 
\newcommand{\source}{\mathfrak{s}}

\newcommand{\supp}{\mathrm{supp}} 
\newcommand{\singsupp}{\mathrm{singsupp}}
\newcommand{\Sym}{\mathrm{Sym}}

\newcommand{\target}{\mathfrak{t}}

\newcommand{\tube}{\mathbb{T}}
\newcommand{\Time}{\boldsymbol{\mathcal{T}}}
\newcommand{\dT}{\cdot_{\Time}}

\newcommand{\V}{\mathscr{V}}

\newcommand{\WF}{\mathrm{WF}} 

\newcommand{\xyC}[1]{%
\xydef@\xymatrixcolsep@{#1}
} 

\newcommand{\xyR}[1]{%
\xydef@\xymatrixrowsep@{#1}
} 

\@ifundefined{showcaptionsetup}{}{%
 \PassOptionsToPackage{caption=false}{subfig}}
\usepackage{subfig}
\makeatother

\bibliography{Literatur_Physik}

\begin{document}
\global\long\def\fps#1#2{#1\![[#2]]}

\global\long\def\poisson#1#2{\left\lfloor #1,#2\right\rceil }

\global\long\def\bld#1{\boldsymbol{#1}}

\global\long\def\oprod#1{\sideset{}{^{\geq}}\prod_{#1}}

\title{Dimensional Regularization in Position Space and~a~Forest~Formula~for
Regularized Epstein-Glaser~Renormalization\emph{\vspace{3cm}}\\
\emph{Dissertation}\\
\emph{zur Erlangung des Doktorgrades}\\
\emph{des Department Physik}\\
\emph{der Universität Hamburg}\vspace{5cm}}

\author{vorgelegt von\\
Kai Johannes Keller\\
aus Landau in der Pfalz\vspace{1cm}}

\author{Hamburg\\
2010}

\maketitle
\vfill

\noindent \begin{flushleft}
\begin{minipage}[t]{1\columnwidth}%
\begin{tabular}{>{\raggedright}b{7cm}>{\raggedright}p{7cm}}
Gutachter der Dissertation: & Prof. Dr. K. Fredenhagen

Prof. Dr. J. M. Gracia-Bondía\tabularnewline
\noalign{\vskip3mm}
Gutachter der Disputation: & Prof. Dr. K. Fredenhagen

Prof. Dr. B. Kniehl\tabularnewline
\noalign{\vskip3mm}
Datum der Disputation: & Mittwoch, 7. April 2010\tabularnewline
\noalign{\vskip3mm}
Vorsitzender des Prüfungsausschusses: & Prof. Dr. G. Sigl \tabularnewline
\noalign{\vskip3mm}
Vorsitzender des Promotionsausschusses: & Prof. Dr. J. Bartels\tabularnewline
\noalign{\vskip3mm}
Dekan der Fakultät für Mathematik, Informatik und Naturwissenschaften: & Prof. Dr. H. Graener\tabularnewline
\noalign{\vskip3mm}
\end{tabular}%
\end{minipage}
\par\end{flushleft}

$\phantom{Hallo,\, ich\, bin\, da,\, aber\, Du\, siehst\, mich\, nicht!}$

\vspace{3cm}\thispagestyle{empty}

\framebox{\begin{minipage}[t]{1\columnwidth}%
\begin{center}
\textbf{\textsc{\Large Zusammenfassung}}
\par\end{center}{\Large \par}

In der vorliegenden Arbeit werden die Methoden von Dimensionaler Regularisierung
(DimReg) und Minimaler Subtraktion (MS) konsistent im Minkowsiki-Ortsraum
formuliert, und in den Rahmen der perturbativen Algebraischen Quantenfeldtheorie
(pAQFT) implementiert. Die entwickelten Kozepte werden benutzt, um
die Rekursion von Epstein und Glaser für die Konstruktion zeitgeordneter
Produkte in allen Ordnungen der kausalen Störungstheorie zu lösen.
Es wird eine geschlossene Lösung in Form einer Waldformel à la Zimmermann
angegeben. Eine Verbindung zu dem alternativen Zugang zur Renormierungstheorie
über Hopf-Algebren wird hergestellt. %
\end{minipage}}

\vspace{3cm}

\framebox{\begin{minipage}[t]{1\columnwidth}%
\begin{center}
\textbf{\textsc{\Large Abstract}}
\par\end{center}{\Large \par}

The present work contains a consistent formulation of the methods
of dimensional regularization (DimReg) and minimal subtraction (MS)
in Minkowski position space. The methods are implemented into the
framework of perturbative Algebraic Quantum Field Theory (pAQFT).
The developed methods are used to solve the Epstein-Glaser recursion
for the construction of time-ordered products in all orders of causal
perturbation theory. A solution is given in terms of a forest formula
in the sense of Zimmermann. A relation to the alternative approach
to renormalization theory using Hopf algebras is established.%
\end{minipage}}

\begin{spacing}{1.2}

\tableofcontents{}

\end{spacing}

\chapter*{Introduction}

\begin{center}
\begin{minipage}[t][1.5\totalheight]{0.5\columnwidth}%
\begin{flushleft}
\emph{Yet knowing how way leads on to way,}
\par\end{flushleft}

\begin{flushleft}
\emph{I doubted if I should ever come back. }
\par\end{flushleft}

\begin{flushright}
\emph{Robert Frost: The Road Not Taken}
\par\end{flushright}%
\end{minipage}
\par\end{center}

In 1957 Bogoliubov and Parasiuk introduced an inductive method for
the solution of the renormalization problem of perturbative Quantum
Field Theory, that is the problem of constructing the terms in the
perturbative expansion of the scattering matrix ($\Sm$-matrix) \citep{Bogoliubow1957,BogoliubovShirkov1959}.
It was later shown by Hepp {}``that the subtraction rules of Bogoliubov
an Parasiuk lead to well-defined renormalized Green's distributions''
\citep{Hepp1966}. From this common origin the method of Bogoliubov,
Parasiuk, and Hepp (BPH) evolved mainly along two different roads.
The \emph{BPHZ method} induced by the solution of the BPH recursion
in terms of Zimmermann's forest formula in momentum space on the one
hand side \citep{Zimmermann1969}, and \emph{causal perturbation theory}
induced by Epstein and Glaser's rigorous solution of the renormalization
problem in position space on the other \citep{Epstein1973}. Both
methods are rigorous incarnations of BPH, however, they have played
quite different roles in the development of perturbative Quantum Field
Theory (pQFT).

Causal perturbation theory has proven to be superior to the so-called
{}``standard approach'' to renormalization in momentum space when
it comes to more conceptual questions of perturbative renormalization,
and it is widely accepted as the landmark with which one has to test
new approaches to renormalization (see, e.g., \citep{Falk2009}).
What is more, Epstein-Glaser renormalization is the only renormalization
method which has been successfully formulated on more general, physical
backgrounds \citep{Brunetti2000}. Induced by the development of Quantum
Field Theory on curved spacetimes \citep{Radzikowski1996,Brunetti1996,Brunetti2000}%
\footnote{See also \citep{QFTonCST2009} for a selfcontained treatment of the
topic.%
} and along with the successful formulation of the renormalization
group in generic, globally hyperbolic spacetimes by Hollands and Wald
\citep{HollandsWald2001,HollandsWald2002,Hollands2003}, Brunetti,
Dütsch, and Fredenhagen started a program on the structural analysis
of perturbative Quantum Field Theory in the algebraic approach \citep{Dutsch1999,Dutsch2000,Dutsch2001,DuetschFredenhagen2003,Hollands2004,DuetschFredenhagen2004,Brunetti2004,Duetsch2005,Brunetti2007a}.
One of the main results of this program was the precise formulation
and proof of what Popineau and Stora called the \emph{Main Theorem
of Perturbative Renormalization} \citep{StoraPopineau1982}. This
is the fact that the definition of the $\Sm$-matrix of pQFT involves
a freedom described by the Stückelberg-Petermann renormalization group
\citep{Stueckelberg1953}. A milestone of the algebraic approach to
perturbative Quantum Field Theory was reached with the formulation
of \emph{perturbative Algebraic Quantum Field} \emph{Theory (pAQFT)},
which has been shown to give a common basis to the different other
incarnations of the renormalization group in literature \citep{Brunetti2009}.

Despite these deep results in perturbative renormalization, causal
perturbation theory has its weak point when it comes to concrete predictions,
say, for scattering amplitudes in collision processes of elementary
particles. Here the standard approach to pQFT in momentum space and
in particular the method of dimensional regularization (DimReg) and
minimal subtraction (MS) \citep{BolliniGiambiagi1972a,tHooftVeltman1972}
combined with Zimmermann's forest formula has proven to be efficient
in its application and to produce predictions which are in astonishing
accordance with measurements in accelerator experiments. The proof
that DimReg+MS is compatible with the combinatorics described by BPHZ
was given in \citep{BreitenlohnerMaison1977a,BreitenlohnerMaison1977b,BreitenlohnerMaison1977c},
and in particular the compatibility with gauge theories has contributed
to the success of dimensional regularization in favor of other analytic
renormalization techniques in elementary particles phenomenology \citep{BecchiRouetStora1975}.
A seemingly forgotten part of this road is that Zimmermann realized
in 1970 and proved in 1975 that the additional subtractions found
in his formula in comparison to BPH do not contribute in the limit
where the regularization is removed \citep{Zimmermann1970,Zimmermann1975}.%
\footnote{I want to thank José Gracia-Bondía for directing my attention to these
references. See also reference \citep{Falk2009} in this respect.%
}

In recent years great interest in the mathematical community for the
renormalization method of DimReg+MS combined with BPHZ has been triggered
by Kreimer's discovery of a Hopf algebra structure underlying the
BPHZ renormalization method \citep{Kreimer1998}. Connes and Kreimer
pointed out various relations of this discovery to fields of research
in pure mathematics, such as Number Theory and Noncommutative Geometry
\citep{Connes1998,Connes2000,Connes2001}. Consequently, by now the
field has grown to a research area of considerable extent between
the poles of more mathematically oriented research in Algebraic Geometry
and Number Theory \citep{BlochEsnaultKreimer2006,ConnesMarcolli2007}
and applications in the computation of higher order contributions
to the perturbative expansion of the $\Sm$-matrix \citep{Weinzierl2006,BognerWeinzierl2009,BognerWeinzierl2010}.
Shortly after Kreimer's discovery Gracia-Bondía and Lazzarini observed,
and Pinter showed that also Epstein-Glaser renormalization exhibits
a Hopf algebra structure of the Connes-Kreimer type \citep{Gracia-Bondia2000,PinterPHD2000,Pinter2000b}.
Thus the Hopf algebra structure was observed to be a remnant of the
common origin of the two roads in perturbative Quantum Field Theory
briefly outlined above. It was shown that the Hopf algebra structure
of BPHZ renormalization is invariant under certain partial summations
of graphs in the perturbative expansion \citep{BrouderFrabetti2000,BrouderFrabetti2001,Frabetti2007,Suijlekom2007c}.
Furthermore, many results on the occurrence of different Hopf algebras
in perturbative QFT have been obtained in recent years and it was
found that the Faà~di~Bruno Hopf algebra plays a distinguished role
among them \citep{Figueroa2005}.

The aim of this thesis is to {}``combine the good parts'' of both
roads to renormalization in perturbative Quantum Field Theory. That
is, to incorporate the effective methods of dimensional regularization
and minimal subtraction combined with Zimmermann's forest formula
in the conceptually clear setting of causal perturbation theory. After
a brief introduction to the theory of extension of distributions,
which is the main tool in modern formulations of causal perturbation
theory, in the first chapter, I will describe the setting of perturbative
Algebraic Quantum Field Theory in the special case of flat Minkowski
spacetime in Chapter~\ref{cha:Setting-pAQFT}. Following the arguments
in an appendix of the pAQFT article (loc.~cit.) I will show in Chapter~\ref{cha:DimRegHadamard}
how one can use a modification of the Bessel parameter in a representation
of the Wightman two point function in Minkowski space to construct
a dimensionally regularized analytic (Hadamard-) two point function
in flat spacetime which depends smoothly on the mass parameter $m^{2}$.
As shown by Hollands, smooth mass dependence is a suitable requirement
for a covariant treatment of renormalization \citep{Hollands2004}.
In Chapter~\ref{cha:DimReg-PositionSpace} I will then construct
the dimensionally regularized position space amplitude to any graph
$\Gamma$ in scalar quantum field theory as a distribution in $\D'(\M^{\left|V(\Gamma)\right|})$
($\M$ denotes Minkowski spacetime and $\left|V(\Gamma)\right|$ the
number of vertices of $\Gamma$). I have to remark here that Bollini
and Giambiagi already gave a formulation of dimensional regularization
in position space by Fourier transforming the regularized momentum
space amplitude to position space and found a modification in the
Bessel parameter of the corresponding two point function \citep{BolliniGiambiagi1996}.
Conversely, a Fourier transformation of the amplitudes constructed
in this work (which are different from the ones found by Bollini and
Giambiagi) to momentum space is not possible in general, since the
condition of smoothness in $m^{2}$ will select a propagator which
is not in Schwartz space. As a result of this chapter, I define the
position space dimensionally regularized $\Sm$-matrix, $\Sm_{\mu,\zeta}$,
which fulfills the conditions of the main theorem of perturbative
renormalization as proven in \citep{DuetschFredenhagen2004,Brunetti2009}.
In Chapter~\ref{cha:Minimal-Subtraction} I will show how minimal
subtraction can be applied to the dimensionally regularized position
space amplitudes in a graph by graph manner, and will test the method
by reproducing the result of Zimmermann that so-called {}``pure BPHZ
subgraphs'' do not contribute to the forest formula in the limit
where the regularization is removed \citep{Zimmermann1975}. The last
chapter of my thesis will use $\Sm_{\mu,\zeta}$ as an example for
an analytically regularized $\Sm$-matrix, but does not depend on
the way it was constructed. In this sense, the results of the last
chapter are independent of the formulation of dimensional regularization
in position space summarized above, and consequently they can be applied
in a much wider range. I will show in Chapter~\ref{cha:ForestFormula}
of the present thesis that a forest formula for regularized Epstein-Glaser
renormalization can be derived directly from the main theorem of perturbative
renormalization. I will give the formula and prove locality of the
MS counterterms. Furthermore I will show that the Hopf algebra structure
observed in perturbative renormalization theory can be understood
as a direct consequence of the main theorem. However, in contrast
to the Connes-Kreimer theory of renormalization, the Feynman rules
will emerge naturally from the construction, and it will be shown
that the commutative Hopf algebra of graphs introduced by Connes and
Kreimer is not enough for an algebraic construction of the counterterms\linebreak[3]
found in pAQFT. Another difference to the original Hopf algebra approach
is that the Hopf algebraic structure found in the construction of
pAQFT counterterms will correspond to sums of graphs rather than individual
ones, however, the correspondence to the Hopf algebra of graphs is
established by linearity; in accordance with the results of \citep{BrouderFrabetti2000,BrouderFrabetti2001,Frabetti2007,Suijlekom2007c}.

In order to be precise and prevent confusion, I want to remark that
I mean by {}``regularization'' in this thesis always parametric
regularization, i.e., the introduction or modification of a parameter,
which makes the extensions of the regularized distributions unique.
{}``Dimensional regularization'' is one example. {}``Renormalization'',
on the other hand, I want to use as a synonym for the extension of
distributions, as it is widely used in the terms {}``Epstein-Glaser
renormalization'' or {}``BPHZ renormalization''. Observe, however,
that the extension of its time-ordered products is only a necessary
but by no means sufficient prerequisite for a quantum field theory
to be renormalizable by power counting. That is to say, we are not
concerned with the number of counterterms that are to be introduced
at each order of perturbation theory, but only with the fact that
this number is finite. Neither will we treat the question whether
the counterterms can be absorbed in a redefinition of the parameters
in a Lagrangian of the theory.

A last remark I want to make is that the extension of the time-ordered
products to the total diagonal, which will be treated in some detail
below and corresponds to the elimination of ultraviolet (UV) divergences
in the standard approach, suffices for the perturbative definition
of the quantum field theory under investigation in the algebraic adiabatic
limit. This was shown in \citep{Brunetti2000,Dutsch2000,Hollands2003}.
The algebraic adiabatic limit is a way to remove the explicit spacetime
dependence of the interaction without introducing so-called infrared
(IR) divergences. IR divergences typically appear in the standard
approach if one removes the cutoff at small momenta (or large distances)
in theories with long range interactions. Such divergences appear
also in the causal approach of Epstein and Glaser in the strong and
in the weak adiabatic limit. Neither strong nor weak adiabatic limit
will be treated in this thesis, and as much as the algebraic adiabatic
limit is concerned I cannot add anything new to the discussion in
\citep[Chap.~6]{Brunetti2009}.

\chapter{Mathematical~Preliminaries: Extension~of~Distributions}

The main tool in renormalization in position space is the extension
of distributions, thus we want to summarize here the basic definitions
and main results of this part of distribution theory. We will first
give the general result on the existence of extensions of distributions
with the same scaling degree and will indicate how such extensions
are constructed. In the second section we will review the special
case of homogeneous distributions; homogeneity being a suitable condition
for the existence of a unique extension. We will generalize the uniqueness
result on homogeneous extensions to the case of heterogeneous distributions
in the third section. The fourth section will be devoted to the definition
of an (analytic) regularization of a distribution. We will derive
some direct consequences to be used in later chapters. A general reference
for this chapter, and a guidance for mathematical questions throughout
the thesis is the book of Hörmander \citep{Hoermander2003}.

We generally use the notation of Laurent Schwartz for the function
spaces, $\E(\mathbb{R}^{d})=C^{\infty}(\mathbb{R}^{d})$ of smooth
functions, and $\D(\mathbb{R}^{d})=C_{0}^{\infty}(\mathbb{R}^{d})$
of smooth functions with compact support (test functions) with their
respective standard topologies; and $\E'$, respectively $\D'$ for
their dual spaces.

\section{Extensions and Steinmann Scaling Degree}
\begin{defn}
[Extension] Let $u\in\D'(\mathbb{R}^{d}\backslash\left\{ 0\right\} )$
be a distribution defined for all test functions supported in the
complement of the origin. We call $\dot{u}\in\D'(\mathbb{R}^{d})$
an \emph{extension} of $u$, if\begin{equation}
\forall f\in\D(\mathbb{R}^{d}\backslash\left\{ 0\right\} ):\quad\dot{u}(f)=u(f)\,.\label{eq:ExtensionOfDistribution}\end{equation}

\end{defn}
Not every distribution $u\in\D(\mathbb{R}^{d}\backslash\left\{ 0\right\} )$
has an extension, and if there is one it is not unique. However, by
(\ref{eq:ExtensionOfDistribution}) two extensions of $u$ differ
by a distribution supported at the origin. By \citep[Thm.~2.3.4]{Hoermander2003}
any distribution supported at the origin is a polynomial in the derivatives
of Dirac's $\delta$-distribution. We call such distributions \emph{local}
and denote the space of all local distributions by $\E_{\mathrm{Dirac}}'$.
One way to restrict the freedom in the extension procedure is to require
that the extension should have the same scaling degree, cf.~\citep{Steinmann1971,Brunetti2000}.
\begin{defn}
[{Steinmann Scaling Degree}] Let\begin{equation}
\begin{array}{rccl}
\Lambda: & \mathbb{R}_{+}\times\D & \rightarrow & \D\\
 & \left(\rho,\phi\right) & \mapsto & \phi^{\rho}:=\rho^{-d}\phi(\rho^{-1}\cdot)\end{array}\label{eq:ScalingOfFunctions}\end{equation}
be the action of the positive reals on test functions in $\D\in\left\{ \D(\mathbb{R}^{d}),\D(\mathbb{R}^{d}\backslash\left\{ 0\right\} )\right\} $.
This induces, via the pullback, the action on distributions. For $u\in\D'$
we define\[
\forall\phi\in\D:\quad u_{\rho}(\phi):=u(\phi^{\rho})\,.\]
The \emph{scaling degree} $\sd(u)$ of a distribution $u$ with respect
to the origin is defined to be\[
\sd(u):=\inf\left\{ \omega\in\mathbb{R}:\lim_{\rho\rightarrow0^{+}}\rho^{\omega}u_{\rho}=0\in\D'\right\} \,.\]
\end{defn}
\begin{example}
Dirac's $\delta$-distribution has scaling degree $\sd(\delta)=d$,
since\[
\lim_{\rho\rightarrow0^{+}}\rho^{\omega}\left\langle \delta,\phi^{\rho}\right\rangle =\lim_{\rho\rightarrow0^{+}}\rho^{\omega-d}\phi(0)\,.\]

\end{example}
Furthermore, a similar argument shows that any smooth function has
scaling degree smaller than or equal to zero. The basic properties
of the scaling degree are summarized in the following
\begin{lem}
[{cf. \citep[Lem.~5.1]{Brunetti2000}}]\label{lem:ScalingDegreeProperties}Let
$u\in\D'(\mathbb{R}^{d})$, $v\in\D'(\mathbb{R}^{k})$ and let $\alpha\in\mathbb{N}^{d}$
be a multiindex, then
\begin{enumerate}
\item [(a)]$\sd(\partial^{\alpha}u)\leq\sd(u)+\left|\alpha\right|$
\item [(b)]$\sd(x^{\alpha}u)\leq\sd(u)-\left|\alpha\right|$
\item [(c)]$\forall f\in\E(\mathbb{R}^{d}):\sd(fu)\leq\sd(u)$
\item [(d)]$\sd(u\otimes v)=\sd(u)+\sd(v)$
\end{enumerate}
\end{lem}
For later reference we also define the related concept of degree of
divergence of a distribution.
\begin{defn}
[Degree of Divergence]\label{def:DegreeOfDivergence}Let $u\in\D'\in\left\{ \D'(\mathbb{R}^{d}),\D'(\mathbb{R}^{d}\backslash\left\{ 0\right\} )\right\} $,
then we define the degree of divergence of $u$\[
\div(u):=\sd(u)-d\,.\]

\end{defn}
Observe that the scaling degree of a product of distributions $u,v\in\D'(\mathbb{R}^{d})$,
if it exists, is given by the scaling degree of the tensor product
$u\otimes v\in\D'(\mathbb{R}^{2d})$,\[
\sd(uv)=\sd(u\otimes v)\,,\]
whereas the degree of divergence of the product is greater\[
\div(uv)=\div(u\otimes v)+d\,.\]
Although this observation follows directly from Lemma~\ref{lem:ScalingDegreeProperties}(d)
the following theorem shows that it reflects the freedom involved
in the definition of the product of distributions. Recall that, if
it exists, the pointwise product $uv\in\D'(\mathbb{R}^{d})$ is defined
as the pullback of $u\otimes v\in\D'(\mathbb{R}^{2d})$ via the diagonal
map\[
\diag:\mathbb{R}^{d}\ni x\mapsto\left(x,x\right)\in\mathbb{R}^{2d}\,,\]
cf.~\citep[Thm.~8.2.10]{Hoermander2003}). One often encounters the
situation, that the pullback $uv=\diag^{*}(u\otimes v)$ defines the
product only in the complement of the origin, even if $u,v\in\D'(\mathbb{R}^{d})$.
This is the case, e.g., in perturbative renormalization theory and
hence one is naturally lead to the problem of finding extensions of
certain (products of) distributions \citep{Bogoliubow1957}. This
lead Epstein and Glaser to their constructive extension procedure
by {}``distribution splitting'' \citep{Epstein1973}. The mathematically
quite involved inductive procedure carried out by Epstein and Glaser
may be called the first rigorous construction of extensions of distributions
in position space. It was Steinmann who introduced the concept of
scaling degree in the discussion related to the construction of extensions
of certain distributions \citep{Steinmann1971}. There are later works
contributing to this topic, such as \citep{Estrada1998a}, and it
is treated by now in several text books \citep{Hoermander2003}, however,
the most general result known to the author is the theorem to be cited
below. It was to my best knowledge first proven in \citep[Thms.~5.2~\&~5.3]{Brunetti2000}.
\begin{thm}
[Extension of Distributions]\label{thm:ExtensionScalingDegree} Let
$u\in\D'(\mathbb{R}^{d}\backslash\left\{ 0\right\} )$ have scaling
degree $\sd(u)$ with respect to the origin. Let
\begin{itemize}
\item $\sd(u)<d$. Then there exists a unique extension $\dot{u}\in\D'(\mathbb{R}^{d})$
of $u$, which has the same scaling degree, $\sd(\dot{u})=\sd(u)$.
\item $d\leq\sd(u)<\infty$. Then there exist several extensions $\dot{u}\in\D'(\mathbb{R}^{d})$
with $\sd(\dot{u})=\sd(u)$. They are uniquely defined by their values
on a finite set of test functions.
\end{itemize}
\end{thm}
For completeness we remark that $u$ has no extension, if $\sd(u)=\infty$,
the distribution $f\mapsto\int dx\, e^{\frac{1}{x}}f(x)$ is a standard
example of this case. Furthermore we remark that the scaling degree
of the extension $\dot{u}$ cannot be smaller than that of $u$. Thus
the condition that $\dot{u}$ should have the same scaling degree
as $u$ is a condition of minimal scaling degree or {}``maximal smoothness''
at the origin.
\begin{proof}
[{Sketch of Proof of Thm. \ref{thm:ExtensionScalingDegree}}]Let
first $\sd(u)<d$. Uniqueness follows immediately from the fact that
two extensions $\dot{u},\ddot{u}\in\D'(\mathbb{R}^{d})$ differ by
a polynomial $P(\delta)$ in derivatives of Dirac's $\delta$-distribution,
which has scaling degree $\sd(P(\delta))\geq d$, cf.~Lemma~\ref{lem:ScalingDegreeProperties}(a).

Let $\vartheta\in\E(\mathbb{R}^{d})$, $0\leq\vartheta\leq1$, such
that $\vartheta(x)=0$ for $\left|x\right|<1$ and $\vartheta(x)=1$
for $\left|x\right|\geq2$ and set $\vartheta_{\rho}(x)=\vartheta(\rho x)$,
then\begin{equation}
\dot{u}:=\lim_{\rho\rightarrow\infty}\vartheta_{\rho}u\label{eq:ExtensionScalingDegreeUnique}\end{equation}
converges in $\D'(\mathbb{R}^{d})$, i.e.,\[
\forall g\in\D(\mathbb{R}^{d}):\quad\lim_{\rho\rightarrow\infty}\left\langle u,\vartheta_{\rho}g\right\rangle \in\mathbb{C}\,,\]
and $\dot{u}$ defines an extension of $u$ with the same scaling
degree, cf.~\citep{Brunetti2000}.

Now regard the case $d\leq\sd(u)<\infty$. We define the space $\D_{\lambda}(\mathbb{R}^{d})$
of functions, which vanish up to order $\lambda>0$ at the origin,\begin{equation}
\D_{\lambda}(\mathbb{R}^{d})=\left\{ f\in\D(\mathbb{R}^{d})|\quad\forall\left|\alpha\right|\leq\lambda:\,\left(\partial^{\alpha}f\right)(0)=0\right\} .\label{eq:FunctionVanishingAtOriginToOrder}\end{equation}
Then $u$ is uniquely defined on functions, which vanish up to the
order given by the degree of divergence of $u$, i.e., $u$ has a
unique extension $\tilde{u}\in\D_{\lambda}'(\mathbb{R}^{d})$, $\lambda=\div(u)$,
with the same scaling degree. Any function $f\in\D_{\lambda}(\mathbb{R}^{d})$
can be written in the form\begin{equation}
f(x)=\sum_{\left|\alpha\right|=\left\lfloor \lambda\right\rfloor +1}x^{\alpha}g_{\alpha}(x)\,,\quad g_{\alpha}\in\D(\mathbb{R}^{d})\,,\label{eq:ExtensionSDUniqueRepresentationF}\end{equation}
where $\left\lfloor \lambda\right\rfloor $ denotes Gau{\ss}'s floor
function, i.e., the largest integer smaller or equal $\lambda$. We
define $\tilde{u}$ by\[
\left\langle \tilde{u},f\right\rangle :=\sum_{\left|\alpha\right|=\left\lfloor \lambda\right\rfloor +1}\left\langle \left(x^{\alpha}u\right)^{\cdot},g_{\alpha}\right\rangle \,,\]
where the extensions $\left(x^{\alpha}u\right)^{\cdot}$ on the right
hand side are unique by Lemma~\ref{lem:ScalingDegreeProperties}(b)
and the first part of the theorem. They can be computed as weak limits
of the form (\ref{eq:ExtensionScalingDegreeUnique}), and thus the
limit exist for each term separately. We have \[
\sum_{\left|\alpha\right|=\left\lfloor \lambda\right\rfloor +1}\lim_{\rho\rightarrow\infty}\left\langle x^{\alpha}u,\vartheta_{\rho}g_{\alpha}\right\rangle =\lim_{\rho\rightarrow\infty}\left\langle u,\vartheta_{\rho}f\right\rangle \,,\]
which shows that the extension $\tilde{u}$ does not depend on the
chosen representation (\ref{eq:ExtensionSDUniqueRepresentationF})
of $f\in\D_{\lambda}(\mathbb{R}^{d})$, cf.~\citep{DuetschFredenhagen2004}.

Regard now the ambiguity left in the extension to $\D=\D(\mathbb{R}^{d})$.
$\D_{\lambda}\subset\D$ is a closed subspace, hence there are projections
$W:\D\rightarrow\D_{\lambda}$, one for each choice of the complement
$\C$,\[
\D=\D_{\lambda}\oplus\C\,.\]
where $\D_{\lambda}\equiv\ran(W)$ and $\C\equiv\ran(1-W)$. This
split of $\D$ induces a split of the dual space $\D'$ according
to the following diagram, cf. \citep[Lem.~VI.3.3]{DunfordSchwartz1967},\\
\begin{minipage}[c][1.1\totalheight]{0.45\columnwidth}%
\begin{equation}
\xymatrix{\C & \ar[l]_{1-W}\D\ar@{->}[d]|{\vphantom{\big|}\mathrm{duality}}\ar[r]^{W} & \D_{\lambda}\\
\D_{\lambda}^{\perp} & \ar[l]_{1-W'}\D'\ar[r]^{W'} & \C^{\perp}}
\label{eq:SplitDualSpace}\end{equation}
\end{minipage}%
\begin{minipage}[c]{0.45\columnwidth}%
\[
\D'=\D_{\lambda}^{\perp}\oplus\C^{\perp}\]
\end{minipage}\\
where \[
{\displaystyle \D_{\lambda}^{\perp}:=\left\{ u\in\D'|\quad\forall f\in\D_{\lambda}:\left\langle u,f\right\rangle =0\right\} ,}\]
and \[
{\displaystyle \C^{\perp}:=\left\{ v\in\D'|\quad\forall g\in\C:\left\langle v,g\right\rangle =0\right\} .}\]
The dual projections are induced by\[
\forall u\in\D',\,\forall f\in\D:\quad\left\langle W'u,f\right\rangle :=\left\langle u,Wf\right\rangle \,.\]
Any dual basis $\left\{ w_{\alpha}\in\D:\left\langle \delta^{\left(\beta\right)},w_{\alpha}\right\rangle =\delta_{\alpha}^{\beta}\right\} $
of the basis $\left\{ \delta^{\left(\alpha\right)}:\left|\alpha\right|\leq\lambda\right\} $
of $\D_{\lambda}^{\perp}$ spans a complement $\C=\D\ominus\D_{\lambda}$
and thus defines a projection $W$. As a consequence we have the following
characterization of projections $W:\D\rightarrow\D_{\lambda}$.
\begin{lem}
[{cf.~\citep[Lem.~B.1]{DuetschFredenhagen2004}}]\label{lem:W-Projection-Functions}
There is a one-to-one correspondence between families of functions\begin{equation}
\left\{ w_{\alpha}\in\D\,|\quad\forall\left|\beta\right|\leq\lambda:\left(\partial^{\beta}w_{\alpha}\right)(0)=\delta_{\alpha}^{\beta},\,\left|\alpha\right|\leq\lambda\right\} \label{eq:LemmaProjectionWFunctions}\end{equation}
and projections $W:\D\rightarrow\D_{\lambda}$. The set (\ref{eq:LemmaProjectionWFunctions})
defines a projection $W$ by\[
Wf:=f-\sum_{\left|\alpha\right|\leq\lambda}f^{\left(\alpha\right)}(0)\, w_{\alpha}\,.\]
Conversely a set of functions of the form (\ref{eq:LemmaProjectionWFunctions})
is given by any basis of $\C\equiv\ran(1-W)$ dual to the basis $\left\{ \delta^{\left(\alpha\right)}:\left|\alpha\right|\leq\lambda\right\} $
of $\D_{\lambda}^{\perp}\subset\D'$.
\end{lem}
We reach the conclusion that for each projection $W:\D\rightarrow\D_{\lambda}$
there is a unique extension $W'\tilde{u}\in\C^{\perp}$ of $u\in\D'(\mathbb{R}^{d}\backslash\left\{ 0\right\} )$.
The most general extension $\dot{u}$ of $u$, fulfilling the assumptions
of the theorem, can be read off from the split of $\D'$ (\ref{eq:SplitDualSpace}),\begin{equation}
\dot{u}=W'\tilde{u}+\sum_{\left|\alpha\right|\leq\div(u)}C_{\alpha}\delta^{\left(\alpha\right)}\,,\label{eq:ExtensionWMostGeneral}\end{equation}
where $C_{\alpha}\in\mathbb{C}$ are free constants. Observe, however,
that a particular extension is fixed by a choice of its values on
$\C$, namely $\left\langle \dot{u},w_{\alpha}\right\rangle =C_{\alpha}$.
\end{proof}
Although (\ref{eq:ExtensionWMostGeneral}) gives the most general
extension of $u$ with the same scaling degree, it is important to
note that the second term in (\ref{eq:ExtensionWMostGeneral}) does
not introduce an additional freedom, but only reflects the freedom
in the choice of the projection $W$.
\begin{lem}
[{cf. \citep[Lem.~B.2]{DuetschFredenhagen2004}}]\label{lem:W-Extension}Let
$\dot{u}\in\D'(\mathbb{R}^{d})$ be an extension of $u\in\D'(\mathbb{R}^{d}\backslash\left\{ 0\right\} )$
with $\div(\dot{u})=\div(u)=\lambda$. Then there exists a complementary
space $\C$ of $\D_{\lambda}$ in $\D$ such that\[
\dot{u}\big|_{\C}=0\,,\]
i.e., $C_{\alpha}=0$ in (\ref{eq:ExtensionWMostGeneral}).
\end{lem}
That is, any extension $\dot{u}$ of $u$ can be written as a \emph{$W$-extension},
$\dot{u}=W'\tilde{u}$ with a suitably chosen projection $W$.

Despite its wide applicability ($\sd(u)$ is defined for any $u\in\D'$),
the scaling degree is often too rough a tool for describing the behavior
of distributions at the origin. A refinement of the notion of scaling
degree is the degree of homogeneity defined only for homogeneous distributions.
We will see in the next section that this refinement leads to a stronger
result regarding the uniqueness of extensions.

\section{Homogeneous Distributions}
\begin{defn}
[{Homogeneous Distribution, cf. \citep[Def. 3.2.2]{Hoermander2003}}]\label{def:HomogeneousDistribution}A
distribution $u\in\D'\in\left\{ \D'(\mathbb{R}^{d}),\D'(\mathbb{R}^{d}\backslash\left\{ 0\right\} )\right\} $
is called \emph{homogeneous} of degree $D\in\mathbb{C}$, if\begin{equation}
\forall\rho>0:\quad\left\langle u,\phi\right\rangle =\rho^{-D}\left\langle u_{\rho},\phi\right\rangle \qquad\forall\phi\in\D\,,\label{eq:DefinitionHomogeneous}\end{equation}
where $\left\langle u_{\rho},\phi\right\rangle :=\left\langle u,\phi^{\rho}\right\rangle $
with $\phi^{\rho}(x):=\rho^{-d}\phi(\rho^{-1}x)$ as in (\ref{eq:ScalingOfFunctions}).
We will sometimes write $D=\homog(u)$ for the homogeneity degree
of a distribution $u$.\end{defn}
\begin{rem}
\label{rem:ScalingDegreeHomogeneity}Observe that a distribution $u$,
homogeneous of degree $D\in\mathbb{C}$, has scaling degree $\sd(u)=-\Re(D)$,\begin{align*}
\sd(u) & =\inf\bigg\{\omega\in\mathbb{R}:\,\lim_{\rho\rightarrow0^{+}}\rho^{\omega}\left\langle u,\phi^{\rho}\right\rangle =\lim_{\rho\rightarrow0^{+}}\rho^{\omega+\Re(D)}\rho^{i\,\Im(D)}=0\bigg\}=-\Re(D)\,.\end{align*}
In this sense the homogeneity degree {}``$\homog$'' is a refinement
of the scaling degree {}``$\sd$'', and Consequently we will get
a stronger statement for the extendability of homogeneous distributions
in Theorem~\ref{thm:HomogeneousExtension} below. However, before
citing this result, let us regard an alternative characterization
of homogeneity.
\end{rem}

\begin{thm}
[Euler]\label{thm:Euler}A distribution ${u\in\D'\in\left\{ \D'(\mathbb{R}^{d}),\D'(\mathbb{R}^{d}\backslash\left\{ 0\right\} )\right\} }$
is homogeneous of degree $D\in\mathbb{C}$ if and only if\begin{equation}
\left\langle \left(x\cdot\partial_{x}-D\right)u,\phi\right\rangle =0\qquad\forall\phi\in\D\,,\label{eq:HomogeneousDistributionEulerThm}\end{equation}
where $x\cdot\partial_{x}=\sum_{i=1}^{d}x^{i}\frac{\partial}{\partial x_{i}}$
denotes the radial vector field or {}``Euler operator''.\end{thm}
\begin{proof}
By (\ref{eq:DefinitionHomogeneous}) we infer that\[
\left(\rho\partial_{\rho}\right)\rho^{-D}\left\langle u_{\rho},\phi\right\rangle =\left(\rho\partial_{\rho}\right)\left\langle u,\phi\right\rangle =0\,.\]
Computing the derivative gives:\begin{align*}
\rho\partial_{\rho}\left(\rho^{-D}\left\langle u_{\rho},\phi\right\rangle \right) & =-D\rho^{-D}\left\langle u_{\rho},\phi\right\rangle +\rho^{-D}\left\langle \left(x\cdot\partial_{x}\right)u_{\rho},\phi\right\rangle \\
 & =\rho^{-D}\left\langle \left(x\cdot\partial_{x}-D\right)u_{\rho},\phi\right\rangle \,,\end{align*}
hence, evaluating at $\rho=1$ gives (\ref{eq:HomogeneousDistributionEulerThm}).
Conversely let $\varphi(\rho):=\left\langle u,\phi^{\rho}\right\rangle $,
then\begin{align*}
\rho\varphi'(\rho) & =\rho\partial_{\rho}\left\langle u(x),\rho^{-d}\phi(\rho^{-1}x)\right\rangle \\
 & =-d\left\langle u,\phi^{\rho}\right\rangle -\left\langle u(x),\left(x\cdot\partial_{x}\right)\phi^{\rho}(x)\right\rangle \\
 & =-d\left\langle u,\phi^{\rho}\right\rangle +d\left\langle u,\phi^{\rho}\right\rangle +\left\langle \left(x\cdot\partial_{x}\right)u(x),\phi^{\rho}(x)\right\rangle \\
 & =D\left\langle u,\phi^{\rho}\right\rangle =D\varphi(\rho)\end{align*}
Hence we have the differential equation\[
\frac{\varphi'(\rho)}{\varphi(\rho)}=\frac{D}{\rho}\]
which is solved by $\varphi(\rho)=C\rho^{D}$, i.e., $C=\varphi(1)$.
This means\[
\left\langle u_{\rho},\phi\right\rangle =\rho^{D}\left\langle u,\phi\right\rangle \,.\qedhere\]

\end{proof}
Observe that $u$ is a (weak) eigenvector of $x\cdot\partial_{x}$
to the (weak) eigenvalue $D$.
\begin{thm}
[cf. { \citep[Thm. 3.2.3]{Hoermander2003}}]\label{thm:HomogeneousExtension}
Let $u\in\D'(\mathbb{R}^{d}\backslash\left\{ 0\right\} )$ scale homogeneously
of degree $D\in\mathbb{C}$ and let $-D\notin\mathbb{N}_{0}+d$, then
$u$ has a unique extension $\dot{u}\in\D'(\mathbb{R}^{d})$ which
is homogeneous of degree $D$. The map\[
\D'(\mathbb{R}^{d}\backslash\left\{ 0\right\} )\ni u\mapsto\dot{u}\in\D'(\mathbb{R}^{d})\]
is continuous.
\end{thm}
A proof of the theorem can be found in the book of Hörmander. Observe,
however, that the existence of a (not necessarily homogeneous) extension
already follows from Theorem~\ref{thm:ExtensionScalingDegree}. And
if there is a homogeneous extension the uniqueness follows from the
proof of the same theorem and the fact that the derivatives of Dirac's
$\delta$-distribution have integer scaling degree greater or equal
$d$. Observe that any homogeneous extension $\dot{u}\in\D'(\mathbb{R}^{d})$
of a homogeneous distribution $u\in\D(\mathbb{R}^{d}\backslash\left\{ 0\right\} )$
in particular has the same scaling degree, $\sd(\dot{u})=\sd(u)$
by Remark~\ref{rem:ScalingDegreeHomogeneity}. Thus Theorem~\ref{thm:HomogeneousExtension}
really is a refinement of the previous result (Theorem~\ref{thm:ExtensionScalingDegree})
for the special case of homogeneous distributions.

\section{Heterogeneous Distributions}

A straight forward generalization of Theorem~\ref{thm:HomogeneousExtension}
to the case when the distribution is not homogeneous, but is given
as a finite sum of homogeneous parts will be important for the construction
of the dimensionally regularized amplitude in Section~\ref{sec:TheRegularizedAmplitude}.
\begin{defn}
[Heterogeneous Distribution]\label{def:HeterogeneousDistribution}A
distribution $u\in\D'$ is called \emph{heterogeneous of order $k\in\mathbb{N}$
and multidegree} $\bld{\alpha}=\left\{ \alpha_{1},\dots,\alpha_{k}\right\} $
($i\neq j\Leftrightarrow\alpha_{i}\neq\alpha_{j}$), if \begin{equation}
\prod_{j=1}^{k}\left(x\cdot\partial_{x}-\alpha_{j}\right)u=0\,.\label{eq:FinitelyHeterogeneous}\end{equation}
\end{defn}
\begin{lem}
\label{lem:HeterogeneousUniqueDecomposition}Heterogeneous distributions
of finite order have a unique decomposition into their homogeneous
components.\end{lem}
\begin{proof}
Let $u$ be a heterogeneous distribution of order $k$ and multidegree
$\bld{\alpha}=\left\{ \alpha_{1},\dotsc,\alpha_{k}\right\} $, i.e.,
$u$ fulfills (\ref{eq:FinitelyHeterogeneous}). Then \[
P_{i}:=\prod_{j\neq i}\frac{x\cdot\partial_{x}-\alpha_{j}}{\alpha_{i}-\alpha_{j}}\]
projects $u$ onto the eigenspace of $x\cdot\partial_{x}$ to eigenvalue
$\alpha_{i}$, since $\left(x\cdot\partial_{x}-\alpha_{i}\right)P_{i}u=0$
by assumption (\ref{eq:FinitelyHeterogeneous}), and\[
P_{i}^{2}u=\prod_{j\neq i}\frac{x\cdot\partial_{x}-\alpha_{j}}{\alpha_{i}-\alpha_{j}}P_{i}u=\prod_{j\neq i}\frac{\alpha_{i}-\alpha_{j}}{\alpha_{i}-\alpha_{j}}P_{i}u=P_{i}u\,.\]
Thus $u_{i}:=P_{i}u$ is homogeneous of degree $\alpha_{i}$, and
$u$ can be uniquely decomposed into eigenvectors of $x\cdot\partial_{x}$,
\[
u=\sum_{i=1}^{k}u_{i}\,.\qedhere\]
\end{proof}
\begin{cor}
[Ext. of Heterogeneous Distributions]\label{cor:ExtensionHeterogeneousDistribution}Let
$u\in\D'(\mathbb{R}^{d}\backslash\left\{ 0\right\} )$ be a heterogeneous
distribution of multidegree $\bld{\alpha}=\left\{ \alpha_{1},\dots,\alpha_{k}\right\} $.
Let furthermore \[
-\alpha_{j}\in\mathbb{C}\backslash\mathbb{N}_{0}\quad\forall j\in\left\{ 1,\dots,k\right\} .\]
Then $u$ has a unique heterogeneous extension $\dot{u}\in\D'(\mathbb{R}^{d})$
of the same multidegree.\end{cor}
\begin{proof}
Uniqueness. Let $\ddot{u}\in\D'(\mathbb{R}^{d})$ be a second extension
of $u$, then $\ddot{u}$ differs from $\dot{u}$ by a distribution
supported at $\left\{ 0\right\} $,\[
\ddot{u}-\dot{u}=\sum_{\left|\alpha\right|\leq\rho}C_{\alpha}\delta^{\left(\alpha\right)}\,.\]
However, any term in the sum on the right hand side has integer degree
of homogeneity smaller or equal to $-d$. Hence the sum in the above
expression is not annihilated by $\prod_{j=1}^{k}\left(x\cdot\partial_{x}-\alpha_{j}\right)$
with $-\alpha_{j}\notin\mathbb{N}_{0}$, whereas $\dot{u}$ is annihilated
by assumption. Thus $\ddot{u}$ is not annihilated by $\prod_{j=1}^{k}\left(x\cdot\partial_{x}-\alpha_{j}\right)$
and hence is not heterogeneous of multidegree $\bld{\alpha}=\left\{ \alpha_{1},\dots,\alpha_{k}\right\} $.

Existence. By Lemma~\ref{lem:HeterogeneousUniqueDecomposition} above,
$u$ has a unique decomposition into homogeneous parts, $u_{i}\in\D'(\mathbb{R}^{d}\backslash\left\{ 0\right\} )$,
each of which has a unique homogeneous extension $\dot{u}_{i}$ by
Theorem~\ref{thm:HomogeneousExtension}. Hence\[
\dot{u}=\sum_{i=1}^{k}\dot{u}_{i}\]
is an extension of $u$ which is heterogeneous of order $k$ and multidegree
$\bld{\alpha}$.
\end{proof}
If the eigenvalues in the product (\ref{eq:FinitelyHeterogeneous})
coincide, $\alpha_{i}=\alpha$ $\forall i\in\left\{ 1,\dots,k\right\} $,
we get to the notion of \emph{almost homogeneous distributions}, which
are homogeneous up to a polynomial of order $k-1$ in $\ln(\rho)$,
where $\rho>0$ is the scaling parameter. Hollands and Wald proved
a uniqueness result for the extension also in this case \citep[Lem.~4.1]{HollandsWald2002},
see also \citep[Prop.~3.3]{DuetschFredenhagen2004}. However, the
distributions we will analyze in this work are at most heterogeneous
of finite order, and thus we will not need the lemma of Hollands and
Wald. This is due to the fact that we regard only regularized distributions,
a concept to be defined in the next section. We want to remark, however,
that in the limit where the regularization is removed we will get
back almost homogeneous distributions in the generic case.

\section{\label{sec:RegularizationOfDistributions}Regularization of Distributions}
\begin{defn}
[Regularization]\label{def:Regularization}Let $u\in\D'(\mathbb{R}^{d}\backslash\left\{ 0\right\} )$
be a distribution with degree of divergence $\div(u)=\lambda$. Let
$\tilde{u}\in\D_{\lambda}'(\mathbb{R}^{d})$ be the unique extension
of $u$ with the same degree of divergence. A family $\left\{ u^{\zeta}\right\} _{\zeta\in\Omega\backslash\left\{ 0\right\} }$
of distributions $u^{\zeta}\in\D'(\mathbb{R}^{d})$, where $\Omega\subset\mathbb{C}$
is a neighborhood of the origin, is called a \emph{regularization
of $u$}, if\begin{equation}
\forall g\in\D_{\lambda}(\mathbb{R}^{d}):\quad\lim_{\zeta\rightarrow0}\left\langle u^{\zeta},g\right\rangle =\left\langle \tilde{u},g\right\rangle \,.\label{eq:RegularizationDefiningProerty}\end{equation}
The regularization $\left\{ u^{\zeta}\right\} $ is called \emph{analytic},
if for any function $f\in\D(\mathbb{R}^{d})$ the map \begin{equation}
\zeta\mapsto\left\langle u^{\zeta},f\right\rangle \label{eq:RegularizationAnalyticity}\end{equation}
is analytic for $\zeta\in\Omega\backslash\left\{ 0\right\} $, possibly
with a pole of finite order at the origin, i.e., (\ref{eq:RegularizationAnalyticity})
is a meromorphic function. We speak of a \emph{finite regularization},
if\[
\forall f\in\D(\mathbb{R}^{d}):\quad\lim_{\zeta\rightarrow0}\left\langle u^{\zeta},f\right\rangle \in\mathbb{C}\,,\]
in this case $\lim_{\zeta\rightarrow0}u^{\zeta}\in\D'(\mathbb{R}^{d})$
is a renormalization, or extension of $u$.
\end{defn}
Given a regularization $\left\{ u^{\zeta}\right\} $ of $u$ we have
for all $f\in\D(\mathbb{R}^{d})$ and any projection $W:\D\rightarrow\D_{\lambda}$
that\begin{equation}
\left\langle \tilde{u},Wf\right\rangle =\lim_{\zeta\rightarrow0}\left\langle u^{\zeta},Wf\right\rangle \,.\label{eq:Regularization-1}\end{equation}
According to Lemma~\ref{lem:W-Projection-Functions} for any $W$-projection
there exists a family of functions $\left\{ w_{\alpha}\in\D:\left(\partial_{\beta}w_{\alpha}\right)\!(0)=\delta_{\alpha\beta}\right\} $
such that \[
W=1-\sum_{\left|\alpha\right|\leq\lambda}\left(-1\right)^{\left|\alpha\right|}w_{\alpha}\delta^{\left(\alpha\right)}\,,\]
and since $u^{\zeta}\in\D'(\mathbb{R}^{d})$ we can write (\ref{eq:Regularization-1})
as\begin{equation}
\left\langle \tilde{u},Wf\right\rangle =\lim_{\zeta\rightarrow0}\left[\left\langle u^{\zeta},f\right\rangle -\sum_{\left|\alpha\right|\leq\lambda}\left\langle u^{\zeta},w_{\alpha}\right\rangle f^{\left(\alpha\right)}(0)\right]\,.\label{eq:Regularization-FiniteLimit}\end{equation}
In the generic case the limit on the right hand side cannot be split,
since the limits of the individual terms might not exist. However,
if $\left\{ u^{\zeta}:\zeta\in\Omega\backslash\left\{ 0\right\} \right\} $
is an analytic regularization, the individual terms can be expanded
in Laurent series around $\zeta=0$, and since the overall limit is
finite the principal parts ($\pp$) of these Laurent series have to
coincide,\begin{align*}
\forall f\in\D:\quad\pp\left\langle u^{\zeta},f\right\rangle  & =\pp\!\left(\sum_{\left|\alpha\right|\leq\lambda}\left\langle u^{\zeta},w_{\alpha}\right\rangle f^{\left(\alpha\right)}(0)\right)=\sum_{\left|\alpha\right|\leq\lambda}f^{\left(\alpha\right)}(0)\,\pp\left\langle u^{\zeta},w_{\alpha}\right\rangle .\end{align*}
We conclude that the principal part of any analytic regularization
$\left\{ u^{\zeta}\right\} $ is a polynomial in derivatives of Dirac's
$\delta$-distribution up to order $\lambda=\div(u)$,\begin{equation}
\pp(u^{\zeta})=\sum_{\left|\alpha\right|\leq\div(u)}C_{\alpha}(\zeta)\delta^{\left(\alpha\right)}\in\E_{\mathrm{Dirac}}'\,,\label{eq:Regularization-PP-Local}\end{equation}
where we set $C_{\alpha}(\zeta)=\pp\left\langle u^{\zeta},w_{\alpha}\right\rangle $.
That is, $\pp(u^{\zeta})$ is local for all $\zeta\in\Omega\backslash\left\{ 0\right\} $
and vanishes if $\div(u)<0$. The fact that the principal part of
an analytic regularization is a local distribution will be crucial
for the discussion in Chapter~\ref{cha:ForestFormula} of the present
thesis. In particular this implies that we can fix an extension $\dot{u}_{\MS}\in\D'$
of $u$ by setting\begin{equation}
\left\langle \dot{u}_{\MS},f\right\rangle :=\lim_{\zeta\rightarrow0}\left[\left\langle u^{\zeta},f\right\rangle -\pp(\left\langle u^{\zeta},f\right\rangle )\right]\,.\label{eq:RegularizationMSDefinition}\end{equation}
This way of choosing a renormalization of $u$ is called \emph{minimal
subtraction} (MS). By construction $\dot{u}_{\MS}$ has the same scaling
degree as $u$, and thus minimal subtraction can be implemented as
a $W$-extension, cf.~Lemma~\ref{lem:W-Extension}. We choose a
projection\[
W^{\MS}:\D\rightarrow\D_{\lambda}\,,\quad\lambda=\div(u)\,,\]
which fulfills\begin{equation}
\forall f\in\D:\quad\left\langle \dot{u}_{\MS},f\right\rangle =\left\langle \tilde{u},W^{\MS}f\right\rangle \,.\label{eq:RegularizationMS-ExtensionW}\end{equation}
Let us regard this projection for finite $\zeta\in\Omega\backslash\left\{ 0\right\} $.
The regular part of $\left\langle u^{\zeta},f\right\rangle $ is given
by \begin{align*}
\left\langle \rp(u^{\zeta}),f\right\rangle  & =\left\langle u^{\zeta}-\pp(u^{\zeta})\,,\, W^{\MS}f+\sum_{\left|\alpha\right|\leq\div(u)}w_{\alpha}^{\MS}f^{\left(\alpha\right)}(0)\right\rangle \\
 & =\left\langle u^{\zeta},W^{\MS}f\right\rangle +\sum_{\left|\alpha\right|\leq\div(u)}\left\langle \rp(u^{\zeta}),w_{\alpha}^{\MS}\right\rangle f^{\left(\alpha\right)}(0)\,,\end{align*}
since $\pp\left\langle u^{\zeta},W^{\MS}f\right\rangle =0$ by (\ref{eq:RegularizationDefiningProerty}).
The first term on the right hand side, as well as the left hand side
of this equation tend to $\left\langle \dot{u}_{\MS},f\right\rangle $
as $\zeta\rightarrow0$, cf.~(\ref{eq:RegularizationMSDefinition})/(\ref{eq:RegularizationMS-ExtensionW}),\begin{equation}
\lim_{\zeta\rightarrow0}\left\langle \rp(u^{\zeta}),f\right\rangle =\lim_{\zeta\rightarrow0}\left\langle u^{\zeta},W^{\MS}f\right\rangle =\left\langle \dot{u}_{\MS},f\right\rangle \,.\label{eq:Regularization-MS}\end{equation}
Hence the sum on the right hand side has to vanish in this limit.
Since it is the regular part of some Laurent series we infer that
it is at least of order one in $\zeta$,\begin{equation}
\forall\zeta\in\Omega\backslash\left\{ 0\right\} :\quad\left\langle \rp(u^{\zeta}),f\right\rangle =\left\langle u^{\zeta},W^{\MS}f\right\rangle +\mathcal{O}(\zeta)\,.\label{eq:RegularizationMS-FiniteZeta}\end{equation}
Hence for finite regularization parameter, $\zeta\in\Omega\backslash\left\{ 0\right\} $,
minimal subtraction can be expressed as a $W$-projection up to a
contribution which vanishes identically in the limit $\zeta\rightarrow0$.
This fact will become important in the discussion of minimal subtraction
on the level of graph amplitudes in Chapter~\ref{cha:Minimal-Subtraction}
and in particular for the proof of Proposition~\ref{pro:RedundantProjections}.

The coefficients $C_{\alpha}(\zeta)$ in (\ref{eq:Regularization-PP-Local})
are called \emph{counterterms}. They are local in the sense that they
are the coefficients of a local distribution. In particular they do
not depend on the chosen $W$-projection. The $C_{\alpha}(\zeta)$
are often referred to as infinite counterterms, since they do not
possess a limit as $\zeta\rightarrow0$ and the way of introducing
them by splitting the $W$-projection before taking the limit in (\ref{eq:Regularization-FiniteLimit})
was also discussed in \citep{KuznetsovTkachovVlasov1996}.

In Fourier space the $C_{\alpha}(\zeta)$ are the coefficients of
a polynomial in (external) momenta $p$, \[
\FT(\sum_{\left|\alpha\right|\leq\lambda}C_{\alpha}(\zeta)\delta^{\left(\alpha\right)})(p)=\sum_{\left|\alpha\right|\leq\lambda}\tfrac{i^{\left|\alpha\right|}}{\left(2\pi\right)^{\frac{d}{2}}}C_{\alpha}(\zeta)p^{\alpha}\,.\]
In this sense the counterterms are invariant under Fourier transform
and thus provide a basis on which one can compare the position space
approach to dimensional regularization and minimal subtraction ($\mbox{DimReg}^{x}+\mbox{MS})$,
to be discussed in the present work, with the standard approach in
momentum space.

\global\long\def\fps#1#2{#1\![[#2]]}

\global\long\def\poisson#1#2{\left\lfloor #1,#2\right\rceil }

\global\long\def\bld#1{\boldsymbol{#1}}

\begin{fmffile}{SettingPAQFT03}

\def\FG{\parbox{13mm}{
\begin{center}
\begin{fmfgraph}(20,15)
\fmfleft{F} \fmfright{G}
\fmfdot{F,G}
\fmf{phantom}{F,G}
\end{fmfgraph}
\end{center}}}

\def\FoneFG{\parbox{20mm}{
\begin{center}
\begin{fmfgraph}(20,15)
\fmfleft{F} \fmfright{G}
\fmfdot{F,G}
\fmf{phantom}{F,G}
\fmf{plain}{F,F}
\end{fmfgraph}
\end{center}}}

\def\FGoneG{\parbox{20mm}{
\begin{center}
\begin{fmfgraph}(20,15)
\fmfleft{F} \fmfright{G}
\fmfdot{F,G}
\fmf{phantom}{F,G}
\fmf{plain}{G,G}
\end{fmfgraph}
\end{center}}}

\def\FoneG{\parbox{13mm}{
\begin{center}
\begin{fmfgraph}(20,15)
\fmfleft{F} \fmfright{G}
\fmfdot{F,G}
\fmf{plain}{F,G}
\end{fmfgraph}
\end{center}}}

\chapter[Setting of pAQFT]{\label{cha:Setting-pAQFT}The Setting of Perturbative~Algebraic~Quantum~Field~Theory}

We want to analyze the methods of analytic regularization (in particular
dimensional regularization) and minimal subtraction, introduced on
the level of distributions in the previous chapter, in the algebraic
approach to perturbative Quantum Field Theory. More specifically,
we will use the framework of \emph{perturbative Algebraic Quantum
Field Theory (pAQFT)} \citep{Brunetti2009}. Although the methods
of pAQFT apply in a much more general framework we want to restrict
ourselves in the present work to the case of $d$-dimensional Minkowski
spacetime. The aim of this chapter is to introduce the main concepts
and the basic constructions of pAQFT, which will be used in the main
part of this work.

\section{Classical Field Theory and Deformation Quantization}

In 1990 Dito showed how the formalism of deformation quantization
can be applied to field theory \citep{Dito1990}. He constructed the
algebra of the free scalar field without reference to an underlying
Hilbert space. In his work he used the earlier analysis of the deformation
of algebras by Bayen, Flato, Fronsdal, Lichnerowicz, and Sternheimer
\citep{Bayen1978-I+II}. Dito also related his approach to the (despite
its mathematical problems) widely known and used Feynman path integral
approach to quantum field theory. The work of Brunetti, Dütsch, and
Fredenhagen \citep{Dutsch2000,Dutsch2001,DuetschFredenhagen2003,DuetschFredenhagen2004,Brunetti2007a,Brunetti2009}
showed that the star product approach of Dito can be extended to a
purely algebraic formulation of perturbative Quantum Field Theory
(pQFT) in general and perturbative renormalization theory in particular.
To give a motivation for the deformation view point in field theory,
we briefly review in this first section the structure of the algebra
of classical field theory and define its Poisson structure. The construction
of the algebra of observables of pQFT will then be carried out in
full detail in the next section for the case of flat Minkowski spacetime.

Let $\M$ denote the $d$-dimensional Minkowski spacetime with metric
tensor $\eta=\left(1,-1,\dots,-1\right)$ on the diagonal. Timelike
(spacelike) vectors fulfill $x^{2}>0$ ($x^{2}<0$) and the set of
all timelike vectors is called the open lightcone $\lc=\lc^{+}\dot{\cup}\lc^{-}$.
It is the disjoint union of two connected components, which we refer
to as the forward and backward lightcones, $\lc^{\pm}=\left\{ x\in\M|\, x^{2}>0,\,\pm x^{0}>0\right\} $.
We denote by $\overline{\lc^{\pm}}$ and $\partial\lc^{\pm}$ the
closure and boundary of these sets, respectively.

The configuration space of classical field theory is the space of
smooth functions\[
\varphi:\M\rightarrow\mathbb{C}\,,\quad\varphi\in\E(\M)\,,\]
and the observables are (not necessarily linear) functionals on this
space.
\begin{defn}
[Smooth Functional] A functional\[
\begin{array}{rccl}
F: & \E(\M) & \rightarrow & \mathbb{C}\\
 & \varphi & \mapsto & F(\varphi)\,,\end{array}\]
is called \emph{smooth}, if for any $\varphi\in\E(\M)$ and for all
$n\in\mathbb{N}$ its $n$th functional derivative,\[
\left\langle F^{\left(n\right)}(\varphi),h_{1}\otimes\cdots\otimes h_{n}\right\rangle :=\frac{d^{n}}{d\lambda_{1}\cdots d\lambda_{n}}F(\varphi+\sum_{i=1}^{n}\lambda_{i}h_{i})\bigg|_{\lambda_{1}=\dots=\lambda_{n}=0}\,,\quad h_{i}\in\E(\M)\,,\]
exists as a symmetric distribution with compact support, $F^{\left(n\right)}(\varphi)\in\E'(\M^{n})$.
Occasionally we will write $F^{\left(n\right)}(\varphi)=\frac{\delta^{n}F}{\delta\varphi^{n}}$.
The triangular brackets denote the dual pairing (of $\E'(\M^{n})$
with $\E(\M^{n})$ here). We denote the space of smooth functionals
by $\widetilde{\mathcal{F}}(\M)$.
\end{defn}
{}
\begin{defn}
[Support of a Functional]Let $F\in\widetilde{\mathcal{F}}(\M)$,
then we define the support of $F$ implicitly by the equivalence,
$h\in\E(\M)$,\[
\supp(F)\cap\supp(h)=\emptyset\quad\Leftrightarrow\quad\forall\varphi\in\E(\M):F(\varphi+h)=F(\varphi)\,.\]

\end{defn}
Smooth functionals form a (commutative) algebra, $\left(\widetilde{\mathcal{F}}(\M),\cdot\right)$,
with respect to the pointwise product, $\forall F,G\in\widetilde{\mathcal{F}}(\M)$,
$\forall\varphi\in\E(\M)$,\[
\left(F\cdot G\right)(\varphi):=F(\varphi)G(\varphi)\,.\]
Typical examples of such functionals are,\begin{equation}
\mbox{(i)}\quad F(\varphi)=\int f(x)\,{\varphi(x)}^{3}\, dx\quad,\quad\mbox{(ii)}\quad G(\varphi)=\int g(x)\left[\partial\varphi(x)\right]^{2}\, dx\,,\label{eq:FunctionalsExamples}\end{equation}
\[
\mbox{or}\quad\mbox{(iii)}\quad K(\varphi)=\int k(x,y)\,\varphi(x)\,\varphi(y)\, dx\, dy\,,\]
where $f,g\in\D(\M)$, $k\in\D(\M^{2})$ are test functions of compact
support, and the integral is taken over the whole spacetime $\M$,
or over $\M^{2}$, respectively. The field itself is represented as
a linear evaluation functional, $\forall f\in\D(\M)$,\[
\varphi\mapsto\varphi(f)=\int f(x)\,\varphi(x)\, dx\,.\]

It is not possible (in the framework presented here) to deform the
whole algebra of smooth functionals $\left(\widetilde{\mathcal{F}}(\M),\cdot\right)$.
However, we can restrict ourselves to a suitably chosen subalgebra
of functionals $\mathcal{F}(\M)\subset\widetilde{\mathcal{F}}(\M)$,
which will have a quantized counterpart in deformation quantization.
This {}``deformable algebra'' is defined by imposing conditions
on the \emph{wave front set} of the functional derivatives $F^{\left(n\right)}(\varphi)\in\E'(\M^{n})$,
$F\in\widetilde{\mathcal{F}}(\M)$. The wave front set of a distribution
$u\in\D'(\M^{n})$, roughly speaking, is a conic subset of the cotangent
bundle $\WF(u)\subset\dot{T}^{*}(\M^{n})=T^{*}(\M^{n})\backslash\left\{ 0\right\} $,
where the first component gives the singular support, $\singsupp(u)$,
and the second gives the directions in which the Fourier transform
$\FT(u)$ does not decrease rapidly. The precise definition of this
mathematical tool, as well as the large number of deep implications
of the wave front set properties of distributions can be found in
the book of Hörmander, cf.~\citep[Def.~8.1.2]{Hoermander2003}.
\begin{defn}
[Deformable Algebra]\label{def:DeformableAlgebra}A smooth functional
$F\in\widetilde{\mathcal{F}}(\M)$ is an element of the \emph{deformable
algebra} $\left(\mathcal{F}(\M),\cdot\right)$, if for all $n\in\mathbb{N}$
the wave front set of $F^{\left(n\right)}(\varphi)$ does not meet
the $n$-fold product of the closed forward or backward light cone,\[
\forall\varphi:\quad\WF(F^{\left(n\right)}(\varphi))\cap\left(\supp(F^{\left(n\right)}(\varphi))\times\left[{\overline{\lc^{+}}}^{n}\cup{\overline{\lc^{-}}}^{n}\right]\right)=\emptyset\,.\]
We refer to elements of the deformable algebra as \emph{deformable
functionals}.
\end{defn}
Since $\WF(u)\subset\dot{T}^{*}(\M^{n})$, the forward and backward
lightcone are here to be understood as subsets of the cotangent space,
${\overline{\lc^{+}}}^{n},{\overline{\lc^{-}}}^{n}\subset T_{x}^{*}(\M^{n}$).

The deformable algebra $\left(\mathcal{F}(\M),\cdot\right)$ can be
made into a Poisson algebra by using the Peierls bracket \citep{DuetschFredenhagen2003}.
The Poisson structure is defined by\begin{equation}
\poisson FG(\varphi):=i\left\langle F^{\left(1\right)}(\varphi),\Delta*G^{\left(1\right)}(\varphi)\right\rangle \,.\label{eq:PoissonBracket}\end{equation}
where $*$ denotes convolution, and $\Delta$ is defined as the difference
of the unique retarded and advanced fundamental solutions $\Delta_{\ret},\Delta_{\adv}\in\D'(\M)$
of the Klein-Gordon operator, cf.~\citep[Eq.~(28)]{DuetschFredenhagen2003}.
Regard these fundamental solutions in more detail, we have\[
\left(\Box+m^{2}\right)\Delta_{\substack{\ret\\
\adv}
}=\delta\,,\quad\mbox{with}\qquad\supp(\Delta_{\substack{\ret\\
\adv}
})=\overline{\lc^{\pm}}\]
\[
\Delta:=\left(\Delta_{\ret}-\Delta_{\adv}\right)\,.\]
We also define the corresponding {}``causal two point function''
or {}``commutator function'', \[
\forall f,g\in\D(\M):\quad\Delta(f,g):=\left\langle f,\Delta*g\right\rangle \,,\]
and it should be clear from the context, which of the two interpretations
of the symbol $\Delta$ is to be understood. The wave-front set of
the commutator function is given by\begin{equation}
\WF(\Delta)=\left\{ \left(x,y,k,-k\right)\in\dot{T}^{*}(\M^{2})|\left(x-y\right)^{2}=0,\, k\Vert(x-y),\, k^{2}=0\right\} \,.\label{eq:WFCommutatorFunction}\end{equation}
And by \citep[Thm.~8.2.10]{Hoermander2003} the pointwise product
of the distributions in\[
\poisson FG=\int\Delta(x,y)\cdot F^{\left(1\right)}(x)\, G^{\left(1\right)}(y)\, dx\, dy\,,\quad F,G\in\mathcal{F}(\M)\,,\]
is well-defined, if the covectors in the second component cannot add
up to zero. The wave front set of the tensor power $\left(F^{\left(1\right)}\otimes G^{\left(1\right)}\right)(\varphi)\in\E'(\M^{2})$
is given by, cf.~\citep[Thm.~8.2.9]{Hoermander2003}, \begin{align*}
\WF(F^{\left(1\right)}\otimes G^{\left(1\right)}) & \subset\WF(F^{\left(1\right)})\times\WF(G^{\left(1\right)})\\
 & \qquad\cup\left\{ \left[\supp(F)\times\left\{ 0\right\} \right]\times\WF(G^{\left(1\right)})\right\} \\
 & \qquad\cup\left\{ \WF(F^{\left(1\right)})\times\left[\supp(G)\times\left\{ 0\right\} \right]\right\} \,,\end{align*}
and by comparing this with (\ref{eq:WFCommutatorFunction}) we see
that the covectors cannot add up to zero, if $F,G\in\mathcal{F}(\M)$.
For $\left(x,y,k_{x},k_{y}\right)\in\WF(\Delta)$ we clearly have
$\left(k_{x},k_{y}\right)\in\overline{\lc}\times\overline{\lc}$,
and since $k_{f},k_{g}\notin\overline{\lc}$ for $(x,k_{f})\in\WF(F^{\left(1\right)}(\varphi))$
and $(y,k_{g})\in\WF(G^{\left(1\right)}(\varphi))$ neither one of
the equations \[
k_{x}+k_{f}=0\qquad\mbox{or}\qquad k_{y}+k_{g}=0\]
has a solution. By a similar argument we also have $\poisson FG\in\mathcal{F}(\M)$
in that case. It was proven that, besides linearity, antisymmetry
and the Leibniz rule, the bracket $\poisson{\cdot}{\cdot}$ fulfills
also the Jacobi identity, and thus defines a genuine Poisson structure
on $\mathcal{F}(\M)$ \citep{DuetschFredenhagen2003}. As was discussed
there a (formal) quantization of $\left(\mathcal{F}(\M),\poisson{\cdot}{\cdot},\cdot\right)$
can be understood as a map\[
\left(\mathcal{F}(\M),\poisson{\cdot}{\cdot},\cdot\right)\rightarrow\left(\fps{\mathcal{F}(\M)}{\hbar},\left[\cdot,\cdot\right]_{\star},\star\right)\,,\]
to a non-commutative, associative algebra, such that\begin{equation}
F\star G\xrightarrow{\hbar\rightarrow0}F\cdot G\quad\mbox{and}\quad\tfrac{1}{\hbar}\left[F,G\right]_{\star}\xrightarrow{\hbar\rightarrow0}\poisson FG\,.\label{eq:QuantizationCondition}\end{equation}
In particular we have for the field itself\begin{equation}
\left[\varphi(f),\varphi(g)\right]_{\star}=i\hbar\Delta(f,g)\,.\label{eq:QuantizationConditionFieldCommutator}\end{equation}
The product $\star$ of the quantized algebra is not to be confused
with the notation $*$ for convolution.

\section{\label{sec:AlgebraOfObservables}The Algebra of Observables}

As the name suggests the subalgebra $\left(\mathcal{F}(\M),\cdot\right)$
can be deformed to give the algebra of perturbative Quantum Field
Theory. To construct this algebra, let $H\in\D'(\M)$ be a \emph{Hadamard
distribution}. That is, cf.~\citep{Radzikowski1996}, $H$ is a (weak)
solution of the Klein Gordon equation

\begin{equation}
\left(\Box+m^{2}\right)H=0\label{eq:HadamardKleinGordonSolution}\end{equation}
and the corresponding two point distribution, defined by \begin{equation}
H(f,g):=\left\langle f,H*g\right\rangle \,,\label{eq:HadamardTwoPointDistribution}\end{equation}
satisfies the causality condition\begin{equation}
\forall f,g\in\D(\M):\quad H(f,g)-H(g,f)=i\Delta(f,g)\label{eq:CausalityHadamardAntisymmetricPart}\end{equation}
and has the {}``positive frequency'' - or Hadamard wave front set\begin{align}
\WF(H)=\left\{ \left(x,y,k,-k\right)\in\dot{T}^{*}(\M^{2})|\right. & \mbox{if }x\neq y:\,\left(x-y\right)^{2}=0,\, k\Vert(x-y),\, k^{0}>0\,;\nonumber \\
 & \mbox{if }x=y:\, k^{2}=0,\, k^{0}>0\left.\vphantom{\left(x,y,k,-k\right)\in T^{*}(\M^{2})|}\right\} .\label{eq:HadamardWF}\end{align}
One example of such a Hadamard distribution is the Wightman function
$\Delta_{+}$, i.e. the positive frequency part of the commutator
function $i\Delta$, $\Delta_{+}=\tfrac{i}{2}\Delta+\Delta_{1}$.
However, any other distribution $H$, which differs from $\Delta_{+}$
by a smooth, symmetric, Lorentz invariant solution of the Klein-Gordon
equation will also fulfill the defining equations (\ref{eq:HadamardKleinGordonSolution}),
(\ref{eq:CausalityHadamardAntisymmetricPart}), and (\ref{eq:HadamardWF}),
and thus be a valid Hadamard function for the construction below.
The algebraic properties are completely independent of this choice,
the analytic properties, however, will change significantly for different
choices of $H$. A Hadamard solution, especially well-suited for our
purposes, will be constructed explicitly in Chapter~\ref{cha:DimRegHadamard}.
Equation (\ref{eq:CausalityHadamardAntisymmetricPart}) fixes the
antisymmetric part of $H$ to be $\frac{i}{2}\Delta$, which implies
that the Poisson structure (\ref{eq:PoissonBracket}) can be induced
by the bidifferential operator\begin{equation}
\Gamma_{H}:=\int dx\, dy\, H(x,y)\frac{\delta}{\delta\varphi(x)}\otimes\frac{\delta}{\varphi(y)}\label{eq:BiDifferentialStar}\end{equation}
in the following sense\[
\poisson FG=M\circ\Gamma_{H}\left(F\otimes G-G\otimes F\right)\,.\]
By means of this differential operator a formal quantization of the
Poisson algebra $\left(\mathcal{F}(\M),\poisson{\cdot}{\cdot},\cdot\right)$
can be given in form of the following
\begin{prop}
[Deformed Algebra]\label{pro:DeformedAlgebra} Let $\fps{\mathcal{F}(\M)}{\hbar}$
be the space of formal power series in $\hbar$ with coefficients
in the deformable algebra $\mathcal{F}(\M)$ and regard $F,G\in\fps{\mathcal{F}(\M)}{\hbar}$.
Then\\
\begin{minipage}[c][1\totalheight][t]{0.57\columnwidth}%
\begin{equation}
\xyC{1pt}\xymatrix{\fps{\mathcal{F}(\M)}{\hbar}^{\otimes2}\ar@{-->}[dr]_{\star}\ar[rr]^{\exp(\hbar\Gamma_{H})} &  & \fps{\mathcal{F}(\M)}{\hbar}^{\otimes2}\,,\ar[dl]^{\cdot}\\
 & \fps{\mathcal{F}(\M)}{\hbar}}
\xyC{2pc}\label{eq:StarProduct}\end{equation}
\end{minipage}\hfill%
\begin{minipage}[c][1\totalheight][t]{0.39\columnwidth}%
\[
F\star G:=\sum_{k=0}^{\infty}\frac{\hbar^{k}}{k!}\left\langle F^{\left(k\right)},H^{\otimes k}G^{\left(k\right)}\right\rangle \]
\end{minipage}\\
defines a (non-commutative) associative product on $\fps{\mathcal{F}(\M)}{\hbar}$
which fulfills the quantization condition (\ref{eq:QuantizationCondition}).
We call $\left(\fps{\mathcal{F}(\M)}{\hbar},\star\right)$ the algebra
of pQFT.\end{prop}
\begin{proof}
First we want to argue that $F\star G$ is well-defined for any pair
$F,G\in\fps{\mathcal{F}(\M)}{\hbar}$. The kernel representation of
the $k$th term of (\ref{eq:StarProduct}) is given by \[
\left\langle F^{\left(k\right)},H^{\otimes k}G^{\left(k\right)}\right\rangle =\int d\vec{x}\, d\vec{y}\, F^{\left(k\right)}(x^{1},\dots,x^{k})\,\prod_{i=1}^{k}H(x^{i},y^{i})\, G^{\left(k\right)}(y^{1},\dots,y^{k})\,.\]
The wave front set of $H^{\otimes k}$ is given by, cf.~\citep[Thm.~8.2.9]{Hoermander2003},\[
\WF(H^{\otimes k})\subset\bigcup_{\sigma\in\Perm(k)}\bigcup_{\substack{i+j=k\\
i\geq1}
}\sigma\left\{ \WF(H)^{i}\times\left[\supp(H)\times\left\{ 0\right\} \right]^{j}\right\} \,,\]
where $\Perm(k)$ denotes the symmetric group in $k$ variables. Hence
all covectors in the $x$-components, $k_{x^{i}}$, are elements of
the closed forward lightcone $\overline{\lc^{+}}$, cf. (\ref{eq:HadamardWF}),
and equivalently for the $y$-components we have $\forall i$: $k_{y^{i}}\in\overline{\lc^{-}}$.
Using \citep[Thm.~8.2.10]{Hoermander2003} we conclude that the pointwise
product of distributions above is well-defined, if \[
\WF(F^{\left(k\right)})\cap\left(\overline{\lc^{-}}\right)^{k}=\emptyset=\WF(G^{\left(k\right)})\cap\left(\overline{\lc^{+}}\right)^{k}\,,\]
which is automatically fulfilled for $F,G\in\fps{\mathcal{F}(\M)}{\hbar}$.

Associativity of $\star$ follows directly from the associativity
of the pointwise product $M\left(F\otimes G\right)=F\cdot G$ and
the form of (\ref{eq:StarProduct}). The detailed argument can be
found in \citep[Sec. 6.2.4]{Waldmann2007}; it is summarized in the
following. Let\[
\Gamma_{H}^{\left(1,2\right)}\left(A\otimes B\otimes C\right):=\Gamma_{H}\left(A\otimes B\right)\otimes C\,,\]
\[
\Gamma_{H}^{\left(1,3\right)}\left(A\otimes B\otimes C\right):=\left\langle H,A^{\left(1\right)}\otimes B\otimes C^{\left(1\right)}\right\rangle \,,\]
and\[
\Gamma_{H}^{\left(2,3\right)}\left(A\otimes B\otimes C\right):=A\otimes\Gamma_{H}\left(B\otimes C\right)\,.\]
Then the Leibniz rule in the second argument implies\[
\Gamma_{H}\circ\left[\id\otimes M\right]=\left[\id\otimes M\right]\circ\left(\Gamma_{H}^{\left(1,2\right)}+\Gamma_{H}^{\left(1,3\right)}\right)\,,\]
and analogously\[
\Gamma_{H}\circ\left[M\otimes\id\right]=\left[M\otimes\id\right]\circ\left(\Gamma_{H}^{\left(1,3\right)}+\Gamma_{H}^{\left(2,3\right)}\right)\,.\]
These formulas generalize to the exponential of the bidifferential
operators $\Gamma_{H}$ and $\Gamma_{H}^{\left(i,j\right)}$ by linearity
of the tensor product. Hence for the star product one gets\begin{align*}
A\star\left(B\star C\right) & =M\circ e^{\hbar\Gamma_{H}}\circ\left(\id\otimes M\circ e^{\hbar\Gamma_{H}}\right)\left(A\otimes B\otimes C\right)\\
 & =M\circ e^{\hbar\Gamma_{H}}\circ\left(\id\otimes M\right)\circ e^{\hbar\Gamma_{H}^{\left(2,3\right)}}\left(A\otimes B\otimes C\right)\\
 & =M\circ\left(\id\otimes M\right)\circ e^{\hbar\left(\Gamma_{H}^{\left(1,2\right)}+\Gamma_{H}^{\left(1,3\right)}+\Gamma_{H}^{\left(2,3\right)}\right)}\left(A\otimes B\otimes C\right)\\
 & =M\circ\left(M\otimes\id\right)\circ e^{\hbar\left(\Gamma_{H}^{\left(1,2\right)}+\Gamma_{H}^{\left(1,3\right)}+\Gamma_{H}^{\left(2,3\right)}\right)}\left(A\otimes B\otimes C\right)\\
 & =\left(A\star B\right)\star C\end{align*}
by associativity of the pointwise product $M$.
\end{proof}
An especially important subclass of $\left(\fps{\mathcal{F}(\M)}{\hbar},\star\right)$
for the description of interactions in QFT is the class of local functionals. 
\begin{defn}
[Local Functional]\label{def:LocalFunctional}Let $\Diag(\M^{n})=\left\{ \vec{x}\in\M^{n}:x_{1}=\cdots=x_{n}\right\} $
denote the thin diagonal in $\M^{n}$. A functional $F\in\mathcal{F}(\M)$
is called local, if for all $n\in\mathbb{N}$:
\begin{itemize}
\item [{[LF-1]}]the $n$th order functional derivative is supported on
the thin diagonal,\[
\forall n\in\mathbb{N}:\quad\supp(F^{\left(n\right)}(\varphi))\subset\Diag(\M^{n})\,,\]
and
\item [{[LF-2]}]\hypertarget{LF-2}the wave front set of the $n$th order
functional derivative lies transversal to the thin diagonal,\[
\forall n\in\mathbb{N}:\quad\WF(F^{n}(\varphi))\subset\left[T\Diag(\M^{n})\right]^{\perp}\,.\]

\end{itemize}
We denote the space of local functionals by $\mathcal{F}_{\loc}(\M)$.

\end{defn}
The first two functionals in (\ref{eq:FunctionalsExamples}) are examples
of local functionals, the third is not local. Observe that \[
\left[T\Diag(\M^{n})\right]^{\perp}\cap\left(\M^{n}\times\left[\left(\overline{\lc^{+}}\right)^{n}\cup\left(\overline{\lc^{-}}\right)^{n}\right]\right)=\emptyset\,,\]
and hence the wave front set condition for local functionals \hyperlink{LF-2}{[LF-2]}
implies the wave front set condition of Definition~\ref{def:DeformableAlgebra};
local functionals are deformable. However, they do not form a subalgebra
of $\left(\fps{\mathcal{F}(\M)}{\hbar},\star\right)$, since the product
of two local functionals is not local in the generic case. Furthermore
we want to remark that, in contrast to the definition of deformability,
the definition of locality does not depend on the underlying Minkowski
signature, a fact which is a major ingredient in the Euclidean formulation
of Epstein-Glaser renormalization in \citep{Keller2009}.

From the viewpoint of microlocal analysis the wave front set condition\linebreak{}
{\hyperlink{LF-2}{[LF-2]}} implies that, if $F\in\mathcal{F}_{\loc}(\M)$,
the distributions $F^{\left(n\right)}(\varphi)\in\E'(\M^{n})$ can
be pulled back to surfaces which lie transversal to the thin diagonal\linebreak[4]
\citep[Thm.~8.2.4]{Hoermander2003}. This implies that their distributional
part depends only on relative coordinates, i.e. the Schwartz kernel
of the functional derivative of any local functional can be written
as\begin{equation}
F^{\left(n\right)}(\varphi)(x_{1},\dots,x_{n})=\sum_{k}f_{\varphi}^{n,k}(x)\, P_{k}(\partial_{\bld r})\delta(\bld r)\,,\qquad f_{\varphi}^{n,k}\in\D(\Diag(\M^{n}))\cong\D(\M)\label{eq:LocalFunctionalKernel}\end{equation}
where $x=\frac{1}{n}\sum_{i=1}^{n}x_{i}$ is the {}``center of mass''-coordinate,
$\bld r=\left(r_{1},\dots,r_{n-1}\right)$ are relative coordinates,
and $P_{k}$ are homogeneous, symmetric polynomials in $n-1$ spacetime
variables. In flat spacetime the thin diagonal is the coordinate space
of the center of mass, and the relative coordinates are defined in
a transversal surface. In this sense \hyperlink{LF-2}{[LF-2]} is
a microlocal remnant of translation invariance.

The causal partition of unity, which was used in \citep{Brunetti2000}
for the distributional construction of time-ordered products in curved
spacetime, can generally not be used in the functional framework introduced
here. It is replaced by the following result.
\begin{lem}
[{\citep[Lem.~3.2]{Brunetti2009}}]\label{lem:LocFuncSupportDecomposition}Any
local functional can be written as a finite sum of local functionals
of arbitrarily small supports. That is\begin{equation}
F=\sum_{i}\sigma_{i}F_{i}\,,\quad\sigma_{i}\in\left\{ \pm1\right\} ;\, F,F_{i}\in\fps{\mathcal{F}_{\loc}(\M)}{\hbar}\,,\label{eq:SumLocFuncArbSmallSupp}\end{equation}
where $\supp(F_{i})\subset B_{i}$, with $B_{i}$ a ball of arbitrarily
small radius $\varepsilon>0$.
\end{lem}
The proof of this Lemma in the given reference uses a different definition
of local functionals than the one given above. However, both definitions
can be shown to be equivalent; see \citep{BrunettiFredenhagenRibeiro2009}
and \citep[App.~C]{Keller2009}.

Combining (\ref{eq:LocalFunctionalKernel}) with (\ref{eq:SumLocFuncArbSmallSupp}),
we have that the functional derivative of any local functional can
be expressed as\begin{equation}
F^{\left(n\right)}(\varphi)(x_{1},\dots,x_{n})=\sum_{k}\sum_{i}f_{\varphi}^{n,k,i}(x)\, P_{k}(\partial_{\bld r})\delta(\bld r)\,,\quad F\in\fps{\mathcal{F}_{\loc}(\M)}{\hbar}\,,\label{eq:LocFuncKernelArbSmallSupp}\end{equation}
where the support of $f_{\varphi}^{n,k,i}\in\D(\M)$ can be chosen
arbitrarily small.

\section{\label{sec:TimeOrderedProduct}The Time Ordered Product}

We want to regard the situation where the time evolution of the interacting
theory is solved perturbatively, i.e., in terms of formal power series.
A convenient way to describe this evolution is by introduction of
a time ordering prescription for the interaction functionals. Given
the theory can be described as a free theory at asymptotic times \citep{Haag1958,Ruelle1962},
one can encode the information on the transition probabilities in
collision processes of elementary particles in the so-called $\Sm$-matrix
\citep{LSZ1955,LSZ1957},\begin{equation}
\Sm(F)=\exp_{\dT}(F)\,.\label{eq:S-Matrix}\end{equation}
We will not consider the problem of defining the $\Sm$-matrix for
non-compactly supported interaction functionals, commonly referred
to as the infrared problem (IR-problem). Instead we will stay within
the algebra $\left(\fps{\mathcal{F}(\M)}{\hbar},\star\right)$ and
discuss the definition of $\Sm$ as a map\[
\Sm:\fps{\mathcal{F}_{\loc}(\M)}{\hbar}\rightarrow\fps{\mathcal{F}(\M)}{\hbar}\]
defined on local functionals, i.e., the ultraviolet problem (UV-problem).
The observation, which lies at the very heart of (perturbative) renormalization
theory is that this map cannot be defined in a unique way, but that
such a definition necessarily introduces a freedom into the theory,
commonly described in terms of the Stückelberg-Petermann renormalization
group \citep{Stueckelberg1953}. We will come to that point in greater
detail in Chapter~\ref{cha:ForestFormula}.

In this section we will define the time ordered product and introduce
the notion of its partial algebra of functionals. In the framework
of pAQFT the time ordered product, despite its significantly different
properties, can be introduced in much the same way as the star product
was defined in the last section. Regard the second order functional
differential operator\begin{equation}
\Gamma_{H_{F}}':=\frac{1}{2}\int dx\, dy\, H_{F}(x,y)\,\frac{\delta^{2}}{\delta\varphi(x)\,\delta\varphi(y)}\,,\label{eq:BiDifferentialTProduct}\end{equation}
where $H_{F}\in\D'(\M^{2})$ denotes a Feynman propagator. That is,
there is a fundamental solution of the Klein-Gordon operator, $H_{F}\in\D'(\M)$,\begin{equation}
\left(\Box+m^{2}\right)H_{F}=\delta\,,\label{eq:KGFundamentalSolution}\end{equation}
such that $\forall f,g\in\D(\M)$:\begin{equation}
H_{F}(f,g)=\left\langle f,H_{F}*g\right\rangle \quad\mbox{and}\quad H_{F}(f,g)=H_{F}(g,f)\,,\label{eq:FeynmanLikePropagator}\end{equation}
and the wave front set of the propagator $H_{F}$ is given by\begin{align}
\WF(H_{F})=\left\{ \left(x,y,k,-k\right)\in\dot{T}^{*}(\M^{2})|\right. & \mbox{for }x\neq y:\,\left(x-y\right)^{2}=0,\, k\Vert(x-y),\nonumber \\
 & \phantom{\mbox{for }x\neq y:}\quad k\in\partial\lc^{\pm}\mbox{ if }x\in\partial\lc_{y}^{\pm};\label{eq:FeynmanWF}\\
 & \mbox{for }x=y:\, k\in\dot{T}_{x}^{*}(\M)\left.\vphantom{\left(x,y,k,-k\right)\in T^{*}(\M^{2})|}\right\} ,\nonumber \end{align}
where we understand $\overline{\lc_{y}^{\pm}}$ as the causal future
/ past of $y\in\M$. Feynman propagators, $H_{F}\in\D(\M^{2})$, have
to be carefully distinguished from the \emph{two point functions},
which are (bi-) solutions of the Klein-Gordon equation (\ref{eq:HadamardKleinGordonSolution})
and have positive frequency wave front sets (\ref{eq:HadamardWF}).

We have that $u=H_{F}(f,\cdot)=H_{F}(\cdot,f)$ is a solution of the
inhomogeneous equation \[
\left(\Box+m^{2}\right)u=f\,,\]
what already implies that the wave front set of $H_{F}$ necessarily
contains $\WF(\delta)$, or more generally \citep[(8.1.11), Thm.~8.3.1]{Hoermander2003},\begin{equation}
\WF(\delta)=\WF((\Box+m^{2})H_{F})\subset\WF(H_{F})\subset\WF(\delta)\cup\Char(\Box+m^{2})\,,\label{eq:FundSolKleinGordWF}\end{equation}
and thus a definition of powers of $H_{F}$ which uses only wave front
set properties is impossible, cf. \citep[Thm.~8.2.10]{Hoermander2003}.
Such a definition will involve an extension procedure as described
in the previous chapter, i.e., renormalization, and it will be helpful
to note that one can read off the scaling degree of $H_{F}$ directly
from (\ref{eq:KGFundamentalSolution}), cf. Lemma~\ref{lem:ScalingDegreeProperties},\begin{equation}
\sd(H_{F})=d-2\,,\quad d=\dim(\M)\,.\label{eq:FeynmanScalingDegree}\end{equation}

As in the case of the star product, there is a freedom in the definition
of $H_{F}$. Observe, however, that $H_{F}$ is fixed once we have
chosen a Hadamard function $H$ which determines the star product,\begin{equation}
H_{F}=H+i\Delta_{\adv}\,.\label{eq:RelationFeynmanPTwoPointF}\end{equation}
This is also reflected by the fact that $H$ and $H_{F}$ both can
be defined as certain boundary values of the same analytic function.
We will exploit the freedom in the choice of the pair $H,H_{F}$ in
Chapter~\ref{cha:DimRegHadamard} for the construction of a propagator
which is especially well-suited for the discussion of dimensional
regularization in position space. That is, we will add a smooth, symmetric,
Lorentz invariant solution of the Klein-Gordon equation to the Wightman
function, such that the so defined Hadamard function and corresponding
Feynman propagator will have desirable additional properties. However,
let us now come to the definition of the time-ordered product and
its partial algebra.
\begin{prop}
[Partial Algebra of $\cdot_{\Time}$]\label{pro:PartialAlgebraTimeProd}Let
$H_{F}\in\D'(\M^{2})$ be a Feynman propagator in the sense of (\ref{eq:KGFundamentalSolution})-(\ref{eq:FeynmanWF})
and let $\Gamma_{H_{F}}$ be defined by (\ref{eq:BiDifferentialTProduct}).
Then the product induced by\\
\begin{minipage}[c]{0.57\columnwidth}%
\begin{equation}
\xymatrix{\fps{\mathcal{F}(\M)}{\hbar}^{\otimes2}\ar[r]^{\Time^{\otimes2}}\ar[d]_{\cdot} & \fps{\mathcal{F}(\M)}{\hbar}^{\otimes2}\ar@{-->}[d]^{\dT}\\
\fps{\mathcal{F}(\M)}{\hbar}\ar[r]^{\Time} & \fps{\mathcal{F}(\M)}{\hbar}}
\label{eq:TimeOrderedProduct}\end{equation}
\end{minipage}%
\begin{minipage}[c]{0.39\columnwidth}%
\[
F\dT G=\sum_{k=0}^{\infty}\frac{\hbar^{k}}{k!}\left\langle F^{\left(k\right)},H_{F}^{\otimes k}G^{\left(k\right)}\right\rangle \,,\]
\end{minipage}\\
where\[
\Time:=\exp(\hbar\Gamma_{H_{F}}')\qquad\left(\Time^{-1}:=\exp(-\hbar\Gamma_{H_{F}}')\right)\]
denotes the time-ordering (anti-time-ordering) operator, makes $\left(\fps{\mathcal{F}(\M)}{\hbar},\dT\right)$
into a partial algebra. That is $F\dT G$ is defined for all pairs
$F,G\in\fps{\mathcal{F}(\M)}{\hbar}$ with \[
\supp(F)\cap\supp(G)=\emptyset\,,\]
and $\dT$ is associative for any three functionals with pairwise
disjoint supports.\end{prop}
\begin{rem}
\label{rem:TimeStar-and-Tadpoles} First. Observe that equation (\ref{eq:RelationFeynmanPTwoPointF})
implies that $\cdot_{\Time}$ really is the time-ordered product for
$\star$. We have for the scalar field\[
\varphi(f)\cdot_{\Time}\varphi(g)=\varphi(f)\varphi(g)+\hbar\left\langle f,H_{F}\, g\right\rangle \]
\[
\varphi(f)\star\varphi(g)=\varphi(f)\varphi(g)+\hbar\left\langle f,H\, g\right\rangle \,.\]
Assume that the support of $f$ is \emph{later than} the support of
$g$, $\supp(f)\gtrsim\supp(g)$, i.e., $\supp(f)$ and $\supp(g)$
can be separated by a Cauchy surface $\Sigma$, such that $\supp(f)$
lies in the future and $\supp(g)$ in the past of $\Sigma$. Then
we infer from $\supp(\Delta_{\adv})\subset\overline{\lc^{-}}$ that
$\left\langle f,\Delta_{\adv}*g\right\rangle =0$ and hence\[
\left\langle f,H_{F}\, g\right\rangle =\left\langle f,H\, g\right\rangle \,.\]
We thus see that \[
\varphi(f)\cdot_{\Time}\varphi(g)=\varphi(f)\star\varphi(g)\qquad\mbox{if }\quad\supp(f)\gtrsim\supp(g)\,.\]

Second. The action of the second order differential operator $\Gamma_{H_{F}}'$
on functionals can be interpreted directly in terms of graphs. Let
$F,G\in\mathcal{F}_{\loc}(\M)$ be interaction functionals, then the
Leibniz rule implies\begin{align}
\Gamma_{H_{F}}'\left(F\cdot G\right) & =\tfrac{1}{2}\left\langle H_{F},F^{\left(2\right)}\right\rangle G+\tfrac{1}{2}F\left\langle H_{F},G^{\left(2\right)}\right\rangle +\left\langle H_{F},F^{\left(1\right)}\otimes G^{\left(1\right)}\right\rangle \label{eq:TadpoleLeibnizRule}\\
 & =\tfrac{1}{2}\FoneFG+\tfrac{1}{2}\hspace{-1em}\FGoneG+\FoneG\nonumber \end{align}
The first two terms in this sum are tadpoles, i.e., graphs with lines
connecting one vertex with itself. Defining the time-ordered product
as a deformation of the pointwise product (through $\Time=e^{\hbar\Gamma_{H_{F}}'}$)
removes all tadpole terms. At low orders in $\hbar$ this can be seen
by a simple computation,\[
\left(1+\hbar\Gamma_{H_{F}}'\right)\left[\left(1-\hbar\Gamma_{H_{F}}'\right)F\cdot\left(1-\hbar\Gamma_{H_{F}}'\right)G\right]=\FG+\hbar\FoneG+\mathcal{O}(\hbar^{2}).\]
See the proof below for the general case.\end{rem}
\begin{proof}
[{Proof of Proposition~\ref{pro:PartialAlgebraTimeProd}}]Before
proving the main part of the proposition, we want to show that the
diagram in (\ref{eq:TimeOrderedProduct}) is equivalent to the given
formula. This is the same as showing that there are no tadpole terms
in the graphical expansion of $F\dT G$ and higher oder products.
The Leibniz rule (\ref{eq:TadpoleLeibnizRule}) can be written as
a coproduct rule for $\Gamma_{H_{F}}'$~,\[
\Delta\Gamma_{H_{F}}'=\Gamma_{H_{F}}'\otimes\id+\id\otimes\Gamma_{H_{F}}'+\Gamma_{H_{F}}\,.\]
Writing $M$ for pointwise multiplication and abbreviating $\Gamma_{H_{F}}\equiv\Gamma$,
we read off the diagram\begin{align*}
F\dT G & =e^{\hbar\Gamma'}\circ M\left(e^{-\hbar\Gamma'}F\otimes e^{-\hbar\Gamma'}G\right)\\
 & =M\circ e^{\hbar\Delta\Gamma'}\left(e^{-\hbar\Gamma'}F\otimes e^{-\hbar\Gamma'}G\right)\\
 & =M\circ e^{\hbar\Gamma}\left(F\otimes G\right)=\sum_{k=0}^{\infty}\frac{\hbar^{k}}{k!}\left\langle F^{\left(k\right)},H_{F}^{\otimes k}G^{\left(k\right)}\right\rangle \,.\end{align*}
And we see that the tadpoles drop out of the expansion. The graph
expansion of the time-ordered will be discussed in more detail in
Section~\ref{sec:GraphStructureT}.

One infers from (\ref{eq:RelationFeynmanPTwoPointF}) that \begin{equation}
F\dT G=\begin{cases}
F\star G & \mbox{if }\supp(F)\gtrsim\supp(G)\\
G\star F & \mbox{if }\supp(G)\gtrsim\supp(F)\,,\end{cases}\label{eq:CausalityProducts}\end{equation}
where, as before, $\mathcal{O}\gtrsim\mathcal{U}$ for two regions
$\mathcal{O},\mathcal{U}\subset\M$ denotes that $\mathcal{O}$ is
later than $\mathcal{U}$. Observe that if $\mathcal{O}\gtrsim\mathcal{U}$
and $\mathcal{U}\gtrsim\mathcal{O}$ then $\mathcal{\mathcal{O}}$
and $\mathcal{U}$ are \emph{causally disjoint}. With (\ref{eq:CausalityProducts})
it follows from Proposition~\ref{pro:DeformedAlgebra} that $F\dT G$
is well-defined as long as the functionals have disjoint supports.
Associativity follows from the same proposition, but can also be proven
directly. Let $A,B,C\in\fps{\mathcal{F}(\M)}{\hbar}$ be three deformable
functionals with pairwise disjoint supports, then the twofold product
is defined, and we have \begin{align*}
A\dT\left(B\dT C\right) & =\Time\left[\Time^{-1}A\cdot\Time^{-1}\Time\left(\Time^{-1}B\cdot\Time^{-1}C\right)\right]\\
 & =\Time\left[\Time^{-1}A\cdot\Time^{-1}B\cdot\Time^{-1}C\right]=\left(A\dT B\right)\dT C\qedhere\end{align*}

\end{proof}
We want to remark that there is a subalgebra $\fps{\mathcal{F}_{0}(\M)}{\hbar}\subset\fps{\mathcal{F}(\M)}{\hbar}$,
where besides the pointwise, $\cdot$ , and the star, $\star$ , also
the time-ordered product, $\dT$ , can be defined as a full product.
Namely, $\mathcal{F}_{0}(\M)$ is given as the algebra of functionals,
such that for any element all functional derivatives are smooth, compactly
supported functions,\[
\forall F\in\mathcal{F}_{0}(\M),\,\forall n\in\mathbb{N}:\quad F^{\left(n\right)}(\varphi)\in\D(\M^{n})\,.\]
The third example given in (\ref{eq:FunctionalsExamples}) is an element
of $\mathcal{F}_{0}(\M)$. The field equation, $\left(\Box+m^{2}\right)\varphi=0$,
generates an ideal in $\left(\mathcal{F}_{0}(\M),\star\right)$,\[
\mathcal{J}=\left\{ F\in\mathcal{F}_{0}(\M):F(\varphi)=\sum_{a}G_{a}(\varphi)\,\partial^{a}(\Box+m^{2})\varphi,\, G_{a}\in\mathcal{F}_{0}(\M)\right\} \,,\]
$a\in\mathbb{N}_{0}^{d}$, $d=\dim(\M)$. And following \citep[Footnote~5]{DuetschFredenhagen2003}
we want to assume that $\mathcal{J}$ is the set of functionals, which
vanish on the space of (smooth) solutions of the Klein-Gordon equation.
Let $F\in\mathcal{J}$ and $G\in\mathcal{F}_{0}(\M)$, one easily
checks that $F\star G\in\mathcal{J}$. However, the same is not true
for the time-ordered product, $F\dT G\notin\mathcal{J}$, in the region
where $\dT$ and $\star$ are different, since $H_{F}$ (in contrast
to $H$) is not a solution of the Klein-Gordon equation. Since we
do not want to deal with this and related issues, it is more convenient
to work with fields, which are not required to satisfy the field equation,
so-called \emph{off-shell fields}. As shown in \citep{DuetschFredenhagen2003}
a restriction to the space of solutions is always possible. See also
\citep{BrouderDuetsch2008} for an explicit construction of the maps
involved.

\section{\label{sec:The-Renormalization-Problem}The Renormalization Problem}

Associativity of the time-ordered product makes it possible to speak
of $n$-fold products\[
\begin{array}{rccl}
\Time_{n}: & \fps{\mathcal{F}(\M)}{\hbar}^{\otimes n} & \rightarrow & \fps{\mathcal{F}(\M)}{\hbar}\\
 & F_{1}\otimes\cdots\otimes F_{n} & \mapsto & F_{1}\dT\cdots\dT F_{n}\,,\end{array}\]
which are well-defined if the supports of the functionals $F_{1},\dots,F_{n}\in\fps{\mathcal{F}(\M)}{\hbar}$
are pairwise disjoint,\[
\supp(F_{i})\cap\supp(F_{j})=\emptyset\quad\forall i,j\in\left\{ 1,\dots,n\right\} ,\, i\neq j\,.\]
The aim of perturbative QFT, however, is to define the terms of the
$\Sm$-matrix (\ref{eq:S-Matrix}), which are time ordered products
of the same interaction functional\[
\Sm(F)=\sum_{n=0}^{\infty}\frac{1}{n!}\Time_{n}(F\otimes\cdots\otimes F)\,.\]
Hence one has to extend the definition of $\Time_{n}$ towards functionals
with overlapping supports. In the present formalism such an extension
is only possible for local functionals $F\in\fps{\mathcal{F}_{\loc}(\M)}{\hbar}$.
One way to extend the maps $\Time_{n}$ to local functionals with
overlapping supports is the inductive procedure of Epstein and Glaser
\citep{Epstein1973}. See also \citep{Brunetti2000} and \citep{Brunetti2009}
for modern generalizations of the original treatment. We want to remark
that the recursive construction of Epstein-Glaser can be performed
without reference to the star-product structure of pQFT, and hence
is suitable also for a discussion of the renormalization problem on
Euclidean space, \citep{Keller2009}.

We will show in the last chapter that the inductive procedure of Epstein
and Glaser can be solved, by implementing analytic regularization
and minimal subtraction, which gives preferred extensions in each
step of the induction. An analytic regularization, which has been
proven to have especially nice properties when it comes to gauge theories,
is dimensional regularization \citep{BecchiRouetStora1975}. And the
next two chapters will be devoted to the implementation of this method
into the framework of perturbative Algebraic Quantum Field Theory,
although we will restrict ourselves to the study of scalar quantum
field theory, only.

\end{fmffile}

\global\long\def\fps#1#2{#1\![[#2]]}

\global\long\def\poisson#1#2{\left\lfloor #1,#2\right\rceil }

\global\long\def\bld#1{\boldsymbol{#1}}

\chapter{\label{cha:DimRegHadamard}The~Dimensionally~Regularized Analytic~Hadamard~Function}

Despite its rigorous formulation in \citep{BolliniGiambiagi1972a,tHooftVeltman1972}
dimensional regularization has always been a somewhat shady or almost
mystic concept, since the idea of a complex spacetime dimension is
quite obscure from a conceptual point of view. Hence, the insight
of Bollini and Giambiagi that dimensional regularization can be implemented
in position space as a modification of the Bessel parameter in the
two point function was, although passing largely unnoticed, an important
one for the mathematical physicist interested in a conceptually clear
formulation of perturbative quantum field theory \citep{BolliniGiambiagi1996}.
We will follow the detailed argument of \citep[App.~A]{Brunetti2009},
which, however, contains a small flaw, to show how a modification
of the Bessel parameter leads to the notion of a {}``dimensionally
regularized'' two point function for arbitrary, integer spacetime
dimensions. This two point function will then be used in the next
chapter to define a dimensionally regularized time-ordered product
and the corresponding $\Sm$-matrix.

In a series of articles Hollands and Wald developed a description
of the renormalization group flow on globally hyperbolic spacetimes
by investigating the behavior of renormalizable theories under rescalings
of the metric \citep{HollandsWald2001,HollandsWald2002,Hollands2003}.
A major ingredient of the approach is their {}``scaling expansion''
of time-ordered products around the thin diagonal. This expansion
has the property that the scaling degree of the individual terms becomes
smaller and smaller as one goes to higher and higher orders in the
expansion. In scalar QFT on Minkowski spacetime such a scaling expansion
can be introduced as an expansion in the mass parameter $m^{2}$ \citep{Hollands2004,DuetschFredenhagen2004}.
This requires, however, that the two point function depends smoothly
on $m^{2}$. The Wightman two point function, $\Delta_{+}$, in even
dimensions, exhibits a logarithmic dependence on the mass parameter,
and thus cannot be used in this framework. However, we can take advantage
of the freedom involved in the choice of a (Hadamard) two point function,
briefly discussed in Sections \ref{sec:AlgebraOfObservables} and
\ref{sec:TimeOrderedProduct}, and add to $\Delta_{+}^{m}$ a smooth,
symmetric, Lorentz invariant solution of the Klein-Gordon equation,
which establishes a smooth dependence on $m^{2}$ for the sum. In
flat spacetime the requirement of smoothness in mass fixes the two
point function uniquely in odd dimensions, and up to a parameter $\mu$
of mass dimension one in even dimensions \citep{DuetschFredenhagen2004}.
Hence, the algebra and the time-ordering are fixed (up to the parameter
$\mu$) by this smoothness condition.

\section{Odd Dimensions}

The Wightman distribution $\Delta_{+}^{m}$ on $d$-dimensional Minkowski
spacetime can be expressed, for spacelike arguments $x\in\M$ (i.e.,
$x^{2}<0$ in our choice of the metric) in terms of modified Bessel
functions (see e.g. \citep{BogoliubovShirkov1959}), \begin{equation}
\Delta_{+}^{m}(x)=\left(2\pi\right)^{1-\nu}m^{\nu}\left(-x^{2}\right)^{-\frac{\nu}{2}}K_{\nu}(\sqrt{-m^{2}x^{2}})\,,\quad\nu=\tfrac{d}{2}-1\,.\label{eq:WightmanTwoPointFunctionModBessel}\end{equation}
The right hand side of this equation is a function of one (real) variable
$x^{2}$, which is parametrized by the complex order $\nu$ of the
modified Bessel function. In the case where $\nu\in\frac{1}{2}\mathbb{N}_{0}$,
and only then, this function has the physical interpretation of the
Wightman two point function on a spacetime of dimension $d=2\left(\nu+1\right)$.
The Wightman two point function, as well as its generalizations for
arbitrary $\nu\in\mathbb{C}$ do not scale smoothly in $m^{2}$. However,
starting from (\ref{eq:WightmanTwoPointFunctionModBessel}) we can
construct a Hadamard two point function $H_{m}$, which scales smoothly
in $m^{2}\in\mathbb{R}$, by adding to $\Delta_{+}^{m}$ a smooth,
Lorentz invariant solution $F$ of the Klein-Gordon equation, $H_{m}=\Delta_{+}^{m}+F$.
Any such Lorentz invariant solution $F$ has the form\begin{equation}
F(x)=\left(-x^{2}\right)^{-\frac{\nu}{2}}G_{\nu}(\sqrt{-m^{2}x^{2}})\,,\label{eq:FormOfLorentzInvariantSolution}\end{equation}
for spacelike arguments $x$, where $G_{\nu}$ is a solution of the
modified Bessel equation of order $\nu$. For non-integer order, $\nu\in\mathbb{C}\backslash\mathbb{N}_{0}$,
(e.g., odd dimensions $d$) $G_{\nu}$ is a linear combination of
the modified Bessel functions of first kind $\left\{ I_{\nu},I_{-\nu}\right\} $.
For integer order, $n\in\mathbb{N}_{0}$, (i.e., even dimensions $d$)
it is a linear combination of $\left\{ I_{n},K_{n}\right\} $, where
$K_{n}$ is the modified Bessel function of second kind, see Appendix~\ref{app:ModifiedBesselFunctions}
for details.

Requiring smoothness at $x=0$ implies for arbitrary order $\nu\in\mathbb{C}$
that \begin{equation}
F(x)\sim\left(-x^{2}\right)^{-\frac{\nu}{2}}I_{\nu}(\sqrt{-m^{2}x^{2}})\,.\label{eq:FsimI}\end{equation}
For $\nu\in\mathbb{C}\backslash\mathbb{N}_{0}$ the modified Bessel
functions are related by\[
K_{\nu}=\frac{\pi}{2\sin(\nu\pi)}\left[I_{-\nu}-I_{\nu}\right]\,,\tag{\ref{eq:BesselFunctionsSecondKind}}\]
hence, using this together with (\ref{eq:FsimI}) we reach\begin{align}
H_{m}^{\nu}(x) & =\Delta_{+}^{m}(x)+F(x)\nonumber \\
 & \hspace{-5mm}=\!\left(2\pi\right)^{1-\nu}\! m^{\nu}\!\left(-x^{2}\right)^{-\frac{\nu}{2}}\!\left[K_{\nu}(\sqrt{-m^{2}x^{2}})+a\cdot I_{\nu}(\sqrt{-m^{2}x^{2}})\right]\nonumber \\
 & \hspace{-5mm}=\!\left(2\pi\right)^{1-\nu}\! m^{\nu}\!\left(-x^{2}\right)^{-\frac{\nu}{2}}\!\left[\tfrac{\pi}{2\sin(\nu\pi)}I_{-\nu}(\sqrt{-m^{2}x^{2}})\!+\!\left(a-\tfrac{\pi}{2\sin(\nu\pi)}\right)I_{\nu}(\sqrt{-m^{2}x^{2}})\right]\!,\label{eq:HadamardFunctionParametrized}\end{align}
where $a\in\mathbb{C}$ is a free parameter yet to be specified. In
order to fix the parameter $a$ we regard the scaling behavior in
$m^{2}$ of the two terms in (\ref{eq:HadamardFunctionParametrized}).
The (modified) Bessel functions are of the form \[
I_{\nu}(y)=y^{\nu}f_{\nu}(y^{2})\,.\]
with an entire analytic function $f_{\nu}$. Thus, in the first term
the factor $m^{-\nu}$ in $I_{-\nu}$ cancels with the prefactor $m^{\nu}$
leaving a smooth function of $m^{2}$ behind. In the second term we
get an overall factor $m^{2\nu}$, which is not a smooth function
of $m^{2}$, unless $\nu\in\mathbb{N}_{0}$ (which is excluded). Hence
the free parameter $a$ has to be chosen in such a way that the term
proportional to $I_{\nu}$ cancels, $a=\frac{\pi}{2\sin(\nu\pi)}$.
Summarizing the above, we have found, for non-integer order $\nu\in\mathbb{C}\backslash\mathbb{N}_{0}$,
and in particular for odd dimensions $d$, a \emph{unique} Hadamard
two point function, which depends smoothly on $m^{2}$, it is given
by\begin{equation}
H_{m}^{\nu}(x)=\frac{\left(2\pi\right)^{2-\nu}m^{\nu}}{4\sin(\nu\pi)}\left(-x^{2}\right)^{-\frac{\nu}{2}}I_{-\nu}(\sqrt{-m^{2}x^{2}})\,,\quad x^{2}<0\,,\quad\nu\in\mathbb{C}\backslash\mathbb{N}_{0}.\label{eq:HadamardFunctionUnique}\end{equation}

\section{Analytic continuation}

We have already seen that $H_{m}^{\nu}(x)$ scales smoothly in the
mass parameter $m^{2}$ for spacelike $x\in\M$, and want to discuss
now the analytic properties of $H_{m}^{\nu}$. It is a fundamental
result of complex analysis that the analytic continuation of $H_{m}^{\nu}$
is unique in the region where it exists, so let $\mathcal{H}_{m}^{\nu}:\M^{\mathbb{C}}\rightarrow\mathbb{C}$
be this continuation, defined as a function on the complexified Minkowski
space, $\M^{\mathbb{C}}:=\M\otimes_{\mathbb{R}}\mathbb{C}$. The modified
Bessel functions are defined for arbitrary complex arguments and writing
$I_{-\nu}(\sqrt{m^{2}z^{2}})=\left(\sqrt{m^{2}z^{2}}\right)^{-\nu}f_{\nu}(m^{2}z^{2})$,
with an entire analytic function $f_{\nu}$, we see that the analytic
continuation of (\ref{eq:HadamardFunctionUnique}) can be written
as \begin{equation}
\mathcal{H}_{m}^{\nu}(z)=\left(2\pi\right)^{1-\nu}\frac{\pi}{2\sin(\nu\pi)}\,\left(-z^{2}\right)^{-\nu}f_{\nu}(m^{2}\, z^{2})\,,\quad z\in\M^{\mathbb{C}},\label{eq:HadamardAnalytic}\end{equation}
from which the smoothness in $m^{2}$ is obvious. Using the series
representation of the modified Bessel function (\ref{eq:ModifiedBesselFirstKind})
we immediately get a series expansion of $\mathcal{H}_{m}^{\nu}$
in $m^{2}$,\begin{equation}
\mathcal{H}_{m}^{\nu}(z)=\frac{\pi^{2-\nu}}{\sin(\nu\pi)}\,\left(-z^{2}\right)^{-\nu}\sum_{s=0}^{\infty}\frac{1}{s!\,\Gamma(-\nu+s+1)}\left(\frac{-z^{2}}{4}\right)^{s}\left(m^{2}\right)^{s}.\label{eq:HadamardAnalyticSeriesRep}\end{equation}
Since this formula contains the power of a complex number, $\left(-z^{2}\right)^{-\nu}=e^{-\nu\,\Log(-z^{2})}$,
we have to choose a branch of the logarithm in order to make $\mathcal{H}_{m}^{\nu}$
single-valued. As we shall see below, choosing the principal branch,
$\Log$, of the complex logarithm,\[
\begin{array}{rccl}
\Log: & \mathbb{C}\backslash\left\{ 0\right\}  & \rightarrow & \mathbb{R}\oplus i\left(-\pi,\pi\right]\\
 & r\, e^{i\vartheta} & \mapsto & \Log(r\, e^{i\vartheta})=\ln(r)+i\vartheta\,,\quad\vartheta\in\left(-\pi,\pi\right]\cong\mathbb{R}/2\pi\,.\end{array}\]
gives a single-valued function $\mathcal{H}_{m}^{\nu}$ with the analytic
properties of the Wightman function. The principal branch has a discontinuity
along the negative real axis ($\vartheta=\pi$), resulting in the
fact that the function on the right hand side of (\ref{eq:HadamardAnalyticSeriesRep}),
regarded as a function of one complex variable \[
z^{2}=\left(x-iy\right)^{2}=x^{2}-y^{2}-2i\, xy\,,\]
i.e., $\mathcal{H}_{m}^{\nu}(z)=\h_{m}^{\nu}(z^{2})$, is analytic
in the cut plane $\mathbb{C}\backslash\mathbb{R}_{+}^{0}$, see Figure~\ref{fig:HadamardVisualized}(a).
The condition $z^{2}\notin\mathbb{R}_{+}^{0}$ is fulfilled if the
complex four vector%
\footnote{By abuse of terminology we will use the term (complex) {}``four vector''
for elements of (complex) Minkowski space of arbitrary dimension $d\geq2$.%
} $z\in\M^{\mathbb{C}}$ lies in the so-called \emph{future} or \emph{past
tube}\[
\tube^{\pm}=\left\{ x-iy:\, x\in\M,\, y\in\lc^{\pm}\right\} \subset\M^{\mathbb{C}}\,,\]
where, as before, $\lc^{\pm}=\left\{ y\in\M:\,\pm y^{0}>\left\Vert \bld y\right\Vert \right\} $
denote the open forward and backward light cones. Observe that $\Im(z^{2})=0$
implies $z^{2}<0$, if $z=x-iy\in\tube^{\pm}$: \[
x^{0}y^{0}=\bld x\cdot\bld y\leq\left\Vert \bld x\right\Vert \left\Vert \bld y\right\Vert <\pm\left\Vert \bld x\right\Vert y^{0}\,,\]
hence $x^{2}<0$ and $z^{2}<0$. We reach the conclusion that the
function $\mathcal{H}_{m}^{\nu}:\M^{\mathbb{C}}\rightarrow\mathbb{C}$,
given by (\ref{eq:HadamardAnalyticSeriesRep}), is analytic in the
future and past tubes $\tube^{\pm}$.%
\footnote{It is well-known that the analyticity domain of the (Hadamard) two
point function is bigger than just the future and past tubes. By Lorentz
invariance, the extended tube, and by permutation symmetry even the
so-called \emph{permuted extended tubes} are part of the analyticity
domain of the corresponding analytic $n$ point functions \citep{HallWightman1957}.
However, for our purposes it will suffice to consider the subsets
$\tube^{\pm}$ of the analyticity domain of $\mathcal{H}_{m}^{\nu}$.%
} Furthermore we see from the explicit formula (\ref{eq:HadamardAnalyticSeriesRep})
that the growth of $\left|\mathcal{H}_{m}^{\nu}(x-iy)\right|$ is
bounded by an inverse polynomial as $y$ approaches zero from within
the forward or backward light cone. Hence, by \citep[Thm.~3.1.15]{Hoermander2003},
the boundary values of $\mathcal{H}_{m}^{\nu}$ from inside the future
and past tube exist as distributions in $\D'(\M)$.

Hence we can define the Hadamard distribution $H_{m}^{\nu}\in\D'(\M)$
to be the boundary value of $\mathcal{H}_{m}^{\nu}$ as the real subspace
$\M\subset\M^{\mathbb{C}}$ is approached from the future tube, $y=\left(\tfrac{\varepsilon}{2},\bld 0\right)\in\lc^{+}$,\begin{equation}
\left\langle H_{m}^{\nu},f\right\rangle :=\lim_{\varepsilon\rightarrow0^{+}}\int_{\M}dx\, f(x)\,\h_{m}^{\nu}(x^{2}-ix^{0}\varepsilon)\,,\quad\nu\in\mathbb{C}\backslash\mathbb{N}_{0}\,.\label{eq:HadamardBoundaryValue}\end{equation}
The wave front set of $H_{m}^{\nu}\in\D'(\M)$ lies within the dual
cone of $\lc^{+}$, cf. \citep[Thm.~8.1.6]{Hoermander2003}, which
is the closed cone $\overline{\lc^{+}}$, and hence we have\[
\WF(H_{m}^{\nu})\subset\M\times\left(\overline{\lc^{+}}\backslash\left\{ 0\right\} \right)\,,\qquad\nu\in\mathbb{C}\backslash\mathbb{N}_{0}\,.\]
Thus we have found a parametrized Hadamard distribution, which can
be used to define a star product of functionals as described in the
previous chapter, cf. Proposition~\ref{pro:DeformedAlgebra}.

The corresponding Feynman fundamental solution can be defined in the
same way, it is the time-ordered version of $H_{m}^{\nu}$,\[
H_{F}^{m,\nu}(x):=\theta(x^{0})H_{m}^{\nu}(x)+\theta(-x^{0})H_{m}^{\nu}(-x)\,,\]
where $\theta$ is the Heaviside step function. Rephrased in the language
of complex analysis, for $x^{0}>0$ the Feynman fundamental solution
is the boundary value of the analytic Hadamard function $\mathcal{H}_{m}^{\nu}:\M^{\mathbb{C}}\rightarrow\mathbb{C}$
from inside the future tube, i.e. $H_{F}^{m,\nu}(x)=\lim_{\varepsilon\rightarrow0^{+}}\h_{m}^{\nu}(x^{2}-ix^{0}\varepsilon)$,
and for $x^{0}<0$ it is the boundary value of the same analytic function
from inside the past tube, $H_{F}^{m,\nu}(x)=\lim_{\varepsilon\rightarrow0^{+}}\h_{m}^{\nu}(x^{2}+ix^{0}\varepsilon)$.
Hence in both cases it results \begin{equation}
\forall f\in\D(\M\backslash\left\{ 0\right\} ):\,\left\langle H_{F}^{m,\nu},f\right\rangle =\lim_{\varepsilon\rightarrow0^{+}}\int dx\, f(x)\,\h_{m}^{\nu}(x^{2}-i\varepsilon)\,.\label{eq:FeynmanBoundaryValue}\end{equation}
The distribution $H_{F}^{m,\nu}$ is given as the boundary value of
the analytic two point function $\mathcal{H}_{m}^{\nu}$ from two
disjoint areas of its analyticity domain in the cases $x^{0}>0$ and
$x^{0}<0$. For $x^{0}=0$ and $\bld x\neq0$ we are in the analyticity
domain of $\mathcal{H}_{m}^{\nu}$. For $x=0$ a definition of $H_{F}^{m,\nu}$
as boundary value of $\mathcal{H}_{m}^{\nu}$ is not possible, hence
we have defined the Feynman propagator as a distribution $H_{F}^{m,\nu}\in\D'(\M\backslash\left\{ 0\right\} )$.
Observe, however, that $H_{F}^{m,\nu}$ has a unique extension to
$\D'(\M)$ with the same scaling degree. The scaling degree of $H_{F}^{m,\nu}$
can be read off directly from (\ref{eq:HadamardAnalyticSeriesRep}),\begin{equation}
\sd(H_{F}^{m,\nu})=2\nu\,.\label{eq:ScalingDegreeHFm-nu}\end{equation}
Hence for half-integer Bessel order $\nu=\frac{d}{2}-1$, we have
that $\sd(H_{F}^{m,d})=d-2$ and hence $H_{F}^{m,d}\in\D'(\M\backslash\left\{ 0\right\} )$
has a unique extension $\dot{H}_{F}^{m,d}\in\D'(\M)$. Observe that
the terms proportional to $\left(m^{2}\right)^{s}$ in the expansion
(\ref{eq:HadamardAnalyticSeriesRep}) are homogeneous of degree $D=2\left(s-\nu\right)$,
cf. Definition~\ref{def:HomogeneousDistribution}. Hence these terms
have unique homogeneous extensions for $2\left(\Re(\nu)-s\right)\in\mathbb{N}_{0}+d$,
cf. Remark~\ref{rem:ScalingDegreeHomogeneity} and Theorem~\ref{thm:HomogeneousExtension},
and in particular for $\nu\notin\frac{1}{2}\mathbb{N}_{0}$. This
observation is the basis for the discussion undertaken in the next
section and the following chapter.

Hence we have found a unique analytic Hadamard function $\mathcal{H}_{m}^{\nu}$,
which depends smoothly on $m^{2}$ for Bessel order $\nu\in\mathbb{C}\backslash\mathbb{N}_{0}$,
and hence in particular for odd dimensions. Before turning to the
more intricate case of even dimensions, let us briefly discuss the
properties of the analytic Hadamard function by visualizing $\mathcal{H}_{m}^{\nu}$
in the two pictures of Fig.~\ref{fig:HadamardVisualized}. In particular
observe that the Hadamard boundary value $H_{m}^{\nu}(x)$ grows exponentially
in spacelike directions, a fact which makes a direct comparison in
terms of the full time-ordered products of our formulation to the
well-established formulation of dimensional regularization in momentum
space difficult, if not impossible. A comparison of the counterterms,
as described briefly in Section~\ref{sec:RegularizationOfDistributions},
however, should be possible.

\begin{figure}[h]
\subfloat[]{\includegraphics[height=5cm]{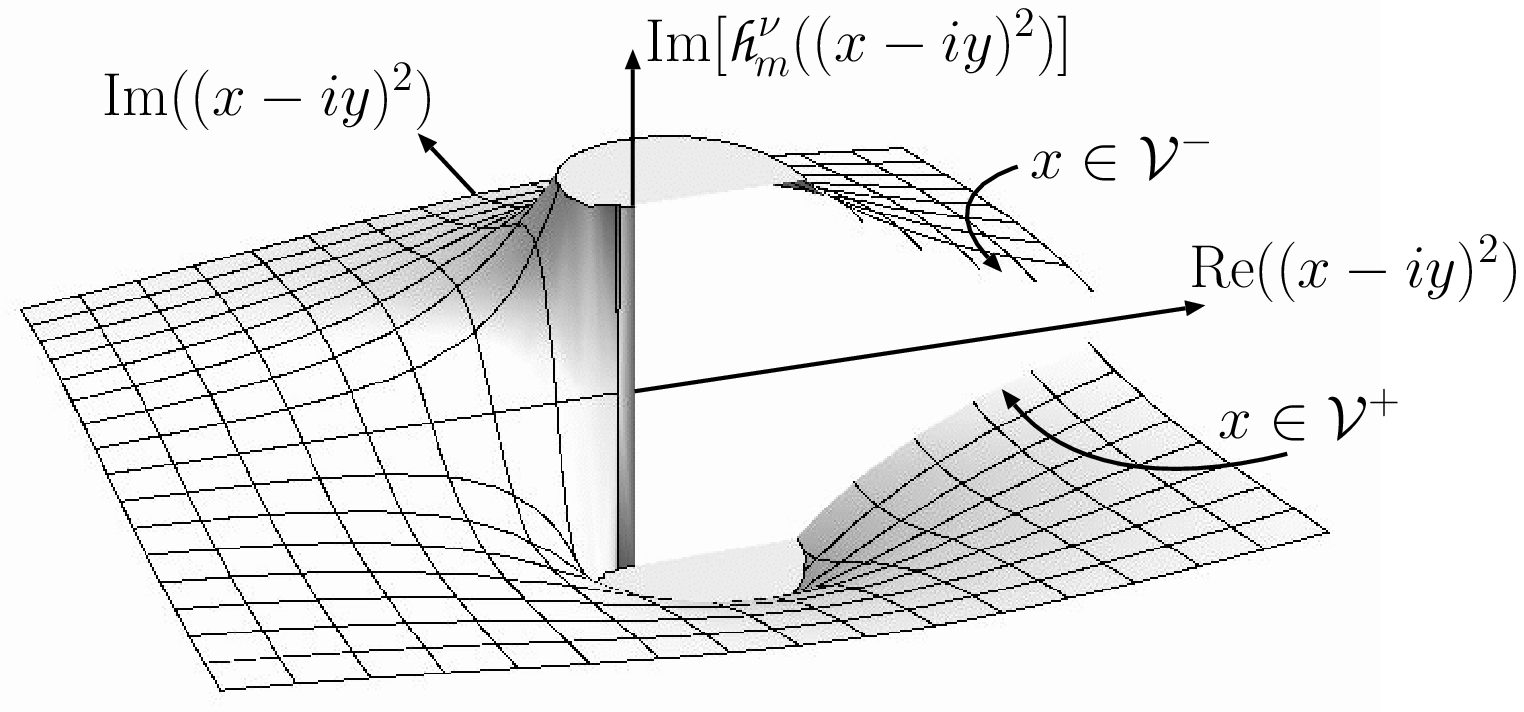}

}\hfill\subfloat[]{\includegraphics[height=5cm]{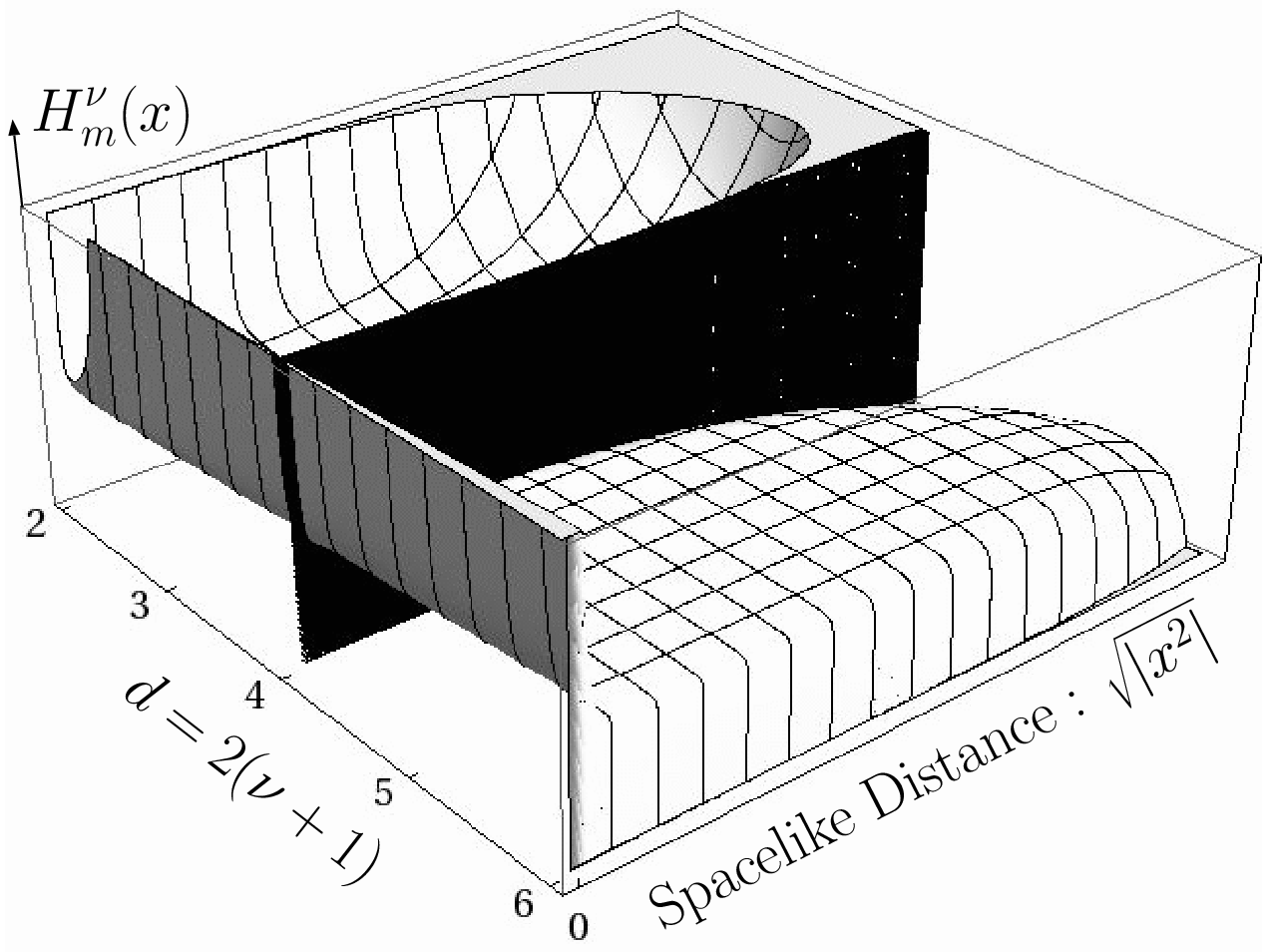}

}

\caption{\label{fig:HadamardVisualized}Generalizations of the Hadamard distribution.\protect \\
(a)~to complex arguments. $\h(\left(x-iy\right)^{2})$ is analytic
in the cut plane $\mathbb{C}\backslash\mathbb{R}_{+}^{0}$, which
implies that $\mathcal{H}_{m}^{\nu}(x-iy)$ is analytic in the future
and past tube $\tube^{\pm}$. The values of $H_{m}^{\nu}(x)$ for
timelike vectors $x$ are the boundary values of this function, as
$y$ approaches zero from inside the forward light cone. For the plot
$\nu=\frac{1}{2}$ and $m^{2}=1$ was chosen.\protect \\
(b)~to {}``complex dimensions''. The plot shows the qualitative
behavior of the Hadamard function $H_{m}^{\nu}(x)$ for spacelike
distances, $x^{2}<0$, in dependence of the parameter $\nu\in\mathbb{C}$.
One sees the (simple) poles at integer values $\nu=\frac{d}{2}-1\in\mathbb{N}_{0}$,
cf. (\ref{eq:HadamardAnalyticSeriesRep}), and that the local singularity
structure at $x=0$ does not change as $\nu$ varies. Observe also
the (alternating) behavior for large spacelike distances, the absolute
value of $H_{m}^{\nu}$ grows exponentially, hence it cannot be the
kernel of a Schwartz distribution, cf. also (\ref{eq:HadamardAnalyticSeriesRep}).
In the plot $\nu$ varies over the reals from $0$ to $2$. We have
chosen $m^{2}=1$ for the plot.}

\end{figure}

\section{Even dimensions}

The fact that (\ref{eq:HadamardFunctionUnique}) or respectively (\ref{eq:HadamardAnalyticSeriesRep})
fixes the analytic two point function uniquely for any complex parameter
$\nu\in\mathbb{C}\backslash\mathbb{N}_{0}$ suggests to construct
the corresponding two point function for $\nu\in\mathbb{N}_{0}$,
i.e., for even dimensions $d=2\left(\nu+1\right)$, by a limiting
procedure. We introduce a regularization parameter $\zeta\in\mathbb{C}$
and set \[
\nu\mapsto\frac{d+\zeta}{2}-1\,,\quad d\in2\mathbb{N},\quad0<\left|\zeta\right|<2\,,\]
in (\ref{eq:HadamardFunctionUnique}). The resulting parametrization
of $\mathcal{H}_{m}^{\nu}$ we denote by \begin{align}
\widetilde{\mathcal{H}}_{m}^{\mu,\zeta}(z) & :=\frac{\left(2\pi\right)^{1-\frac{\left(d+\zeta\right)}{2}}\mu^{-\zeta}m^{\frac{\left(d+\zeta\right)}{2}-1}}{4\sin(\left(\frac{\left(d+\zeta\right)}{2}-1\right)\pi)}\left(-z^{2}\right)^{\frac{2-\left(d+\zeta\right)}{4}}I_{1-\frac{\left(d+\zeta\right)}{2}}(\sqrt{-m^{2}z^{2}})\label{eq:DimRegHadamard}\\
 & \hspace{-10mm}=\left(-1\right)^{\frac{d}{2}-1}\left(\tfrac{1}{2\pi}\right)^{\frac{d}{2}}m^{d-2}\tfrac{\pi}{2\sin(\zeta\frac{\pi}{2})}\left(\tfrac{m}{\mu}\right)^{\zeta}\left(\sqrt{-m^{2}z^{2}}\right)^{1-\frac{\left(d+\zeta\right)}{2}}I_{1-\frac{\left(d+\zeta\right)}{2}}(\sqrt{-m^{2}z^{2}})\,.\nonumber \end{align}
We had to introduce a parameter $\mu$ of mass dimension one, in
order to get the right mass dimension for the two point function,
$\md(\widetilde{\mathcal{H}}_{m}^{\mu,\zeta})=d-2$. Apart from the
dependence on the free parameter $\mu$, $\widetilde{\mathcal{H}}_{m}^{\mu,\zeta}$
has for $0<\left|\zeta\right|<2$ the same properties as $\mathcal{H}_{m}^{\nu}$
for $\left|\nu\right|$ between two integers. In particular $\widetilde{\mathcal{H}}_{m}^{\mu,\zeta}$
is analytic in future and past tube $\tube^{\pm}$ and depends smoothly
on the mass parameter $m^{2}$.

In order to get an expression for even dimensions, we want to perform
the limit $\zeta\rightarrow0$. This limit can not be performed directly,
since $\widetilde{\mathcal{H}}_{m}^{\mu,\zeta}$ diverges as the parameter
$\zeta$ tends to zero. The aim of this section is to construct a
dimensionally regularized analytic two point function $\mathcal{H}_{m}^{\mu,\zeta}$
(without the tilde), which is a smooth function of $m^{2}$, and solves
the Klein-Gordon equation in the limit $\zeta\rightarrow0$. This
is done by exploiting the freedom in the choice of the Hadamard distribution
$H$ mentioned earlier. We will see that $\mathcal{H}_{m}^{\mu,\zeta}$
differs from $\widetilde{\mathcal{H}}_{m}^{\mu,\zeta}$ by an analytic,
Lorentz invariant function, which is a smooth function of $m^{2}$
and solves the Klein-Gordon equation {}``in $d+\zeta$ dimensions''.
The exact meaning of this assertion will become clear in the construction
to be carried out now.

For $0<\left|\zeta\right|<2$ we can express $\widetilde{\mathcal{H}}_{m}^{\mu,\zeta}$
in terms of the modified Bessel functions, cf. (\ref{eq:BesselFunctionsSecondKind}),
\begin{align}
\widetilde{\mathcal{H}}_{m}^{\mu,\zeta}(z) & =\left(2\pi\right)^{-\frac{\left(d+\zeta\right)}{2}}\mu^{-\zeta}m^{\frac{\left(d+\zeta\right)}{2}-1}\left(-z^{2}\right)^{\frac{2-\left(d+\zeta\right)}{4}}\cdot\label{eq:DimRegHadamardMS1}\\
 & \hspace{10mm}\cdot\left[K_{\frac{\left(d+\zeta\right)}{2}-1}(\sqrt{-m^{2}z^{2}})+\tfrac{\pi}{2\sin(\left(\frac{\left(d+\zeta\right)}{2}-1\right)\pi)}I_{\frac{\left(d+\zeta\right)}{2}-1}(\sqrt{-m^{2}z^{2}})\right]\nonumber \\
 & =\mathcal{W}_{m}^{\mu,\zeta}(z)+\widetilde{\mathcal{B}}_{m}^{\mu,\zeta}(z)\,,\nonumber \end{align}
where we set $\mathcal{W}_{m}^{\mu,\zeta}$ for the term proportional
to $K_{\frac{\left(d+\zeta\right)}{2}-1}$ and $\widetilde{\mathcal{B}}_{m}^{\mu,\zeta}$
for the $I_{\frac{\left(d+\zeta\right)}{2}-1}$-term. The first term
in (\ref{eq:DimRegHadamardMS1}), $\mathcal{W}_{m}^{\mu,\zeta}(z)$,
is well defined for $\zeta=0$, it is (the analytic continuation of)
the Wightman function (\ref{eq:WightmanTwoPointFunctionModBessel})
in $d$ dimensions. The second term, \begin{align*}
 & \widetilde{\mathcal{B}}_{m}^{\mu,\zeta}(z)=\\
 & \qquad\left(-1\right)^{\left(\frac{d}{2}-1\right)}\left(2\pi\right)^{-\frac{\left(d+\zeta\right)}{2}}\mu^{-\zeta}m^{\frac{\left(d+\zeta\right)}{2}-1}\left(-z^{2}\right)^{\frac{2-\left(d+\zeta\right)}{4}}\cdot\tfrac{\pi}{2\sin(\zeta\frac{\pi}{2})}I_{\frac{\left(d+\zeta\right)}{2}-1}(\sqrt{-m^{2}z^{2}})\,,\end{align*}
is a meromorphic function in $\zeta$. More precisely, for any fixed
$d\in2\mathbb{N}$ the map $\zeta\mapsto\widetilde{\mathcal{B}}_{m}^{\mu,\zeta}$
is analytic in the punctured disk $\left\{ 0<\left|\zeta\right|<2\right\} $
and has a simple pole at $\zeta=0$. That is, $\zeta\widetilde{\mathcal{B}}_{m}^{\mu,\zeta}(z)$
has a removable singularity at $\zeta=0$, or equivalently $\lim_{\zeta\rightarrow0}\zeta^{2}\widetilde{\mathcal{B}}_{m}^{\mu,\zeta}(z)=0$,
cf. \citep[Def. 1.6 and Thm. 1.2]{Conway1978}. Using the fact that
$\nu\mapsto I_{\nu}(y)$ is an entire analytic function \citep[\textsection\,3$\cdot$13]{Watson1922},
\citep{WhittakerWatson1902}, and abbreviating \[
f(\zeta)=\left(2\pi\right)^{-\frac{\left(d+\zeta\right)}{2}}\mu^{-\zeta}m^{\frac{\left(d+\zeta\right)}{2}-1}\left(-z^{2}\right)^{\frac{2-\left(d+\zeta\right)}{4}}I_{\frac{\left(d+\zeta\right)}{2}-1}(\sqrt{-m^{2}z^{2}})\,,\]
we compute for $d\in2\mathbb{N}$ (using l'Hôspital's rule), \[
\lim_{\zeta\rightarrow0}\tfrac{\pi}{2\sin(\left(\frac{\left(d+\zeta\right)}{2}-1\right)\pi)}\left[\zeta^{2}f(\zeta)\right]=\left\{ \tfrac{1}{\cos(\left(\frac{\left(d+\zeta\right)}{2}-1\right)\pi)}\left[\zeta^{2}f'(\zeta)+2\zeta f(\zeta)\right]\right\} _{\zeta=0}=0\,.\]
Given these properties, $\zeta\mapsto\widetilde{\mathcal{B}}_{m}^{\mu,\zeta}(z)$
can be expanded in a Laurent series \citep[1.11]{Conway1978},\begin{equation}
\widetilde{\mathcal{B}}_{m}^{\mu,\zeta}(z)=\sum_{n=-1}^{\infty}a_{n}(x)\,\zeta^{n}=\tfrac{1}{\zeta}\Res(\widetilde{\mathcal{B}}_{m}^{\mu,\zeta}(z),\zeta=0)+\mathcal{G}_{m}^{\mu,\zeta}(z)\,,\label{eq:DimRegHadamardSecondTermLaurentSeries}\end{equation}
where $\zeta\rightarrow\mathcal{G}_{m}^{\mu,\zeta}(x)$ is analytic
in the full disk $\left\{ \left|\zeta\right|<2\right\} $. For finite
$\zeta$ the function $\widetilde{\mathcal{B}}_{m}^{\mu,\zeta}$ is
(the analytic continuation of) a smooth, Lorentz invariant solution
of the Klein-Gordon equation {}``in $d+\zeta$ dimensions'', i.e.,
$\widetilde{\mathcal{B}}_{m}^{\mu,\zeta}(z)\!\sim\!\left(-z^{2}\right)^{-\frac{\nu}{2}}\! I_{\nu}(\sqrt{-m^{2}z^{2}})$,
where $I_{\nu}$ is the modified Bessel function of first kind of
order $\nu=\frac{d+\zeta}{2}-1$, cf.~(\ref{eq:FormOfLorentzInvariantSolution}).
Furthermore $\widetilde{\mathcal{H}}_{m}^{\mu,\zeta}(z)=\mathcal{W}_{m}^{\mu,\zeta}(z)+\widetilde{\mathcal{B}}_{m}^{\mu,\zeta}(z)$
is a smooth function of $m^{2}$. It is the whole purpose of this
derivation, to maintain as many of these properties as possible as
$\zeta$ tends to zero. Subtracting just the pole part of (\ref{eq:DimRegHadamardSecondTermLaurentSeries}),
as suggested in the original treatment \citep[App. A]{Brunetti2009},
although preserving smoothness in $z$ and $m^{2}$, does not lead
to a solution of the Klein-Gordon equation. The reason for this is
that the residue does not solve the Klein-Gordon equation in $d+\zeta$,
but in $d$ dimensions,\begin{align}
a_{-1}(x) & =\Res(\widetilde{\mathcal{B}}_{m}^{\mu,\zeta}(z),\zeta=0)=\lim_{\zeta\rightarrow0}\zeta\widetilde{\mathcal{B}}_{m}^{\mu,\zeta}(z)\nonumber \\
 & =\left(-1\right)^{\left(\frac{d}{2}-1\right)}\left(2\pi\right)^{-\frac{d}{2}}m^{\frac{d}{2}-1}\left(-z^{2}\right)^{\frac{2-d}{4}}I_{\frac{d}{2}-1}(\sqrt{-m^{2}z^{2}})\,.\label{eq:HadamardResidue}\end{align}
In order to maintain also the solution property in the limit $\zeta\rightarrow0$,
we need to subtract from $\widetilde{\mathcal{B}}_{m}^{\mu,\zeta}$
a smooth, Lorentz invariant solution of the Klein-Gordon equation
{}``in $d+\zeta$ dimensions''. Furthermore, in order not to spoil
the smoothness of $\widetilde{\mathcal{H}}_{m}^{\mu,\zeta}(z)$ in
the mass parameter $m^{2}$, the subtraction has to be a smooth function
of $m^{2}$. We conclude that the subtraction is a scalar multiple
of \begin{align}
\mathcal{S}_{m}^{\zeta}(z) & =\left(-1\right)^{\left(\frac{d}{2}-1\right)}m^{d-2}\left(\sqrt{-m^{2}z^{2}}\right)^{1-\frac{d+\zeta}{2}}I_{\frac{d+\zeta}{2}-1}(\sqrt{-m^{2}z^{2}})\,.\label{eq:MinimalSubtractionSmoothSolution}\end{align}
Observe that $\mathcal{S}_{m}^{\zeta}$ is an entire analytic function
of $z^{2}$, since $I_{\nu}(z)=z^{\nu}f_{\nu}(z^{2})$ with $f_{v}$
entire analytic. A possible subtraction is given by \begin{align*}
\mathcal{B}_{m}^{\mu,\zeta}(z) & =\widetilde{\mathcal{B}}_{m}^{\mu,\zeta}(z)-\alpha\tfrac{\pi}{2\sin(\zeta\frac{\pi}{2})}\mathcal{S}_{m}^{\zeta}(z)\\
 & \hspace{-20mm}=\left[\left(2\pi\right)^{-\frac{\zeta}{2}}\left(\tfrac{m}{\mu}\right)^{\zeta}\tfrac{\pi}{2\sin(\zeta\frac{\pi}{2})}-\tfrac{\pi}{2\sin(\zeta\frac{\pi}{2})}\right]\cdot\\
 & \hspace{-15mm}\cdot\left(-1\right)^{\left(\frac{d}{2}-1\right)}\left(2\pi\right)^{-\frac{d}{2}}m^{\frac{d}{2}-1}\left(-z^{2}\right)^{\frac{2-d}{4}}\left(\sqrt{-m^{2}z^{2}}\right)^{-\frac{\zeta}{2}}I_{\frac{d+\zeta}{2}-1}(\sqrt{-m^{2}z^{2}})\,,\end{align*}
where we had to set the factor in front of the subtraction to be $\alpha=\left(2\pi\right)^{-\frac{d}{2}}$
in order to get a well-defined limit. Adding this to the analytic
Wightman function $\mathcal{W}_{m}^{\mu,\zeta}$ defines the \emph{dimensionally
regularized analytic Hadamard function}\[
\mathcal{H}_{m}^{\mu,\zeta}:=\mathcal{W}_{m}^{\mu,\zeta}+\mathcal{B}_{m}^{\mu,\zeta}\equiv\widetilde{\mathcal{H}}_{m}^{\mu,\zeta}-\left(2\pi\right)^{-\frac{d}{2}}\tfrac{\pi}{2\sin(\zeta\frac{\pi}{2})}\mathcal{S}_{m}^{\zeta}(z)\,,\]
which has the explicit form(s)\begin{align}
\mathcal{H}_{m}^{\mu,\zeta}(z) & =\left(2\pi\right)^{-\frac{d}{2}}m^{d-2}\left(m\sqrt{-z^{2}}\right)^{1-\frac{d+\zeta}{2}}\cdot\nonumber \\
 & \qquad\cdot\left(-1\right)^{\frac{d}{2}-1}\frac{\pi}{2\sin(\zeta\tfrac{\pi}{2})}\left[\left(\frac{m}{\mu}\right)^{\zeta}I_{1-\frac{d+\zeta}{2}}(m\sqrt{-z^{2}})-I_{\frac{d+\zeta}{2}-1}(m\sqrt{-z^{2}})\right]\label{eq:DimRegAnalyticTwoPointFunction-1}\\
 & =\left(2\pi\right)^{-\frac{d}{2}}m^{d-2}\left(m\sqrt{-z^{2}}\right)^{1-\frac{d+\zeta}{2}}\cdot\nonumber \\
 & \qquad\cdot\left[\left(\frac{m}{\mu}\right)^{\zeta}K_{\frac{d+\zeta}{2}-1}+\left(-1\right)^{\frac{d}{2}-1}\frac{\pi}{2\sin(\zeta\frac{\pi}{2})}\left\{ \left(\frac{m}{\mu}\right)^{\zeta}-1\right\} I_{\frac{d+\zeta}{2}-1}\right]\,.\label{eq:DimRegAnalyticTwoPointFunction-2}\end{align}

From the representation (\ref{eq:DimRegAnalyticTwoPointFunction-1})
we can directly derive a series expansion of the dimensionally regularized
analytic two point function in powers of $m^{2}$ by inserting the
power series expansion (\ref{eq:ModifiedBesselFirstKind}) for the
modified Bessel functions, $d\in2\mathbb{N}$,\begin{align}
\mathcal{H}_{m}^{\mu,\zeta}(z) & =\left(-1\right)^{\frac{d-2}{2}}\left(2\pi\right)^{-\frac{d}{2}}2^{\frac{2-\left(d+\zeta\right)}{2}}\frac{\pi}{2\sin(\zeta\tfrac{\pi}{2})}\cdot\label{eq:DimRegAnalyticHadamardSeriesM}\\
 & \cdot\left[\sum_{s=0}^{\infty}\left\{ \left(\frac{2}{\mu\sqrt{-z^{2}}}\right)^{\zeta}\frac{1}{s!\,\Gamma(-\frac{\zeta}{2}-\frac{d-4}{2}+s)}-\frac{\theta(s-\frac{d-3}{2})}{\left(s-\frac{d-2}{2}\right)!\,\Gamma(\frac{\zeta}{2}+s+1)}\right\} \right.\cdot\nonumber \\
 & \hspace{70mm}\cdot\left.\left(\frac{\sqrt{-z^{2}}}{2}\right)^{2s-\left(d-2\right)}\left(m^{2}\right)^{s}\right]\!.\nonumber \end{align}
Observe that the second term in curly brackets is a smooth function
of $z^{2}$, since it only contributes for $s>\frac{d-3}{2}$, i.e.
$2s-\left(d-2\right)\geq0$.

Since $\mathcal{H}_{m}^{\mu,\zeta}$ differs from $\widetilde{\mathcal{H}}_{m}^{\mu,\zeta}$
by an entire analytic function, cf. (\ref{eq:MinimalSubtractionSmoothSolution}),
it is also analytic in the future and past tube. Furthermore $\zeta\mapsto\mathcal{H}_{m}^{\mu,\zeta}(z)$
is analytic in $\left\{ \left|\zeta\right|<2\right\} $ by construction,
however, for completeness we want to give an explicit argument here.
Regard (\ref{eq:DimRegAnalyticTwoPointFunction-2}), $K_{\frac{d+\zeta}{2}-1}$
and $I_{\frac{d+\zeta}{2}-1}$ are entire analytic functions of $\zeta$
by general properties of the (modified) Bessel functions (see e.g.
\citep[Sec.~9.6]{Abramowitz1970}). The analyticity domain of $\zeta\mapsto\mathcal{H}_{m}^{\mu,\zeta}(z)$
is thus determined by the factor $\alpha(\zeta)=\frac{\pi}{2\sin(\zeta\frac{\pi}{2})}\left\{ \left(\frac{m}{\mu}\right)^{\zeta}-1\right\} $,
which obviously is analytic in $\left\{ 0<\left|\zeta\right|<2\right\} $.
One can show differentiability in $\zeta=0$ by computing the differential
quotient directly,\[
\lim_{\zeta\rightarrow0}\frac{1}{\zeta}\left[\alpha(\zeta)-\alpha(0)\right]=\frac{1}{2}\left[\ln\left(\frac{m}{\mu}\right)\right]^{2}\,.\]
Hence $\left(z,\zeta\right)\mapsto\mathcal{H}_{m}^{\mu,\zeta}(z)$
is analytic for $z\in\tube^{\pm}$ and $\left|\zeta\right|<2$ as
asserted above. The limit $\zeta\rightarrow0$ of $\mathcal{H}_{m}^{\mu,\zeta}$
exists and defines the analytic Hadamard function in even dimensions\[
\mathcal{H}_{m}^{\mu}(z):=\lim_{\zeta\rightarrow0}\mathcal{H}_{m}^{\mu,\zeta}(z)\,,\quad d\in2\mathbb{N}\,,\]
it can be read off directly from (\ref{eq:DimRegAnalyticTwoPointFunction-2}),
$d\in2\mathbb{N}$, \[
\mathcal{H}_{m}^{\mu}(z)=\frac{m^{\frac{d}{2}-1}\left(-z^{2}\right)^{\frac{2-d}{4}}}{\left(2\pi\right)^{\frac{d}{2}}}\left[K_{\frac{d}{2}-1}(\sqrt{-m^{2}z^{2}})+\frac{\left(-1\right)^{\frac{d}{2}}}{2}\ln\left(\frac{\mu^{2}}{m^{2}}\right)I_{\frac{d}{2}-1}(\sqrt{-m^{2}z^{2}})\right].\]
By construction $\mathcal{H}_{m}^{\mu}$, as well as its regularization
$\mathcal{H}_{m}^{\mu,\zeta}$, $0\leq\left|\zeta\right|<2$, is a
smooth function of $m^{2}$. The boundary values of $\mathcal{H}_{m}^{\mu}(z)=\h_{m}^{\mu}(z^{2})$
define the Hadamard distribution and Feynman fundamental solution,\begin{equation}
\left\langle H_{m}^{\mu},f\right\rangle =\lim_{\varepsilon\rightarrow0^{+}}\int dx\, f(x)\,\h_{m}^{\mu}(x^{2}-ix^{0}\varepsilon)\,,\qquad H_{m}^{\mu}\in\D'(\M)\,,\label{eq:HadamardEvenDimensions}\end{equation}
and\begin{equation}
\left\langle H_{F}^{m,\mu},g\right\rangle =\lim_{\varepsilon\rightarrow0^{+}}\int dx\, f(x)\,\h_{m}^{\mu}(x^{2}-i\varepsilon)\,,\quad H_{F}^{m,\mu}\in\D'(\M\backslash\left\{ 0\right\} )\,.\label{eq:FeynmanEvenDimensions}\end{equation}

Within the analyticity domain of $\mathcal{H}_{m}^{\mu,\zeta}(z)$
the limit $\zeta\rightarrow0$ can be exchanged with taking boundary
values, resulting in regularizations of these distributions. What
seems artificial at this stage, since $ $$H_{m}^{\mu}$ needs no
regularization at all and $H_{F}^{m,\mu}$ already has a unique extension
with the same scaling degree by Theorem~\ref{thm:ExtensionScalingDegree},
will prove to be useful for the regularization of higher time-ordered
products in the next chapter. Hence we define the dimensionally regularized
Hadamard distribution, as well as the corresponding Feynman fundamental
solution as boundary values of $\mathcal{H}_{m}^{\mu,\zeta}(z)=\h_{m}^{\mu,\zeta}(z^{2})$,\begin{equation}
\left\langle H_{m}^{\mu,\zeta},f\right\rangle =\lim_{\varepsilon\rightarrow0^{+}}\int dx\, f(x)\,\h_{m}^{\mu,\zeta}(x^{2}-ix^{0}\varepsilon)\,,\quad H_{m}^{\mu,\zeta}\in\D'(\M)\,;\label{eq:DimRegHadamardBoundaryValue}\end{equation}
and\begin{equation}
\left\langle H_{F}^{m,\mu,\zeta},g\right\rangle =\lim_{\varepsilon\rightarrow0^{+}}\int dx\, g(x)\,\h_{m}^{\mu,\zeta}(x^{2}-i\varepsilon)\,,\quad H_{F}^{\mu,\zeta}\in\D'(\M\backslash\left\{ 0\right\} )\,.\label{eq:DimRegFeynmanBoundaryValue}\end{equation}
The scaling degree of $H_{F}^{m,\mu,\zeta}$ can be read off directly
from (\ref{eq:DimRegAnalyticHadamardSeriesM}), \begin{equation}
\sd(H_{F}^{m,\mu,\zeta})=d+\Re(\zeta)-2\,.\label{eq:ScalingDegreeHF-m-mu-zeta}\end{equation}
It is smaller than $d$ for $\Re(\zeta)<2$ and we infer again from
Theorem~\ref{thm:ExtensionScalingDegree} that $H_{F}^{m,\mu,\zeta}$
has a unique extension in this case. For $\zeta\notin\mathbb{R}$,
observe that the singular term of the expansion (\ref{eq:DimRegAnalyticHadamardSeriesM})
is homogeneous of degree $D=4-\left(d+\zeta\right)$ and hence we
get a unique extension by means of Theorem~\ref{thm:HomogeneousExtension}.

This unique extension $\dot{H}_{F}^{m,\mu,\zeta}\in\D'(\M)$ is a
regularization of $H_{F}^{m,\mu}\in\D'(\M\backslash\left\{ 0\right\} )$
in the sense of Definition~\ref{def:Regularization}, since by what
was said above we have\[
\forall f\in\D(\M\backslash\left\{ 0\right\} ):\quad\lim_{\zeta\rightarrow0}\left\langle \dot{H}_{F}^{m,\mu,\zeta},f\right\rangle =\lim_{\zeta\rightarrow0}\left\langle H_{F}^{m,\mu,\zeta},f\right\rangle =\left\langle H_{F}^{m,\mu},f\right\rangle \,.\]
And since the distribution on the right hand side also has a unique
extension $\dot{H}_{F}^{m,\mu}$, we even have\begin{equation}
\forall f\in\D(\M):\quad\lim_{\zeta\rightarrow0}\left\langle \dot{H}_{F}^{m,\mu,\zeta},f\right\rangle =\left\langle \dot{H}_{F}^{m,\mu},f\right\rangle \,.\label{eq:RegularizedFeynman}\end{equation}
Observe that $H_{F}^{m,\mu,\zeta}$ has a unique extension by means
of homogeneity, and not, as $H_{F}^{m,\mu}$, by a pure scaling degree
argument. We will see in the next chapter how this leads to an analytic
regularization of arbitrary time-ordered products.

\global\long\def\fps#1#2{#1\![[#2]]}

\global\long\def\poisson#1#2{\left\lfloor #1,#2\right\rceil }

\global\long\def\bld#1{\boldsymbol{#1}}

\begin{fmffile}{DimRegPos01}



\def\FGH{\parbox{10mm}{
\begin{center}
\begin{fmfgraph}(20,15)
\fmfbottom{g,h} \fmftop{f}
\fmfdot{f,g,h}
\end{fmfgraph}
\end{center}}}


\def\FGoneHone{\parbox{10mm}{
\begin{center}
\begin{fmfgraph}(20,15)
\fmfbottom{g,h} \fmftop{f}
\fmfdot{f,g,h}
\fmf{plain}{g,h}
\end{fmfgraph}
\end{center}}}

\def\FoneGHone{\parbox{10mm}{
\begin{center}
\begin{fmfgraph}(20,15)
\fmfbottom{g,h} \fmftop{f}
\fmfdot{f,g,h}
\fmf{plain}{f,h}
\end{fmfgraph}
\end{center}}}

\def\FoneGoneH{\parbox{10mm}{
\begin{center}
\begin{fmfgraph}(20,15)
\fmfbottom{g,h} \fmftop{f}
\fmfdot{f,g,h}
\fmf{plain}{f,g}
\end{fmfgraph}
\end{center}}}


\def\FGtwoHtwo{\parbox{10mm}{
\begin{center}
\begin{fmfgraph}(20,15)
\fmfbottom{g,h} \fmftop{f}
\fmfdot{f,g,h}
\fmf{plain,right=.5}{g,h}
\fmf{plain,left=.5}{g,h}
\end{fmfgraph}
\end{center}}}

\def\FoneGoneHtwo{\parbox{10mm}{
\begin{center}
\begin{fmfgraph}(20,15)
\fmfbottom{g,h} \fmftop{f}
\fmfdot{f,g,h}
\fmf{plain}{f,h}
\fmf{plain}{g,h}
\end{fmfgraph}
\end{center}}}

\def\FoneGoneoneHone{\parbox{10mm}{
\begin{center}
\begin{fmfgraph}(20,15)
\fmfbottom{g,h} \fmftop{f}
\fmfdot{f,g,h}
\fmf{plain}{f,g}
\fmf{plain}{g,h}
\end{fmfgraph}
\end{center}}}

\def\FtwoGHtwo{\parbox{10mm}{
\begin{center}
\begin{fmfgraph}(20,15)
\fmfbottom{g,h} \fmftop{f}
\fmfdot{f,g,h}
\fmf{plain,right=.5}{f,h}
\fmf{plain,left=.5}{f,h}
\end{fmfgraph}
\end{center}}}

\def\FtwoGoneHone{\parbox{10mm}{
\begin{center}
\begin{fmfgraph}(20,15)
\fmfbottom{g,h} \fmftop{f}
\fmfdot{f,g,h}
\fmf{plain}{f,g}
\fmf{plain}{f,h}
\end{fmfgraph}
\end{center}}}

\def\FtwoGtwoH{\parbox{10mm}{
\begin{center}
\begin{fmfgraph}(20,15)
\fmfbottom{g,h} \fmftop{f}
\fmfdot{f,g,h}
\fmf{plain,right=.5}{f,g}
\fmf{plain,left=.5}{f,g}
\end{fmfgraph}
\end{center}}}


\def\FthreeGoneHtwoF{\parbox{35mm}{
\begin{center}
\begin{fmfgraph*}(40,30)
\fmfbottom{g,h} \fmftop{f}
\fmfdot{f,g,h}
\fmflabel{$F^{(5)}$}{f}
\fmflabel{$G^{(4)}$}{g}
\fmflabel{$H^{(3)}$}{h}
\fmf{plain,right=.3}{f,g}
\fmf{plain}{f,g}
\fmf{plain,left=.3}{f,g}
\fmf{plain,right=.3}{g,h}
\fmf{plain,right=.3}{h,f}
\fmf{plain,left=.3}{h,f}
\end{fmfgraph*}
\end{center}}}


\def\LineX{\parbox{15mm}{
\begin{center}
\begin{fmfgraph*}(30,20)
\fmfleft{v} \fmfright{w}
\fmfdot{v,w}
\fmf{plain,label=$e$, l.side=left, l.dist=1mm}{v,w}
\fmfv{label=$v$,l.angle=-90}{v}
\fmfv{label=$w$,l.angle=-90}{w}
\end{fmfgraph*}
\end{center}}}

\chapter{\label{cha:DimReg-PositionSpace}Dimensional Regularization in Position
Space}

The work of Bollini and Giambiagi on dimensional regularization in
position space, mentioned previously, focused on the Fourier transform
of this regularization method between momentum space and position
space \citep{BolliniGiambiagi1996}. In contrast to their work, our
analysis will be formulated exclusively in position space, and a direct
translation to momentum space will generally not be possible. However,
an advantage of the method presented here is that all expressions
will depend smoothly on the mass parameter $m^{2}$, which makes it
possible to apply the covariant framework of \citep{Hollands2004},
see also \citep{DuetschFredenhagen2004}. We will analyze the graph
structure of the time-ordered product in the first section and use
this in the second section to construct for any graph a unique dimensionally
regularized amplitude. We will define the dimensionally regularized
time-ordered product and the corresponding scattering matrix. This
dimensionally regularized $\Sm$-matrix will then be used as an example
in the solution of the Epstein-Glaser recursion, and to establish
the relation to Connes-Kreimer theory of renormalization in the last
chapter.

\section{\label{sec:GraphStructureT}Graph structure of the Time-Ordered Product}

Before turning to the time-ordered product we want to introduce very
briefly the basic notions connected with the definition of a graph.
An oriented graph $\Gamma$ is a set of vertices $V(\Gamma)$ and
a set of edges $E(\Gamma)$ together with maps \[
\source,\target:E(\Gamma)\rightarrow V(\Gamma)\,,\]
which give source and target of an edge $e\in E(\Gamma)$, respectively.
Furthermore we give $\Gamma$ an \emph{orientation} by assigning to
any pair $\left(e,v\right)\in E(\Gamma)\times V(\Gamma)$ the value\[
\left(e:v\right):=\left\{ \begin{array}{cl}
+1 & \mbox{if }\target(e)=v\\
-1 & \mbox{if }\source(e)=v\\
0 & \mbox{otherwise.}\end{array}\right.\]
We call $e$ \emph{adjacent} \emph{to} $v$ if $\left(e:v\right)\neq0$.
A graph for which the orientation map $\left(e,v\right)\mapsto\left(e:v\right)$
is multi-valued we call \emph{tadpole}. However, the definition of
the time-ordered product in (\ref{eq:TimeOrderedProduct}) implies
that there are no tadpoles occurring in its graph expansion, i.e.,
Equation~(\ref{eq:TimeOrderedGraphStructure}) below; see the proof
of Proposition~\ref{pro:PartialAlgebraTimeProd} and the preceding
Remark~\ref{rem:TimeStar-and-Tadpoles}. In particular this implies
that we will only need to consider graphs for which\begin{equation}
\forall e\in E(\Gamma):\quad\source(e)\neq\target(e).\label{eq:NoTadpoleCondition}\end{equation}
Furthermore we remark that for scalar QFT the orientation of a given
(Feynman-) graph can be chosen freely, one speaks of an \emph{unoriented
graph} in this case. Let $\mathcal{G}$ denote the set of all unoriented
graphs $\Gamma$ for which the orientation map $\left(e,v\right)\mapsto\left(e:v\right)$
is single-valued, i.e. (\ref{eq:NoTadpoleCondition}) holds.

Consider the $n$-fold time-ordered product introduced in Section~\ref{sec:The-Renormalization-Problem}
as a map\[
\begin{array}{rccl}
\Time_{n}: & \fps{\mathcal{F}_{\loc}(\M)}{\hbar}^{\otimes n} & \rightarrow & \fps{\mathcal{F}(\M)}{\hbar}\\
 & F_{1}\otimes\cdots\otimes F_{n} & \mapsto & F_{1}\dT\cdots\dT F_{n}\,.\end{array}\]
It was defined with the help of a second order, symmetric functional
differential operator (\ref{eq:BiDifferentialTProduct}), which can
be written as \[
\Gamma_{H_{F}}(F\cdot G)=\left\langle H_{F},F^{\left(1\right)}\otimes G^{\left(1\right)}\right\rangle \,,\]
due to the absence of tadpoles. On the level of graphs this operation
is represented by drawing one line, $H_{F}$, between the interaction
vertices, $F$ and $G$. We can split the time-ordered product, $\Time_{n}$,
in a similar way into two parts \citep{FredenhagenGoettingen2009};
a differential operator, \[
\bld{\delta}^{\alpha}:F_{1}\otimes\cdots\otimes F_{n}\mapsto F_{1}^{\left(\alpha_{1}\right)}\otimes\cdots\otimes F_{n}^{\left(\alpha_{n}\right)}\,,\quad\alpha\in\mathbb{N}^{n}\,,\]
where $\alpha_{i}$ is the number of lines adjacent to the vertex
with interaction $F_{i}$,%
\footnote{Since we do not want to restrict ourselves to any particular type
of interaction, the number of edges adjacent to a given vertex is
not fixed a priori.%
} and a distribution,\[
S_{\Gamma}:F_{1}^{\left(\alpha_{1}\right)}\otimes\cdots\otimes F_{n}^{\left(\alpha_{n}\right)}\mapsto\left\langle S_{\Gamma},F_{1}^{\left(\alpha_{1}\right)}\otimes\cdots\otimes F_{n}^{\left(\alpha_{n}\right)}\right\rangle ,\]
containing the information as to which vertices of the graph $\Gamma$
are connected by a line. The $n$-fold time-ordered product can then
be written as\begin{equation}
F_{1}\dT\cdots\dT F_{n}=\sum_{\alpha\in\mathbb{N}^{n}}\sum_{\Gamma\in\mathcal{G}_{\alpha}}\frac{\hbar^{\left|E(\Gamma)\right|}}{\Sym(\Gamma)}\left\langle S_{\Gamma},\bld{\delta}^{\alpha}\left(F_{1}\otimes\cdots\otimes F_{n}\right)\right\rangle \,,\label{eq:TimeOrderedGraphStructure}\end{equation}
where $\mathcal{G}_{\alpha}$ is the set of (non-tadpole) graphs with
$n=\dim(\alpha)$ vertices and $\tfrac{\left|\alpha\right|}{2}$ lines
such that there are $\alpha_{i}$ lines joining at vertex $i$. $\Sym(\Gamma)\in\mathbb{N}$
is the symmetry factor of the graph $\Gamma$ to be defined below.
Observe that in the case of polynomial interactions, e.g., $F_{i}(\varphi)=\left\langle \varphi^{k_{i}},f\right\rangle $,
$i\in\left\{ 1,\dots,n\right\} $, and fixed $n\in\mathbb{N}$ only
finitely many of the functional derivatives $\bld{\delta}^{\alpha}$
give non-vanishing contributions to (\ref{eq:TimeOrderedGraphStructure}).
For arbitrary interactions the limiting parameter is the order in
$\hbar$ up to which one wants to compute. One can generate a dependence
on the loop number $\ell(\Gamma)$ for connected graphs $\Gamma$
by absorbing one factor $\hbar$ in each interaction functional,\[
\ell(\Gamma)=\left|E(\Gamma)\right|-\left|V(\Gamma)\right|+1\,,\]
a well known identity from graph theory \citep{GrossYellen2003}.
\begin{example}
As an example regard the threefold time-ordered product of (not necessarily
local) functionals $F,G,H\in\fps{\mathcal{F}(\M)}{\hbar}$. Using
Cauchy's product formula and the Leibniz rule one derives from the
power series expansion (\ref{eq:TimeOrderedProduct}) the following
expression for the threefold time-ordered product \begin{equation}
F\dT G\dT H=\sum_{n=0}^{\infty}\frac{\hbar^{n}}{n!}\sum_{m=0}^{n}\sum_{k=0}^{m}{n \choose m}{m \choose k}\left\langle F^{\left(k+m-k\right)}G_{\left(k\right)}^{\left(n-m\right)}H_{\left(n-m\right)\left(m-k\right)}\right\rangle \,,\label{eq:ExampleThreefoldProduct-1}\end{equation}
where we used the abbreviation $G_{\left(k\right)}:=H_{F}^{\otimes k}*G^{\left(k\right)}$,
cf.~\citep{Keller2009}. The first terms of this expansion are given
by\begin{align*}
F\dT G\dT H & =FGH+\hbar\left\{ \left\langle F^{\left(1\right)}G_{\left(1\right)}H\right\rangle +\left\langle F^{\left(1\right)}GH_{\left(1\right)}\right\rangle +\left\langle FG^{\left(1\right)}H_{\left(1\right)}\right\rangle \right\} \\
 & \quad+\hbar^{2}\left\{ \tfrac{1}{2}\left\langle F^{\left(2\right)}G_{\left(2\right)}H\right\rangle +\left\langle F^{\left(2\right)}G_{\left(1\right)}H_{\left(1\right)}\right\rangle +\tfrac{1}{2}\left\langle F^{\left(2\right)}GH_{\left(2\right)}\right\rangle \right.\\
 & \qquad\qquad\qquad\left.+\left\langle F^{\left(1\right)}G^{\left(1\right)}H_{\left(1\right)\left(1\right)}\right\rangle +\tfrac{1}{2}\left\langle FG^{\left(2\right)}H_{\left(2\right)}\right\rangle +\left\langle F^{\left(1\right)}G_{\left(1\right)}^{\left(1\right)}H_{\left(1\right)}\right\rangle \right\} \\
 & \quad+\cdots\end{align*}
\begin{align*}
\phantom{F\cdot_{T}G\cdot_{T}H} & =\FGH+\hbar\left\{ \FoneGoneH+\FoneGHone+\FGoneHone\right\} \\
 & \quad+\hbar^{2}\left\{ \frac{1}{2}\FtwoGtwoH+\FtwoGoneHone+\frac{1}{2}\FtwoGHtwo+\FoneGoneHtwo+\frac{1}{2}\FGtwoHtwo+\FoneGoneoneHone\right\} \\
 & \quad+\cdots\end{align*}
Observe that the prefactor of each graph is given by its symmetry
factor, \linebreak ${\Sym(\Gamma)^{-1}=\frac{1}{n!}{n \choose m}{m \choose k}=\frac{1}{\left(n-m\right)!k!\left(m-k\right)!}}$,
where for a general graph, $\Gamma\in\mathcal{G}$, $\Sym(\Gamma)$
is the product of the number of possible permutations of edges which
join the same two vertices in $\Gamma$.

The terms in (\ref{eq:ExampleThreefoldProduct-1}) can equivalently
be expressed as a composition of the maps $S_{\Gamma}$ and $\bld{\delta}^{\alpha}$
above,\[
\xymatrix{F\otimes G\otimes H\ar[d]^{\bld{\delta}^{\alpha}}\\
F^{\left(m\right)}\otimes G^{\left(n-m+k\right)}\otimes H^{\left(n-k\right)}\ar[d]^{\frac{1}{\Sym(\Gamma)}S_{\Gamma}}\\
\frac{1}{n!}{n \choose m}{m \choose k}\left\langle F^{\left(k+m-k\right)}G_{\left(k\right)}^{\left(n-m\right)}H_{\left(n-m\right)\left(m-k\right)}\right\rangle \,,}
\]
hence we can write (\ref{eq:ExampleThreefoldProduct-1}) equivalently
as graph expansion,\[
F\dT G\dT H=\sum_{\Gamma\in\mathcal{G}}\frac{\hbar^{\left|E(\Gamma)\right|}}{\Sym(\Gamma)}\left\langle S_{\Gamma},\bld{\delta}^{\alpha}(F\otimes G\otimes H)\right\rangle \,,\quad\alpha\in\mathbb{N}^{3},\quad\left|\alpha\right|=2\left|E(\Gamma)\right|\,.\]
We will properly define these maps, $\bld{\delta}^{\alpha}$ and $S_{\Gamma}$,
in the sequel.
\end{example}
In the case of local functionals $F_{v}\in\fps{\mathcal{F}_{\loc}(\M)}{\hbar}$
we have that the functional derivative can be written in the form,
cf. Eq.~(\ref{eq:LocFuncKernelArbSmallSupp}), \[
F_{v}^{\left(\alpha_{v}\right)}(\varphi)(x_{1},\dots,x_{\alpha_{v}})=\sum_{k}\sum_{i}f_{\varphi}^{v,k,i}(x_{v})\, P_{k}(\partial_{\bld r_{v}})\delta(\bld r_{v})\in\D(\M)\otimes\E_{\mathrm{Dirac}}'(\M^{\alpha_{i}-1})\,,\]
where $x_{v}=\frac{1}{\alpha_{v}}\sum_{k=1}^{\alpha_{v}}x_{k}$ is
the center of mass coordinate and $\bld r_{v}$ are relative coordinates
at vertex $v\in V(\Gamma)$. $P_{k}$ are homogeneous polynomials
of order $k$ in $\alpha_{v}-1$ variables, $\E_{\mathrm{Dirac}}'$
denotes the space of distributions supported at zero, and $\supp(f_{\varphi}^{v,k,i})$
can be chosen arbitrarily small. We want to introduce the short hand
notation\[
\D_{\loc}(\M^{\alpha_{v}}):=\D(\M)\otimes\E_{\mathrm{Dirac}}'(\M^{\alpha_{v}-1})\]
for the space of the $\alpha_{v}$th functional derivative of a local
functional, $F^{\left(\alpha_{v}\right)}(\varphi)\in\D_{\loc}(\M^{\alpha_{v}})$
and call \[
\begin{array}{rccl}
\bld{\delta}^{\alpha}\big|_{\varphi}: & \fps{\mathcal{F}_{\loc}(\M)}{\hbar}^{\otimes\left|V(\Gamma)\right|} & \rightarrow & \bigotimes_{v\in V(\Gamma)}\D_{\loc}(\M^{\alpha_{v}})\\
 & F_{1}\otimes\cdots\otimes F_{n} & \mapsto & F_{1}^{\left(\alpha_{1}\right)}(\varphi)\otimes\cdots\otimes F_{n}^{\left(\alpha_{n}\right)}(\varphi)\quad,\, n=\left|V(\Gamma)\right|\end{array}\]
the \emph{adjacency differential operator}.

While the definition of $\bld{\delta}^{\alpha}$ can be done purely
algebraically, the construction of the distribution $S_{\Gamma}$
on the other hand involves renormalization, i.e., an extension procedure
for distributions. We start from the tensor power\begin{equation}
\widetilde{S_{\Gamma}}=\bigotimes_{e\in E(\Gamma)}H_{F}(e)\,,\label{eq:TensorTProduct}\end{equation}
containing one factor $H_{F}\in\D'(\M\backslash\left\{ 0\right\} )$
for every edge $e$ in $\Gamma$. Hence $\widetilde{S_{\Gamma}}$
is a well-defined distribution in $\D'((\M\backslash\left\{ 0\right\} )^{\left|E(\Gamma)\right|})$
that can be uniquely extended to \linebreak[4]$\D'(\M^{\left|E(\Gamma)\right|})$,
since the Feynman fundamental solution $H_{F}\in\D'(\M\backslash\left\{ 0\right\} )$
has a unique extension $\dot{H}_{F}\in\D'(\M)$ with the same scaling
degree.

The renormalization problem is now to find a restriction $S_{\Gamma}$
of the tensor distribution $\widetilde{S_{\Gamma}}$ to the space
\[
\bigotimes_{v\in V(\Gamma)}\D_{\loc}(\M^{\alpha_{v}})=\D(\M^{\left|V(\Gamma)\right|})\otimes\bigotimes_{v\in V(\Gamma)}\E_{\mathrm{Dirac}}'(\M^{\alpha_{v}-1})\,.\]
The space $\E_{\mathrm{Dirac}}'$ is spanned by the $\delta$-distribution
and its derivatives \linebreak[4] \citep[Thm.~2.3.4]{Hoermander2003},
thus the tensor product \[
\V=\bigotimes_{v\in V(\Gamma)}\E_{\mathrm{Dirac}}'(\M^{\alpha_{v}-1})\]
is graded by the number of derivatives in front of the $\delta$-distributions.\[
\V=\bigoplus_{\left|\vec{k}\right|}\V_{\left|\vec{k}\right|}\,,\quad\left|\vec{k}\right|=\sum_{v\in V(\Gamma)}k_{v}\,.\]
Regard the application of $\widetilde{S_{\Gamma}}$ to the image of
$\bld{\delta}^{\alpha}\big|_{\varphi}$ ,\begin{align*}
\left\langle \widetilde{S_{\Gamma}},\bigotimes_{v\in V(\Gamma)}F_{v}^{\left(\alpha_{v}\right)}\right\rangle  & =\left\langle \widetilde{S_{\Gamma}},\bigotimes_{v\in V(\Gamma)}\sum_{k_{v}}\sum_{I_{v}}f_{\varphi}^{v,k_{v},I_{v}}\, P_{k_{v}}(\partial_{\bld r_{v}})\delta(\bld r_{v})\right\rangle \\
 & =\left\langle \widetilde{S_{\Gamma}},P_{\vec{k}}(\partial_{\vec{\bld r}})\bigotimes_{v\in V(\Gamma)}\sum_{k_{v}}\sum_{I_{v}}f_{\varphi}^{v,k_{v},I_{v}}\,\delta(\bld r_{v})\right\rangle \end{align*}
where $\vec{k}=\left(k_{v}\right)_{v\in V(\Gamma)}$ and $\vec{\bld r}=\left(\bld r_{v}\right)_{v\in V(\Gamma)}$.
We dualize the application of $P_{\vec{k}}(\partial_{\vec{\bld r}})$
and get\[
\left\langle \widetilde{S_{\Gamma}},\bigotimes_{v\in V(\Gamma)}F_{v}^{\left(\alpha_{v}\right)}\right\rangle =\left\langle \widetilde{S_{\left(\Gamma,\vec{k}\right)}}\,,\bigotimes_{v\in V(\Gamma)}\sum_{k_{v}}\sum_{I_{v}}f_{\varphi}^{v,k_{v},I_{v}}\,\delta(\bld r_{v})\right\rangle \]
where according to Lemma~\ref{lem:ScalingDegreeProperties} this
will increase the scaling degree of the distribution by $\left|\vec{k}\right|$,\begin{equation}
\sd(\widetilde{S_{\left(\Gamma,\vec{k}\right)}})=\sd(\widetilde{S_{\Gamma}})+\left|\vec{k}\right|\,.\label{eq:ScalingDegreeWithExternalStructure}\end{equation}
The multiindex $\vec{k}$ thus encodes the derivative couplings (i.e.,
the interaction functionals containing derivatives of the fields)
in the graph $\Gamma$. In the framework of Connes-Kreimer Hopf algebras,
or Feynman graphs in general, $\vec{k}$ sometimes is called the {}``external
structure of the graph'', see \citep{ConnesMarcolli2007} for instance. 

The remaining restriction of $\widetilde{S_{\left(\Gamma,\vec{k}\right)}}$
to a distribution in $\D'(\M^{\left|V(\Gamma)\right|})$ can conveniently
be described by the (simplicial) cohomology of the graph $\Gamma$.
For ease of presentation, we will forget about the external structure
$\vec{k}$ for the time being. The algebraic structure to be presented
below can be developed to a large extend without recourse to the external
structure. We will reintroduce $\vec{k}$ by replacing $\Gamma\mapsto\left(\Gamma,\vec{k}\right)$,
where we find it to be relevant for the understanding.

\subsection{\label{sub:SimplicialCohomology}Simplicial cohomology of a graph
and choice of relative coordinates}

The presentation in this subsection is very much inspired by \citep[Sec.~2.1]{BergbauerBrunettiKreimer2009}.
Let $\mathbb{K}\in\left\{ \mathbb{R},\mathbb{C}\right\} $ be a field.
We define the (simplicial) cohomology $H^{1}(\Gamma,\mathbb{K})$
with coefficients in $\mathbb{K}$ of a connected graph $\Gamma\in\mathcal{G}$
by the exact sequence\begin{equation}
\xymatrix{0\ar[r] & \mathbb{K}\ar@{^{(}->}[r]^{\cc\quad} & \mathbb{K}^{\left|V(\Gamma)\right|}\ar@{->}[r]^{\dc} & \mathbb{K}^{\left|E(\Gamma)\right|}\ar@{->>}[r]^{\oc} & H^{1}(\Gamma,\mathbb{K})\ar[r] & 0}
\,.\label{eq:SimplicialCohomology}\end{equation}
Let $\left\{ a_{v}:v\in V(\Gamma)\right\} $ be a basis of $\mathbb{K}^{\left|V(\Gamma)\right|}$
and $\left\{ b_{e}:e\in E(\Gamma)\right\} $ a basis of $\mathbb{K}^{\left|E(\Gamma)\right|}$.
The maps in (\ref{eq:SimplicialCohomology}) are then defined as the
{}``center of mass'', \[
\cc:x\mapsto x\sum_{v\in V(\Gamma)}a_{v}\,,\]
and \[
\dc:a_{v}\mapsto\sum_{e\in E(\Gamma)}\left(e:v\right)b_{e}\,.\]
One immediately checks that\[
\forall x\in\mathbb{K}:\quad\left(\dc\circ\cc\right)(x)=x\sum_{e\in E(\Gamma)}\sum_{v\in V(\Gamma)}\left(e:v\right)b_{e}=0\,.\]
Furthermore, $\oc\circ\dc=0$ is equivalent to \[
H^{1}(\Gamma,\mathbb{K})=\coker(\dc)=\mathbb{K}^{\left|E(\Gamma)\right|}/\im(\dc)\,,\]
an alternative definition of $H^{1}(\Gamma,\mathbb{K})$. The dimension
of this cohomology is called the first Betti number and gives the
number of independent loops of the graph $\Gamma$, $\dim(H^{1}(\Gamma,\mathbb{K}))=\left|E(\Gamma)\right|-\left|V(\Gamma)\right|+1$.

Let us regard the map $\dc$. The image of a general element $\vec{x}=\sum_{v\in V(\Gamma)}x^{v}a_{v}$
is given by\[
\dc(\vec{x})=\sum_{e\in E(\Gamma)}\left(x^{\target(e)}-x^{\source(e)}\right)b_{e}\,.\]
Thus $\dc$ expresses the coordinates of a given edge $e\in E(\Gamma)$
in terms of the coordinates of the adjacent vertices, $r^{e}=x^{\target(e)}-x^{\source(e)}$.
\begin{example}
Regard the very simple graph with two vertices and one edge, \[
\gamma=\LineX\,.\]
Let $v=\source(e)$ and $w=\target(e)$, and choose a basis $\left\{ a_{v},a_{w}\right\} $
of $\mathbb{K}^{\left|V(\gamma)\right|}$. Then $x\, a_{v}+y\, a_{w}\in\mathbb{K}^{\left|V(\gamma)\right|}$
and we have\[
\dc(x\, a_{v}+y\, a_{w})=\left(y-x\right)b_{e}\,.\]
Thus the pullback of $\dc:\mathbb{K}^{\left|V(\gamma)\right|}\rightarrow\mathbb{K}^{\left|E(\gamma)\right|}\equiv\mathbb{K}$
maps a function $f\in\D(\mathbb{K})$ to\[
\left(\dc^{*}f\right)(x,y)=f(y-x)\,.\]
Consequently, a distribution $u\in\D'(\mathbb{K})$ will be mapped
to $\dc^{*}u\in\D'(\mathbb{K}^{2})$, with\[
\left(\dc^{*}u\right)(f\otimes g)=\int_{\mathbb{K}^{2}}dx\, dy\, u(y-x)\, f(x)\, g(y)=\left\langle f,u*g\right\rangle \,,\]
where $*$ denotes the convolution product and the pullback is defined
in the sense of \citep[Thm.~6.1.2]{Hoermander2003}.
\end{example}
This construction can be lifted to any $\mathbb{K}$-vector space
$V$, by forming the tensor product $\mathbb{K}\otimes V$. We are
interested here in the lift to Minkowski spacetime $\M=\mathbb{R}\otimes\M$.
Thus we have, $\M^{n}=\mathbb{R}^{n}\otimes\M$,\[
\xymatrix{0\ar[r] & \M\ar@{^{(}->}[r]^{\hat{\cc}\quad} & \M^{\left|V(\Gamma)\right|}\ar@{->}[r]^{\hat{\dc}} & \M^{\left|E(\Gamma)\right|}\ar@{->>}[r] & H^{1}(\Gamma,\M)\ar[r] & 0}
\,,\]
i.e., one short exact sequence for each component of $z\in\M$.\addtocounter{thm}{-1}
\begin{example}
[revisited]In terms of this cohomology the Hadamard two point function
is the pullback of the Hadamard solution $H\in\D'(\M)$ by $\dc_{\gamma}$,\[
\left(\hat{\dc}_{\gamma}^{*}H\right)(f,g)=\left\langle f,H*g\right\rangle \,,\]
analogously the Feynman propagator is the pullback of the Feynman
fundamental solution $H_{F}\in\D'(\M\backslash\left\{ 0\right\} )$,
\[
\left(\hat{\dc}_{\gamma}^{*}H_{F}\right)(f,g)=\left\langle f,H_{F}*g\right\rangle \,,\]
$\supp(f)\cap\supp(g)=\emptyset$.
\end{example}
Also translation invariance can be formulated very conveniently in
this cohomological framework. The image of $\cc$ gives all possible
translations of the vertex coordinates by a given vector $a\in\mathbb{K}$.
Hence the orbits of these translations are the elements of the cokernel
$\coker(\cc)=\mathbb{K}^{\left|V(\Gamma)\right|}/\im(\cc$). We can
fix a basis of $\coker(\cc)$ by choosing the coordinates of a vertex
$v_{0}$ and setting $V_{0}:=V(\Gamma)\backslash\left\{ v_{0}\right\} $.
This provides us with a projection\[
\pi_{\Gamma}:\mathbb{K}^{\left|V(\Gamma)\right|}\rightarrow\mathbb{K}^{\left|V_{0}\right|}\equiv\mathbb{K}^{\left|V(\Gamma)\right|-1}\]
and an isomorphism $\phi$ between $\mathbb{K}^{\left|V_{0}\right|}$
and $\coker(\cc)$,\[
\begin{array}{rccl}
\phi: & \mathbb{K}^{\left|V_{0}\right|} & \rightarrow & \coker(\cc)\\
 & a_{v} & \mapsto & a_{v}+\im(\cc)\,.\end{array}\]
All translation invariant functions in $\mathbb{K}^{\left|V(\Gamma)\right|}$
can be seen as generic functions on $\coker(\cc)$, or $\mathbb{K}^{\left|V_{0}\right|}$
respectively. They are related by the pullback via $\pi_{\Gamma}$,
e.g. for smooth functions,\[
\begin{array}{rccl}
\pi_{\Gamma}^{*}: & \E(\mathbb{K}^{\left|V_{0}\right|}) & \rightarrow & \E_{\mathrm{tr.inv.}}(\mathbb{K}^{\left|V(\Gamma)\right|})\\
 & f & \mapsto & \left(\pi_{\Gamma}^{*}f\right)=f\circ\pi_{\Gamma}\,.\end{array}\]
We define the \emph{choice of relative coordinates} in $\mathbb{K}^{\left|E(\Gamma)\right|}$
by \[
\begin{array}{rccl}
\iota_{\Gamma}:=\dc_{\Gamma}\circ\phi: & \mathbb{K}^{\left|V_{0}\right|} & \rightarrow & \mathbb{K}^{\left|E(\Gamma)\right|}\\
 & \sum_{v\in V_{0}}x^{v}a_{v} & \mapsto & \sum_{e\in E(\Gamma)}r^{e}(\vec{x})\, b_{e}\,,\end{array}\]
where $r^{e}(\vec{x})=\sum_{v\in V_{0}}\left(e:v\right)x^{v}$ is
computed to be \[
r^{e}(\vec{x})=\begin{cases}
x^{\target(e)}-x^{\source(e)} & \mbox{if }v_{0}\notin\left\{ \source(e),\target(e)\right\} \\
x^{\target(e)} & \mbox{if }v_{0}=\source(e)\\
-x^{\source(e)} & \mbox{if }v_{0}=\target(e)\,,\end{cases}\]
giving the {}``coordinates of the edges'' relative to $v_{0}$.
In Minkowski spacetime we define correspondingly,\[
\hat{\iota}_{\Gamma}:=\hat{\dc}_{\Gamma}\circ\hat{\phi}:\M^{\left|V_{0}\right|}\rightarrow\M^{\left|E(\Gamma)\right|}\]
as the\emph{ }choice of relative coordinates in $\M^{\left|E(\Gamma)\right|}$.

We now want to define $S_{\Gamma}\in\D'(\M^{\left|V(\Gamma)\right|})$
as the pullback of $\widetilde{S_{\Gamma}}$ via $\hat{\dc}_{\Gamma}$,
\[
S_{\Gamma}=\hat{\dc}_{\Gamma}^{*}\,\widetilde{S_{\Gamma}}\,.\]
Let us regard the case of the unextended amplitude $\widetilde{S_{\Gamma}}\in\D'(\left(\M\backslash\left\{ 0\right\} \right)^{\left|E(\Gamma)\right|})$.
Each edge corresponds to a Feynman propagator $H_{F}$ and any set
of edges joining the same two vertices will have the same coordinate
$r^{e}(\vec{x})$. This introduces powers of $H_{F}$, which are well-defined
distributions only outside the origin, $\left(H_{F}\right)^{k}\in\D'(\M\backslash\left\{ 0\right\} )$.
As a consequence the pullback $S_{\Gamma}$ is a well-defined distribution
only outside the large diagonal \[
\DIAG=\left\{ \vec{x}\in\M^{\left|V(\Gamma)\right|}|\,\exists v,w\in V(\Gamma),v\neq w:\, x_{v}=x_{w}\right\} \,,\]
$S_{\Gamma}\in\D'(\M^{\left|V(\Gamma)\right|}\backslash\DIAG)$. As
remarked before a restriction of $\widetilde{S_{\Gamma}}$ by means
of a wave front set argument, i.e. by applying \citep[Thm.~8.2.4]{Hoermander2003},
is not possible due to the wave front set of $H_{F}$. A restriction
of $\widetilde{S_{\Gamma}}$, or equivalently an extension of $\hat{\dc}_{\Gamma}^{*}\widetilde{S_{\Gamma}}$
to $\D'(\M^{\left|V(\Gamma)\right|})$, will involve renormalization.
In the case of even dimensions, $d\in2\mathbb{N}$, the amplitude
is a tensor power of the Feynman propagator (\ref{eq:FeynmanEvenDimensions}),\[
\widetilde{S_{\Gamma}^{\mu}}:=\bigotimes_{e\in E(\Gamma)}H_{F}^{\mu}(e)\in\D'(\left(\M\backslash\left\{ 0\right\} \right)^{\left|E(\Gamma)\right|})\,,\]
and hence depends on an additional parameter $\mu$ of mass dimension
one.

We review briefly the Epstein-Glaser induction for constructing the
extension $S_{\Gamma}$ in order to discuss the renormalization freedom
in the cohomological framework advertised here.

\subsection{\label{sub:EpsteinGlaserInduction}The Epstein-Glaser induction}

Having defined what we mean by a graph $\Gamma\in\mathcal{G}$, we
define an \emph{Epstein-Glaser subgraph (EG subgraph)} $\gamma\subseteq\Gamma$
to be a subset of the set of vertices $V(\gamma)\subseteq V(\Gamma)$
together with all lines in $\Gamma$ connecting them,\[
E(\gamma)=\left\{ e\in E(\Gamma):\left\{ \source(e),\target(e)\right\} \subset V(\gamma)\right\} .\]
The orientation of $\gamma$ is inherited from $\Gamma$. The first
step of the Epstein-Glaser induction is to choose extensions for all
EG subgraphs with two vertices, $\left|V(\gamma)\right|=2$. In this
case we have translation invariant distributions $\hat{\dc}_{\gamma}^{*}\widetilde{S_{\gamma}}\in\D'(\M^{2}\backslash\Diag)$
($\Diag=\left\{ \vec{x}\in\M^{\left|V(\gamma)\right|}|\,\forall v,w\in V(\Gamma):\, x_{v}=x_{w}\right\} $
denotes the thin diagonal), which correspond to generic distributions
$\hat{\iota}_{\gamma}^{*}\widetilde{S_{\gamma}}\in\D'(\M\backslash\left\{ 0\right\} )$.
The scaling degree of these distributions is given by their number
of lines $\sd(\hat{\iota}_{\gamma}^{*}\widetilde{S_{\gamma}})=\left|E(\gamma)\right|\left(d-2\right)$,
and we can choose a (possibly unique) extension according to Theorem~\ref{thm:ExtensionScalingDegree}.
By translation invariance this gives extensions $S_{\gamma}\in\D'(\M^{2})$
of the distributions $\hat{\dc}_{\gamma}^{*}S_{\gamma}\in\D'(\M^{2}\backslash\Diag)$.
By causality, i.e. relation (\ref{eq:CausalityProducts}) of time-ordered
and algebra product, these extensions define the (translation-invariant)
restrictions of all EG subgraphs with three vertices up to the thin
diagonal.

For a generic EG subgraph $\gamma\subseteq\Gamma$ we make the assumption
that the restrictions of all EG subgraphs of $\gamma$ with less vertices
have already been chosen (induction hypothesis). The causality condition
then gives a translation invariant distribution $\hat{\dc}_{\gamma}^{*}\widetilde{S}_{\gamma}\in\D'(\M^{\left|V(\gamma)\right|}\backslash\Diag)$
which corresponds to a generic distribution $\hat{\iota}_{\gamma}^{*}\widetilde{S_{\gamma}}\in\D'(\M^{\left|V(\gamma)\right|-1}\backslash\left\{ 0\right\} )$.
The scaling degree and hence the degree of divergence of this distribution
is completely fixed by the structure of the graph, cf. (\ref{eq:FeynmanScalingDegree}),\begin{equation}
\div(\gamma)=\div(\hat{\iota}_{\gamma}^{*}\widetilde{S_{\gamma}})=\left|E(\gamma)\right|\left(d-2\right)-\left(\left|V(\gamma)\right|-1\right)d\,,\qquad d=\dim(\M)\,.\label{eq:DegreeOfDivergenceGraph}\end{equation}
We call $\gamma$ superficially convergent if $\div(\gamma)<0$, logarithmically
divergent if \linebreak[4]$\div(\gamma)=0$ and divergent of degree
$\div(\gamma)$ otherwise. Again by Theorem~\ref{thm:ExtensionScalingDegree}
there is a choice to be made in the extension of $\hat{\iota}_{\gamma}^{*}\widetilde{S_{\gamma}}$
in the case $\div(\gamma)\geq0$. The inductive procedure of Epstein-Glaser
will thus lead to an extension $S_{\Gamma}\in\D'(\M^{\left|V(\Gamma)\right|})$
of $\hat{\dc}_{\Gamma}^{*}\widetilde{S_{\Gamma}}\in\D'(\M^{\left|V(\Gamma)\right|}\backslash\DIAG)$.
As suggested above, we will refer to any such extension of $\hat{\dc}_{\Gamma}^{*}\widetilde{S_{\Gamma}}$
as \emph{restriction} of $\widetilde{S_{\Gamma}}\in\D'((\M\backslash\left\{ 0\right\} )^{\left|E(\Gamma)\right|})$.

In the case of couplings which involve derivatives of the fields,
also the external structure of $\gamma$ has to be taken into account,
cf.~(\ref{eq:ScalingDegreeWithExternalStructure}), \[
\div\!\left(\gamma,\vec{k}\right)=\div(\gamma)+\left|\vec{k}\right|\,.\]
This introduces an additional freedom in the choice of the extension
in each step, but does otherwise not change the inductive procedure.

The combination of all choices involved in the inductive construction
of a restriction $S_{\Gamma}\in\D'(\M^{\left|V(\Gamma)\right|})$
of $\widetilde{S_{\Gamma}}$ make up the Stückelberg-Petermann renormalization
group acting on local functionals, cf. \citep{Brunetti2009}. We will
see in the next section that the freedom in this construction is considerably
restricted, if we replace in (\ref{eq:TensorTProduct}) the Feynman
propagator by its dimensionally regularized counterpart.

\section{\label{sec:TheRegularizedAmplitude}The Regularized Amplitude}

The aim of this section is to construct a regularization of the above
defined amplitude $\hat{\dc}_{\Gamma}^{*}\widetilde{S_{\Gamma}^{\mu}}\in\D'(\M^{\left|V(\Gamma)\right|}\backslash\DIAG)$
by applying the Epstein-Glaser reduction procedure to\begin{equation}
\widetilde{S_{\Gamma}^{\mu,\zeta}}:=\bigotimes_{e\in E(\Gamma)}H_{F}^{m,\mu,\zeta}(e)\,.\label{eq:RegularizedTensorDistribution}\end{equation}
We will see in the sequel that $\widetilde{S_{\Gamma}^{\mu,\zeta}}$
has a unique restriction $S_{\Gamma}^{\mu,\zeta}\in\D'(\M^{\left|V(\Gamma)\right|})$
by means of Corollary~\ref{cor:ExtensionHeterogeneousDistribution}.
This will provide a regularization of the original amplitude $\hat{\dc}_{\Gamma}^{*}\widetilde{S_{\Gamma}^{\mu}}$
outside the large diagonal.

Regard the dimensionally regularized Feynman fundamental solution
$H_{F}^{m,\mu,\zeta}\in\D'(\M\backslash\left\{ 0\right\} )$ constructed
in Chapter~\ref{cha:DimRegHadamard}. The expansion of $H_{F}^{m,\mu,\zeta}$
in powers of the mass parameter $m^{2}$ follows directly from the
expansion of the analytic Hadamard function (\ref{eq:DimRegAnalyticHadamardSeriesM}),\begin{align}
H_{F}^{m,\mu,\zeta}(x) & =\left(-1\right)^{\frac{d-2}{2}}\left(2\pi\right)^{-\frac{d}{2}}2^{\frac{2-\left(d+\zeta\right)}{2}}\frac{\pi}{2\sin(\zeta\tfrac{\pi}{2})}\cdot\nonumber \\
 & \hspace{-5mm}\cdot\left[\sum_{s=0}^{\infty}\left\{ \left(\frac{2}{\mu\sqrt{-x^{2}+i0}}\right)^{\zeta}\frac{1}{s!\,\Gamma(-\frac{\zeta}{2}-\frac{d-4}{2}+s)}-\frac{\theta(s-\frac{d-3}{2})}{\left(s-\frac{d-2}{2}\right)!\,\Gamma(\frac{\zeta}{2}+s+1)}\right\} \right.\cdot\label{eq:DimRegFeynmanSeriesM}\\
 & \hspace{60mm}\cdot\left.\left(\frac{\sqrt{-x^{2}+i0}}{2}\right)^{2s-\left(d-2\right)}\left(m^{2}\right)^{s}\right]\!,\nonumber \\
 & =:\sum_{s=0}^{\infty}H_{F}^{s,\mu,\zeta}(x)\,\left(m^{2}\right)^{s}\,,\nonumber \end{align}
where we used the common shorthand $f(x^{2}-i0)=\lim_{\varepsilon\rightarrow0^{+}}f(x^{2}-i\varepsilon)$.
The coefficients of this series, $H_{F}^{s,\mu,\zeta}\in\D'(\M\backslash\left\{ 0\right\} )$,
are sums of a distributional and a smooth part, both of which are
homogeneous, but of different degree. The distributional part is homogeneous
of degree $2s-\left(d+\zeta-2\right)$, whereas the smooth part is
identically zero for $s<\frac{d-3}{2}$ and homogeneous of degree
$2s-\left(d-2\right)$ otherwise.

Regard now the finite tensor powers of the dimensionally regularized
Feynman distribution. The expansion of $\left(H_{F}^{m,\mu,\zeta}\right)^{\otimes k}$,
$k=\left|E(\Gamma)\right|$, in $m^{2}$ follows directly from (\ref{eq:DimRegFeynmanSeriesM})
and Cauchy's product formula,\begin{align}
\left(H_{F}^{m,\mu,\zeta}\right)^{\otimes k} & =\sum_{s_{k}=0}^{\infty}\left(m^{2}\right)^{s_{k}}\sum_{s_{k-1}=0}^{s_{k}}\cdots\sum_{s_{1}=0}^{s_{2}}H_{F}^{\left(s_{k}-s_{k-1}\right),\mu,\zeta}\otimes\cdots\otimes H_{F}^{s_{1},\mu,\zeta}\label{eq:DimRegFeynmanTensorPowerSeriesM}\\
 & =:\sum_{s_{k}=0}^{\infty}\left(m^{2}\right)^{s_{k}}\left(H_{F}^{\mu,\zeta}\right)^{\otimes\vec{s}},\nonumber \end{align}
where $\vec{s}=\left(s_{k}-s_{k-1},\dots,s_{1}\right)\in\mathbb{N}^{k}$.
The scaling degree of the coefficient of $\left(m^{2}\right)^{s_{k}}$
is given by the sum of the scaling degrees of the individual factors,
cf. Lemma~\ref{lem:ScalingDegreeProperties} and Remark~\ref{rem:ScalingDegreeHomogeneity},\begin{equation}
\sd(\left(H_{F}^{\mu,\zeta}\right)^{\otimes\vec{s}})=k\left(d+\Re(\zeta)-2\right)-2s_{k}\,.\label{eq:ScalingDegreeCoefficientM}\end{equation}
Hence the scaling degree of the coefficients become arbitrarily small
as one regards higher powers of $m^{2}$. Thus according to Theorem~\ref{thm:ExtensionScalingDegree}
these coefficients will have unique restrictions to arbitrary subdiagonals.%
\footnote{Keep in mind that the product of distributions can be defined as the
restriction of their tensor product to subdiagonals, cf.~\citep[Thm.~8.2.10]{Hoermander2003}.%
} In other words, only a finite number of coefficients in (\ref{eq:DimRegFeynmanTensorPowerSeriesM})
need renormalization. This is one of the advantages of the concept
of a {}``scaling expansion'' introduced in \citep{HollandsWald2002};
see also \citep{Hollands2004}, \citep[p.~1310ff]{DuetschFredenhagen2004}.
The coefficients in (\ref{eq:DimRegFeynmanTensorPowerSeriesM}) which
need renormalization when restricted to subdiagonals are in general
not homogeneous, since the coefficients $H_{F}^{s,\mu,\zeta}$ in
(\ref{eq:DimRegFeynmanSeriesM}) are not homogeneous. However, since
we regard graphs with a finite number of edges, they are certainly
heterogeneous of finite order. If we assume $\zeta\notin\mathbb{Q}_{+}$
their multidegree contains no integer number. We can construct a restriction
of $\widetilde{S_{\Gamma}^{\mu,\zeta}}$ by the same procedure described
in Section~\ref{sub:EpsteinGlaserInduction}, with the only difference
that we have a preferred choice of the extension at each order according
to Corollary~\ref{cor:ExtensionHeterogeneousDistribution}, namely
the extensions which are heterogeneous of the same multidegree. Thus
we are lead to a unique restriction $S_{\Gamma}^{\mu,\zeta}\in\D'(\M^{\left|V(\Gamma)\right|})$.
\begin{prop}
[Regularization outside $\DIAG$]\label{pro:RegularizedAmplitude}The
restriction of $\widetilde{S_{\Gamma}^{\mu,\zeta}}$ is \linebreak[4]uniquely
defined by the above homogeneity condition and gives a distribution
$S_{\Gamma}^{\mu,\zeta}\in\D'(\M^{\left|V(\Gamma)\right|})$. \textup{$S_{\Gamma}^{\mu,\zeta}$}
is a regularization of $\hat{\dc}_{\Gamma}^{*}\widetilde{S_{\Gamma}^{\mu}}\in\D'(\M^{\left|V(\Gamma)\right|}\backslash\DIAG)$,
in the sense that\[
\forall f\in\D(\M^{\left|V(\Gamma)\right|}\backslash\DIAG):\quad\lim_{\zeta\rightarrow0}\left\langle S_{\Gamma}^{\mu,\zeta},f\right\rangle =\left\langle \hat{\dc}_{\Gamma}^{*}\widetilde{S_{\Gamma}^{\mu}},f\right\rangle \,.\]
We will refer to $S_{\Gamma}^{\mu,\zeta}\in\D'(\M^{\left|V(\Gamma)\right|})$
as the dimensionally regularized amplitude of $\Gamma$. By translation
invariance it naturally corresponds to a unique dimensionally regularized
amplitude in relative coordinates, $s_{\Gamma}^{\mu,\zeta}\in\D'(\M^{\left|V(\Gamma)\right|-1})$,
which is the unique extension of $\hat{\iota}_{\Gamma}^{*}\widetilde{S_{\Gamma}^{\mu,\zeta}}\in\D'(\left(\M\backslash\left\{ 0\right\} \right)^{\left|V(\Gamma)\right|-1})$.\end{prop}
\begin{proof}
The first part follows from the construction above. By the discussion
at the end of the previous chapter we have that the unique extensions
$\dot{H}_{F}^{m,\mu,\zeta}\in\D'(\M)$ is a regularization of $H_{F}^{m,\mu}\in\D'(\M\backslash\left\{ 0\right\} )$,
cf. (\ref{eq:RegularizedFeynman}). And by continuity of the maps
involved we get the assertion,\begin{align*}
\hat{\dc}_{\Gamma}^{*}\left(\bigotimes_{e\in E(\Gamma)}H_{F}^{m,\mu}(e)\right) & =\hat{\dc}_{\Gamma}^{*}\left(\bigotimes_{e\in E(\Gamma)}\lim_{\zeta\rightarrow0}H_{F}^{m,\mu,\zeta}(e)\right)\\
 & =\lim_{\zeta\rightarrow0}\hat{\dc}_{\Gamma}^{*}\left(\bigotimes_{e\in E(\Gamma)}H_{F}^{m,\mu,\zeta}(e)\right)\end{align*}
Observe that the extension map $\hat{\dc}_{\Gamma}^{*}\widetilde{S_{\Gamma}^{\mu,\zeta}}\mapsto S_{\Gamma}^{\mu,\zeta}$
is also continuous by Theorem~\ref{thm:HomogeneousExtension}.
\end{proof}
We have that $S_{\Gamma}^{\mu,\zeta}$ is a regularization of $\hat{\dc}_{\Gamma}^{*}\widetilde{S_{\Gamma}^{\mu}}$
in a broader sense of the word, since $\hat{\dc}_{\Gamma}^{*}\,\widetilde{S_{\Gamma}^{\mu}}$
is defined only in the complement of the large diagonal. However,
the regularization $S_{\Gamma}^{\mu,\zeta}\in\D'(\M^{\left|V(\Gamma)\right|})$
comes with a natural renormalization prescription, defined at any
order of (causal) perturbation theory: minimal subtraction (MS). This
has already been introduced on the conceptual level in Section~\ref{sec:RegularizationOfDistributions},
and we will see in the next chapter, how minimal subtraction is to
be applied to the regularized graph amplitudes $s_{\Gamma}^{\mu,\zeta}$
and $S_{\Gamma}^{\mu,\zeta}$, respectively. The complete renormalization
of the graph amplitudes will then be discussed in Chapter~\ref{cha:ForestFormula}
and it will be useful, for the derivation of the underlying combinatorial
structure to collect all the different contributions to the perturbative
expansion in the definition of a unique dimensionally regularized
$\Sm$-matrix, defined as a map on local functionals, cf.~Section~\ref{sec:TimeOrderedProduct}.
\begin{defn}
[{Dimensionally Regularized $\Sm$-matrix}]\label{def:DimRegSMatrix}Let
\[
\Gamma_{H_{F}^{m,\mu,\zeta}}':=\frac{1}{2}\int dx\, dy\, H_{F}^{m,\mu,\zeta}(x,y)\frac{\delta^{2}}{\delta\varphi(x)\,\delta\varphi(y)}\,,\]
be the dimensionally regularized Feynman bidifferential operator.
Define the regularized time-ordering operator\[
\Time_{\mu,\zeta}:=\exp(\hbar\Gamma_{H_{F}^{m,\mu,\zeta}}')\,,\]
and the dimensionally regularized time-ordered product on local functionals
$F,G\in\fps{\mathcal{F}_{\loc}(\M)}{\hbar}$\[
F\cdot_{\Time_{\mu,\zeta}}G:=\Time_{\mu,\zeta}\left(\Time_{\mu,\zeta}^{-1}F\cdot\Time_{\mu,\zeta}^{-1}G\right)\,.\]
Then we define the dimensionally regularized $\Sm$-matrix as\[
\Sm_{\mu,\zeta}(F):=\exp_{\cdot_{\Time_{\mu,\zeta}}}(F)=\sum_{n=0}^{\infty}\frac{1}{n!}\Time_{\mu,\zeta}^{n}(F^{\otimes n})\,,\qquad F\in\fps{\mathcal{F}_{\loc}(\M)}{\hbar}\,,\]
where $\Time_{\mu,\zeta}^{n}$ denotes the uniquely extended regularized
$n$-fold time-ordered product constructed (graph by graph) by Epstein-Glaser
induction.
\end{defn}
Inserting (\ref{eq:TimeOrderedGraphStructure}) we can write the regularized
$\Sm$-matrix also in terms of a graph expansion\begin{equation}
\Sm_{\mu,\zeta}(F)=\sum_{n=0}^{\infty}\frac{1}{n!}\sum_{\alpha\in\mathbb{N}^{n}}\sum_{\Gamma\in\mathcal{G}_{\alpha}}\frac{\hbar^{\left|E(\Gamma)\right|}}{\Sym(\Gamma)}\left\langle S_{\Gamma}^{\mu,\zeta},\bld{\delta}^{\alpha}\left(F^{\otimes n}\right)\right\rangle \,,\qquad F\in\fps{\mathcal{F}_{\loc}(\M)}{\hbar}.\label{eq:PerturbationSeries}\end{equation}
This expansion is often referred to as the perturbative expansion
of the $\Sm$-matrix. And we want to remark that the sum over all
graphs at a fixed order $n$ of causal perturbation theory is finite,
if we assume that $F\in\fps{\mathcal{F}_{\loc}(\M)}{\hbar}$ is a
polynomial interaction functional. This remains valid, if $F$ contains
derivatives of the field. The order of causal perturbation theory
is given by the number of vertices of the graphs contributing to (\ref{eq:PerturbationSeries}),
irrespective of the fact if they are connected or not. Conversely,
the sum is finite at each order $\mathcal{O}($$\hbar^{\left|E(\Gamma)\right|}$),
and we repeat the remark that this is in essence the {}``loop order'',
if we regard only graphs with a fixed number of connected components
$c(\Gamma)$. The order is given by the Betti number of the graph
\citep{GrossYellen2003}, \[
\ell(\Gamma)=\left|E(\Gamma)\right|-\left|V(\Gamma)\right|+c(\Gamma)\,,\]
if we {}``hide'' one power of $\hbar$ in the interaction functional
$F$.

\end{fmffile}

\global\long\def\fps#1#2{#1\![[#2]]}

\global\long\def\poisson#1#2{\left\lfloor #1,#2\right\rceil }

\global\long\def\bld#1{\boldsymbol{#1}}

\global\long\def\oprod#1{\sideset{}{^{\geq}}\prod_{#1}}

\begin{fmffile}{MS01}

\def\FthreeGoneHtwoF{\parbox{20mm}{
\begin{center}
\begin{fmfgraph}(24,18)
\fmfbottom{g,h} \fmftop{f}
\fmfdot{f,g,h}
\fmf{plain,right=.3}{f,g}
\fmf{plain}{f,g}
\fmf{plain,left=.3}{f,g}
\fmf{plain,right=.3}{g,h}
\fmf{plain,right=.3}{h,f}
\fmf{plain,left=.3}{h,f}
\end{fmfgraph}
\end{center}}}

\def\FthreeG{\parbox{20mm}{
\begin{center}
\begin{fmfgraph}(24,18)
\fmfleft{f} \fmfright{g}
\fmfdot{f,g}
\fmf{plain,right=.3}{f,g}
\fmf{plain}{f,g}
\fmf{plain,left=.3}{f,g}
\end{fmfgraph}
\end{center}}}

\def\FGoneHtwoF{\parbox{20mm}{
\begin{center}
\begin{fmfgraph}(24,18)
\fmfbottom{g,h} \fmftop{f}
\fmfdot{f,g,h}
\fmf{phantom,right=.3}{f,g}
\fmf{phantom}{f,g}
\fmf{phantom,left=.3}{f,g}
\fmf{plain,right=.3}{g,h}
\fmf{plain,right=.3}{h,f}
\fmf{plain,left=.3}{h,f}
\end{fmfgraph}
\end{center}}}

\def\DotH{\parbox{5mm}{
\begin{center}
\begin{fmfgraph}(10,5)
\fmfleft{f} \fmfright{g}
\fmf{phantom}{f,h,g}
\fmfdot{h}
\end{fmfgraph}
\end{center}}}

\chapter{\label{cha:Minimal-Subtraction}Minimal Subtraction}

Minimal subtraction (MS) in combination with dimensional regularization
\linebreak[4](DimReg) and Zimmermann's forest formula as a renormalization
technique has earned wide acclaim in the standard approach to perturbative
renormalization in momentum space. After having constructed the dimensionally
regularized position space amplitude to any graph $\Gamma\in\mathcal{G}$,
we want to extend the notion of minimal subtraction given in Section~\ref{sec:RegularizationOfDistributions}
also to graph amplitudes and products thereof. As a matter of fact,
we will find that minimal subtraction can be formulated independently
of the graph expansion and the representation (position- or momentum
space). This is to say that we can define a minimal subtraction operator
which acts directly on the prepared, dimensionally regularized time-ordered
product, regarded as a linear map between functional spaces,\[
\Time_{\mu,\zeta,\prep}^{n}:\fps{\mathcal{F}_{\loc}(\M)}{\hbar}^{\otimes n}\rightarrow\fps{\mathcal{F}(\M)}{\hbar}\,.\]
The fact that this leads to local counterterms will be the crucial
observation which makes the abstraction in the next chapter possible,
and the presented forest formula for Epstein-Glaser renormalization
applicable in any chosen representation.

What will be said in this chapter relies on the fact that we dispose
of a \emph{prepared amplitude}. This will be defined in the first
section, and we will implement the graph structure in the second.
In the third section we will define minimal subtraction at subgraphs,
and we will test our method by rederiving the result of Zimmermann
that only Epstein-Glaser subgraphs contribute to nested projections
in the limit where the regularization is removed \citep{Zimmermann1975}.
The independence on the representation will be discussed in the fourth
section and as a result we will define the minimal subtraction operator
on prepared time-ordered products.

\section{Prepared Amplitude}
\begin{defn}
[Prepared Amplitude]\label{def:PreparedAmplitude}A regularization
$\left\{ s_{\Gamma,\prep}^{\mu,\zeta}:\zeta\in\Omega\backslash\left\{ 0\right\} \right\} $
(in the strict sense of Definition~\ref{def:Regularization}) is
called \emph{prepared amplitude} of $\left\{ s_{\Gamma}^{\mu,\zeta}\right\} $,
if it is a regularization of $\hat{\iota}_{\Gamma}^{*}\widetilde{S_{\Gamma}^{\mu}}\in\D'(\left(\M\backslash\left\{ 0\right\} \right)^{\left|V(\Gamma)\right|-1})$
outside the large diagonal in the sense of Proposition~\ref{pro:RegularizedAmplitude},
i.e.,\[
\forall f\in\D(\left(\M\backslash\left\{ 0\right\} \right)^{\left|V(\Gamma)\right|-1}):\qquad\lim_{\zeta\rightarrow0}\left\langle s_{\Gamma,\prep}^{\mu,\zeta},f\right\rangle =\left\langle \hat{\iota}_{\Gamma}^{*}\widetilde{S_{\Gamma}^{\mu}},f\right\rangle ,\]
and $s_{\Gamma,\prep}^{\mu,\zeta}$ is heterogeneous of finite, non-integer
order in $\Omega\backslash\left\{ 0\right\} $.
\end{defn}
Observe that for $\left|V(\Gamma)\right|=2$ the regularization outside
the large diagonal is already a regularization in the strict sense
of Definition~\ref{def:Regularization}, and thus a prepared amplitude.
Hence minimal subtraction can be applied and leads to a finite regularization.
In the logical framework of Epstein-Glaser one would then define the
prepared amplitude of the third order and subtract the counterterm,
and so on to the order one chooses to compute. At each step minimal
subtraction is applied to a prepared amplitude, but if we want to
define the subtraction performed on the unrenormalized amplitude,
these subtractions will be nested. One aim of this chapter is to analyze
these nested subtractions. They will be used in Chapter~\ref{cha:ForestFormula}
to solve the recursion of Epstein-Glaser. A closed expression for
the prepared amplitude will then follow immediately from the solution.
Thus we can assume here that we dispose of a prepared amplitude $s_{\Gamma,\prep}^{\mu,\zeta}$.
Since $s_{\Gamma,\prep}^{\mu,\zeta}$ is a regularization by assumption,
we can directly apply the analysis of Section~\ref{sec:RegularizationOfDistributions}
and have that the principal part of its Laurent series is a local
distribution, \begin{equation}
\pp(s_{\Gamma,\prep}^{\mu,\zeta})\in\E_{\mathrm{Dirac}}'(\M^{\left|V(\Gamma)\right|-1})\,,\label{eq:PP-Prepared-is-Local}\end{equation}
where we denoted by $\E_{\mathrm{Dirac}}'$ the space of distributions
supported at the origin. We infer that\[
\left\{ \rp(s_{\Gamma,\prep}^{\mu,\zeta}):\,\zeta\in\Omega\backslash\left\{ 0\right\} \right\} \]
is a finite regularization of $\hat{\iota}_{\Gamma}^{*}\widetilde{S_{\Gamma}^{\mu}}$
and hence\[
s_{\Gamma,\ren}^{\mu}:=\lim_{\zeta\rightarrow0}\rp(s_{\Gamma,\prep}^{\mu,\zeta})\in\D'(\M^{\left|V(\Gamma)\right|-1})\,.\]
is a renormalization. To have a name for it, we call $\pp(s_{\Gamma,\prep}^{\mu,\zeta})$
and $\rp(s_{\Gamma,\prep}^{\mu,\zeta})$ the \emph{projected} prepared
amplitudes. Nested projections will lead to projections in different
parts of the same graph. The different components of a graph needed
for the discussion later on will be defined in the following section.

\section{Subgraphs and Complements}

Since the method we are analyzing here was originally formulated in
momentum space, where the edges of the graphs carry as label the {}``momentum
flowing through this line'', it is natural to consider as subgraphs
all graphs, which are given by a subset of the set of edges. Given
a graph $\Gamma$, we call a \emph{BPHZ subgraph} any subgraph $\gamma\subseteq\Gamma$
given by a subset of the set of edges, $E(\gamma)\subseteq E(\Gamma)$,
and all adjacent vertices,\[
V(\gamma)=\left\{ v\in V(\Gamma)|\:\exists e\in E(\gamma):\left(e:v\right)\neq0\right\} .\]
The orientation is inherited from $\Gamma$. See, e.g., \citep{CaswellKennedy1982}
for a description of the BPHZ procedure within dimensional regularization
and minimal subtraction in momentum space. 

The set of BPHZ subgraphs of a graph $\Gamma\in\mathcal{G}$ is a
superset to the set of Epstein-Glaser subgraphs defined in Section~\ref{sub:EpsteinGlaserInduction},
and we can associate to any BPHZ subgraph a unique Epstein-Glaser
subgraph. Given a graph $\Gamma$ with BPHZ subgraph $\gamma\subseteq\Gamma$,
we define the \emph{full vertex part} $\overline{\gamma}$ of $\gamma$
to be the graph with the same set of vertices, $V(\overline{\gamma})=V(\gamma)$,
and all lines in $\Gamma$ connecting them,\[
E(\overline{\gamma})=\left\{ e\in E(\Gamma):\source(e),\target(e)\in V(\Gamma)\right\} .\]
$\overline{\gamma}$ obviously is an Epstein-Glaser subgraph. Any
BPHZ subgraph, which is not a full vertex part, we call a \emph{pure
BPHZ subgraph}.

For the definition of products of (projected) amplitudes corresponding
to different parts of the same graph $\Gamma$ it is important to
have the notion of a complement of a subgraph. Observe, however, that
there are two natural ways to define this complement, and both will
be of relevance in the sequel. 
\begin{defn}
[Complements of a graph]\label{def:GraphComplements}Let $\Gamma\in\mathcal{G}$
be a graph and $G\subset\Gamma$ be a subgraph. We define the \emph{line
complement} $\Gamma\bbs G$ of $G$ in $\Gamma$ to be the graph with
\[
V(\Gamma\bbs G)=V(\Gamma)\qquad\mbox{and}\qquad E(\Gamma\bbs G)=E(\Gamma)\backslash E(G)\,.\]
Furthermore we define the \emph{vertex complement} $\Gamma\obs G$
to be the full vertex part with vertex set\[
V(\Gamma\obs G)=V(\Gamma)\backslash V(G)\,,\]
i.e.,\[
E(\Gamma\obs G)=\left\{ e\in E(\Gamma)|\source(e),\target(e)\in V(\Gamma\obs G)\right\} .\]

\end{defn}
Observe that, while the vertex complement $\Gamma\obs G$ is a full
vertex part by definition, the line complement is not a full vertex
part in the generic case. For the line complement the number of lines
is preserved in the sense that\[
E(\Gamma)=E(G)\dot{\cup}E(\Gamma\bbs G)\,.\]
For the vertex complement, on the other hand, the number of lines
is not preserved; $E(G)\dot{\cup}E(\Gamma\obs G)$ will be a subset
of $E(\Gamma)$ in general, because the lines connecting $G$ with
$\Gamma\obs G$ are not considered. We have $E(\Gamma)=E(G)\dot{\cup}E(\Gamma\obs G)$
if and only if $\Gamma$ is multiply connected with $G$ one of its
(possibly also multiply connected) components. However, the vertex
set is preserved for the vertex complement,\[
V(\Gamma)=V(G)\dot{\cup}V(\Gamma\obs G)\,,\]
a fact that will be of importance in the discussion of Chapter~\ref{cha:ForestFormula}.
\begin{example}
\label{exa:GraphComplements}Regard the graph $\Gamma$ and subgraph
$G\subset\Gamma$,\[
\Gamma=\FthreeGoneHtwoF\,,\quad G=\FthreeG.\]
Then the two complements of Definition~\ref{def:GraphComplements}
are depicted by\[
\Gamma\bbs G=\FGoneHtwoF\qquad\mbox{and}\qquad\Gamma\obs G=\DotH\,,\]
respectively.
\end{example}
We will now use the line complement for the definition of minimal
subtractions at BPHZ subgraphs, and meet the vertex complement again
in the next chapter.

\section{\label{sec:RedundantProjections}MS at Subgraphs and Redundant Projections}

Let $\gamma\subset\Gamma$ be a proper BPHZ subgraph. Then $\gamma$
has less edges than $\Gamma$, $E(\gamma)\subsetneq E(\Gamma)$, and
it may have less vertices, $V(\gamma)\subseteq V(\Gamma)$. Let \begin{equation}
\varpi_{\Gamma,\gamma}:\M^{\left|V(\Gamma)\right|-1}\rightarrow\M^{\left|V(\gamma)\right|-1}\label{eq:ProjectionVertexSpaces}\end{equation}
be the induced projection ($\varpi_{\Gamma,\gamma}=\id$, if $V(\Gamma)=V(\gamma)$).
Then the pullback $\varpi_{\Gamma,\gamma}^{*}s_{\gamma}^{\mu,\zeta}$
exists as a distribution in $\D'(\M^{\left|V(\Gamma)\right|-1})$,
cf.~\citep[Thm.~6.1.2]{Hoermander2003}, and we have that\begin{equation}
s_{\Gamma}^{\mu,\zeta}=s_{\Gamma\bbs\gamma}^{\mu,\zeta}\cdot\varpi_{\Gamma,\gamma}^{*}s_{\gamma}^{\mu,\zeta}\,,\quad\zeta\notin\mathbb{Q}_{+}\,,\label{eq:ProductSubgraphAmplitudes}\end{equation}
where by the expression on the right hand side we understand the unique
heterogeneous extension of the pointwise product of the distributions,
as constructed in Section~\ref{sec:TheRegularizedAmplitude}. Let
us now regard the same product (\ref{eq:ProductSubgraphAmplitudes}),
when the subgraph part is replaced by a projected prepared amplitude,\begin{equation}
\left(1-T_{\MS}^{\gamma}\right)s_{\Gamma}^{\mu,\zeta}:=s_{\Gamma\bbs\gamma}\cdot\rp(s_{\gamma,\prep}^{\mu,\zeta})\,,\quad\mbox{ or }\quad T_{\MS}^{\gamma}s_{\Gamma}^{\mu,\zeta}:=s_{\Gamma\bbs\gamma}\cdot\pp(s_{\gamma,\prep}^{\mu,\zeta})\,,\label{eq:SubgraphProjections}\end{equation}
respectively. Where the pullback via the projection (\ref{eq:ProjectionVertexSpaces})
is understood, but not explicitly written to improve readability.
Then we define the product on the respective right hand sides term
by term in the Laurent expansion, i.e.,\begin{equation}
s_{\Gamma\bbs\gamma}\cdot\rp(s_{\gamma,\prep}^{\mu,\zeta})=\sum_{n=0}^{\infty}\zeta^{n}\frac{1}{2\pi i}\oint_{C}d\xi\,\frac{1}{\xi^{n+1}}s_{\Gamma\bbs\gamma}^{\mu,\zeta}\cdot\varpi_{\Gamma,\gamma}^{*}s_{\gamma}^{\mu,\xi}\,,\label{eq:rpSubgraphLaurentSeries}\end{equation}
and\begin{equation}
s_{\Gamma\bbs\gamma}\cdot\pp(s_{\gamma,\prep}^{\mu,\zeta})=\sum_{n=-\infty}^{-1}\zeta^{n}\frac{1}{2\pi i}\oint_{C}d\xi\,\frac{1}{\xi^{n+1}}s_{\Gamma\bbs\gamma}^{\mu,\zeta}\cdot\varpi_{\Gamma,\gamma}^{*}s_{\gamma}^{\mu,\xi}\,,\label{eq:ppSubgraphLaurentSeries}\end{equation}
where $C\subset\Omega\backslash\left\{ 0\right\} $ is a small circle
around the origin. The product of the distributions under the complex
line integral is defined by (\ref{eq:ProductSubgraphAmplitudes})
for an appropriate parameter value $\zeta$ and almost all values
of $\xi$. Observe also that the extension map ($u\mapsto\dot{u}$)
is continuous for a homogeneous distribution, in the case the map
is uniquely defined, cf. Theorem~\ref{thm:HomogeneousExtension}.
The fact that we regard finite sums of homogeneous distributions (i.e.,
heterogeneous distributions of finite order) does not spoil this continuity,
and hence the extension of the distribution under the integral above
commutes with the integration over one of its parameters.

In Zimmermann's forest formula \citep[Thm.~3.3]{Zimmermann1969},
if one reads it as if it was formulated in position space with the
above definitions, there occur nested projections of the form\begin{equation}
\left(1-T_{G}^{\MS}\right)\left(1-T_{\gamma}^{\MS}\right)s_{G,\prep}^{\mu,\zeta}=\rp\!\left[s_{G\bbs\gamma}^{\mu,\zeta}\cdot\rp\!(s_{\gamma,\prep}^{\mu,\zeta})\right],\quad\mbox{\ensuremath{\gamma\subset}G}.\label{eq:ProductMSOperatorsGraphwise}\end{equation}
Shortly after the publication in 1969 Zimmermann himself realized
that not all nested projections of the above form contribute in the
limit where the regularization is removed, i.e., $\zeta\rightarrow0$
\citep{Zimmermann1970}. The projection $\left(1-T_{\gamma}^{\MS}\right)$
in (\ref{eq:ProductMSOperatorsGraphwise}) is redundant, if $G$ and
$\gamma$ have the same set of vertices. This, in turn, leads to the
fact that only Epstein-Glaser graphs contribute to the forest formula.
Zimmermann used the Pauli-Villars regularization method to prove this
fact in \citep{Zimmermann1975}. In the momentum space version of
dimensional regularization and minimal subtraction the canceling of
spurious terms in the limit has also been observed by Falk, Häußling,
and Scheck by calculating explicit examples. Consequently the authors
proposed an alternative renormalization method in momentum space,
which takes into account the spurious subtractions \citep{Falk2009}.

We want to use Zimmermann's observation as a test of our position
space dimensional regularization method, and the prescription for
minimal subtraction. We will see that in position space, i.e., for
the nested projection (\ref{eq:ProductMSOperatorsGraphwise}) with
the definitions given above, Zimmermann's result is a direct consequence
of the fact that we can write the projection to the regular part,
$\rp$, as a $W$-projection on test functions (up to a term of $\mathcal{O}(\zeta)$),
cf.~Equation~(\ref{eq:RegularizationMS-FiniteZeta}).
\begin{prop}
[Redundant Projections]\label{pro:RedundantProjections}Let $\gamma\subsetneq G$
be two BPHZ subgraphs of $\Gamma\in\mathcal{G}$ with the same vertex
set, i.e., $\gamma$ is a pure BPHZ subgraph, \[
V(\gamma)=V(G)\,,\quad\mbox{and}\quad E(\gamma)\subsetneq E(G)\,.\]
The contribution of the pure BPHZ subgraph $\gamma\subsetneq\Gamma$
to\[
\left(1-T_{G}^{\MS}\right)\left(1-T_{\gamma}^{\MS}\right)s_{G,\prep}^{\mu,\zeta}=\rp\!\left[s_{G\bbs\gamma}^{\mu,\zeta}\cdot\rp(s_{\gamma,\prep}^{\mu,\zeta})\right]\]
vanishes identically in the limit $\zeta\rightarrow0$. That is, $\forall f\in\D(\M^{\left|V(\overline{\gamma})\right|-1})$:\begin{equation}
\lim_{\zeta\rightarrow0}\left\langle \rp(s_{G,\prep}^{\mu,\zeta}),f\right\rangle =\lim_{\zeta\rightarrow0}\left\langle \rp\left[s_{G\bbs\gamma}^{\mu,\zeta}\cdot\rp(s_{\gamma,\prep}^{\mu,\zeta})\right],f\right\rangle .\label{eq:ReductionAssertion}\end{equation}
\end{prop}
\begin{proof}
The argument of the limit on the right hand side of (\ref{eq:ReductionAssertion})
can be rewritten using (\ref{eq:RegularizationMS-FiniteZeta}),\begin{align}
\left\langle \rp\left[s_{G\bbs\gamma}^{\mu,\zeta}\cdot\rp(s_{\gamma,\prep}^{\mu,\zeta})\right],f\right\rangle  & =\left\langle s_{G\bbs\gamma}^{\mu,\zeta}\cdot\rp(s_{\gamma,\prep}^{\mu,\zeta}),W_{G}^{\MS}f\right\rangle +\mathcal{O}(\zeta)\nonumber \\
 & \hspace{-7mm}=\left\langle s_{G,\prep}^{\mu,\zeta},W_{G}^{\MS}f\right\rangle -\left\langle s_{G\bbs\gamma}^{\mu,\zeta}\cdot\pp(s_{\gamma,\prep}^{\mu,\zeta}),W_{G}^{\MS}f\right\rangle +\mathcal{O}(\zeta)\,,\label{eq:2ndTermHasToVanish}\end{align}
and we will show in the sequel that the second term in this expression
vanishes identically for finite $\zeta$.

The principal part $\pp(s_{\gamma,\prep}^{\mu,\zeta})$ is a local
distribution, $\supp(\pp(s_{\gamma,\prep}^{\mu,\zeta}))=\left\{ 0\right\} $,
cf.~(\ref{eq:PP-Prepared-is-Local}). Hence also the product in the
second term of (\ref{eq:2ndTermHasToVanish}) is supported at the
origin,\[
\supp(s_{G\bbs\gamma}^{\mu,\zeta}\cdot\pp(s_{\gamma}^{\mu,\zeta}))=\left\{ 0\right\} ,\]
and thus local, $s_{G\bbs\gamma}^{\mu,\zeta}\cdot\pp(s_{\gamma}^{\mu,\zeta})\in\E_{\mathrm{Dirac}}'$.
The degree of divergence of this local distribution can be inferred
directly from the scaling degrees of the individual lines, cf. (\ref{eq:Regularization-PP-Local})
and (\ref{eq:ScalingDegreeHF-m-mu-zeta}), \begin{align*}
\sd(s_{G\bbs\gamma}^{\mu,\zeta}\cdot\pp(s_{\gamma}^{\mu,\zeta})) & =\left|E(G\bbs\gamma)\right|\left(d+\Re(\zeta)-2\right)+\left|E(\gamma)\right|\left(d-2\right)\,,\\
\div(s_{G\bbs\gamma}^{\mu,\zeta}\cdot\pp(s_{\gamma}^{\mu,\zeta})) & =\left|E(G)\right|\left(d-2\right)-\left|V(G)\right|\left(d-1\right)+\left|E(G\bbs\gamma)\right|\Re(\zeta)\\
 & =\div(G)+\left|E(G\bbs\gamma)\right|\Re(\zeta)\,.\end{align*}
Hence we infer that \[
s_{G\bbs\gamma}^{\mu,\zeta}\cdot\pp(s_{\gamma}^{\mu,\zeta})=\sum_{\left|\alpha\right|\leq\left\lfloor \div(s_{G\bbs\gamma}^{\mu,\zeta}\cdot\pp(s_{\gamma}^{\mu,\zeta}))\right\rfloor }C_{\alpha}(\zeta)\delta^{\left(\alpha\right)}\,,\]
where $\left\lfloor \cdot\right\rfloor $ denotes, as before, Gauß's
floor function. Given $\left\lfloor \div\Big(s_{G\bbs\gamma}^{\mu,\zeta}\cdot\pp(s_{\gamma}^{\mu,\zeta})\Big)\right\rfloor =\div(G)$,
which is the case if $\Re(\zeta)<\frac{1}{\left|E(G\bbs\gamma)\right|}$,
we have that \[
\forall g\in\D_{\div(G)}:\quad\left\langle s_{G\bbs\gamma}^{\mu,\zeta}\cdot\pp(s_{\gamma}^{\mu,\zeta}),g\right\rangle =0\]
and hence, by the uniqueness property of analytic functions, \[
\forall f\in\D:\quad\left\langle s_{G\bbs\gamma}^{\mu,\zeta}\cdot\pp(s_{\gamma}^{\mu,\zeta}),W_{G}^{\MS}f\right\rangle =0\]
for $\zeta$ in a neighborhood of the origin. 
\end{proof}
As a matter of fact Proposition~\ref{pro:RedundantProjections} implies
only that all forests containing the same set of full vertex parts
give the same contribution to the sum. Hence it could happen that
the contributions add up to give multiple contributions to the forest
formula. However, one can show that of all forests with the same set
of full vertex parts only one contributes to Zimmermann's formula.
The combinatorial argument is also given in Zimmermann's proof in
\citep{Zimmermann1975}. We don't want to repeat it at this point,
since the result is implied by the forest formula for regularized
Epstein-Glaser renormalization we will prove in the next chapter.
Motivated by these results we drop the cumbersome distinction between
Epstein-Glaser - and BPHZ subgraphs and define a subgraph to be what
we called to this point an Epstein-Glaser subgraph or full vertex
part.
\begin{defn}
[Subgraph]\label{def:Subgraph}Let $\Gamma\in\mathcal{G}$ be a graph.
We define a subgraph $\gamma\subseteq\Gamma$ to be given by a subset
of the set of vertices $V(\gamma)\subseteq V(\Gamma)$ and all lines
in $\Gamma$ connecting them,\begin{equation}
E(\gamma)=\left\{ e\in E(\Gamma):\left\{ \source(e),\target(e)\right\} \subset V(\gamma)\right\} .\label{eq:EGSubgraphSetOfEdges}\end{equation}
We explicitly allow single vertices as subgraphs, and since there
are no tadpoles in $\mathcal{G}$ (cf.~Section~\ref{sec:GraphStructureT}),
these one vertex subgraphs will have no lines. Observe that also $\Gamma\subseteq\Gamma$,
trivially, is a subgraph.
\end{defn}

\section{\label{sec:MS-For-TProduct}MS for the Time-ordered Product}

Regard the set $\mathcal{G}_{V}\subset\mathcal{G}$ of all graphs
with the same set of vertices $V$,\[
\mathcal{G}_{V}=\left\{ \Gamma\in\mathcal{G}:V(\Gamma)=V\right\} .\]
This set gives all graph contributions to the order $\left|V\right|$
of causal perturbation theory, and $\mathcal{G}_{V}(\left|E\right|)=\mathcal{G}_{\alpha}$,
$\dim(\alpha)=\left|V\right|$, $\left|\alpha\right|=2\left|E\right|$,
is finite if one regards only the contributions up to a given order
$\left|E\right|$ in $\hbar$, cf. Section~\ref{sec:GraphStructureT}
and Equation~(\ref{eq:PerturbationSeries}). Assume that we have
a prepared amplitude $s_{\Gamma,\prep}^{\mu,\zeta}\in\D'(\M^{\left|V\right|-1})$
for all graphs $\Gamma\in\mathcal{G}_{V}(\left|E\right|)$ at any
order $\left|E\right|$ of $\hbar$. Let $S_{\Gamma,\prep}^{\mu,\zeta}\in\D'(\M^{\left|V\right|})$
be the corresponding translation invariant amplitude defined for local
functionals in the sense of formal power series in $\hbar$. Then
we can write the minimal subtraction operator at order $\left|V\right|$
of causal perturbation theory on the level of graph amplitudes as,\begin{equation}
R_{V}(\Time_{\mu,\zeta,\prep}^{\left|V\right|})=\sum_{\Gamma\in\mathcal{G}_{\alpha}}\frac{\hbar^{\left|E(\Gamma)\right|}}{\Sym(\Gamma)}\, R_{V}\!\left\langle S_{\Gamma,\prep}^{\mu,\zeta},\bld{\delta}^{\alpha}\right\rangle \quad\mbox{ with }\quad R_{V}=\begin{cases}
\id & \mbox{if }\left|V\right|=1\\
-\pp & \mbox{if }\left|V\right|>1\,.\end{cases}\label{eq:MS-Operator-GraphExpansion}\end{equation}
Here $\pp(S_{\Gamma,\prep}^{\mu,\zeta})$ is the translation invariant
analogue of the local distribution \linebreak[4]$\pp(s_{\Gamma,\prep}^{\mu,\zeta})\in\E_{\mathrm{Dirac}}'(\M^{\left|V(\Gamma)\right|-1})$
defined above, $\supp(\pp(S_{\Gamma,\prep}^{\mu,\zeta}))\subset\Diag(\M^{\left|V(\Gamma)\right|})$.
We want to apply the corresponding term in the above sum to a tensor
product of local functionals. Analogous to the discussion in Section~\ref{sec:GraphStructureT}
we get\begin{equation}
\left\langle -\pp(S_{\Gamma,\prep}^{\mu,\zeta}),\bld{\delta}^{\alpha}\right\rangle (\bigotimes_{v\in V}F_{v})=\left\langle -\pp(S_{\Gamma,\prep}^{\mu,\zeta}),\bigotimes_{v\in V(\Gamma)}f_{\varphi}^{v}\,\delta(\bld r_{v})\right\rangle \,,\quad f_{\varphi}^{v}\in\D(\M),\label{eq:LocalDistriButionsLocalFunctionals}\end{equation}
where $f_{\varphi}^{v}\in\D(\M)$ is a sum of pointwise products of
test functions with the field $\varphi$. Since $-\pp(S_{\Gamma,\prep}^{\mu,\zeta})$
is supported on the thin diagonal, all functions $f_{\varphi}^{v}$
are evaluated at the same point and the expression on the right hand
side of (\ref{eq:LocalDistriButionsLocalFunctionals}) gives a local
functional. We have \[
\left\langle \pp(S_{\Gamma,\prep}^{\mu,\zeta}),\bld{\delta}^{\alpha}\right\rangle :\fps{\mathcal{F}_{\loc}(\M)}{\hbar}^{\otimes\left|V\right|}\rightarrow\fps{\mathcal{F}_{\loc}(\M)}{\hbar}\,.\]
The fact which establishes the independence of the presented formalism
on the chosen representation, is that the projection to the principal
part, $\pp$, is an operation with respect to the parameter $\zeta$,
and can be performed outside the brackets {}``~$\left\langle \cdot\right\rangle $~''.
Actually it was defined like that in Section~\ref{sec:RegularizationOfDistributions}.
Although the evaluation of these brackets might look very different,
depending on the chosen representation. Thus minimal subtraction is
really an operation which can be performed directly on time-ordered
products,and it is sensible to define
\begin{defn}
[Minimal Subtraction Operator on Subsets]\label{def:MS-OperatorSubsets}For
any vertex set $V$, we define the \emph{minimal subtraction operator
(MS operator) on subsets} as\[
R_{V}(\Time_{\mu,\zeta,\prep}^{\left|V\right|}):=\begin{cases}
\id & \mbox{if }\left|V\right|=1\\
-\pp(\Time_{\mu,\zeta,\prep}^{\left|V\right|}) & \mbox{if }\left|V\right|>1\,,\end{cases}\]
where $\id:\fps{\mathcal{F}_{\loc}(\M)}{\hbar}\rightarrow\fps{\mathcal{F}_{\loc}(\M)}{\hbar}$
is the identity map on local functionals, and\[
-\pp(\Time_{\mu,\zeta,\prep}^{\left|I\right|}):\fps{\mathcal{F}_{\loc}(\M)}{\hbar}^{\otimes\left|I\right|}\rightarrow\fps{\mathcal{F}_{\loc}(\M)}{\hbar}\]
is the local counterterm at order $\left|I\right|$ of causal perturbation
theory.
\end{defn}
\end{fmffile}

\global\long\def\fps#1#2{#1\![[#2]]}

\global\long\def\poisson#1#2{\left\lfloor #1,#2\right\rceil }

\global\long\def\bld#1{\boldsymbol{#1}}

\global\long\def\oprod#1{\sideset{}{^{\geq}}\prod_{#1}}

\begin{fmffile}{FFRenHopf02}

\def\FourVertexPhiThree{\parbox{30mm}{
\begin{center}
\begin{fmfgraph}(80,15)
\fmfstraight
\fmftop{t1,t3,t4,t6}
\fmfbottom{b1,b2,b3,b4,b5,b6}
\fmf{plain}{b1,b2}
\fmf{plain, right=.6}{b2,b5}
\fmf{plain}{b5,b6}
\fmf{plain,left=.2}{b2,t3}
\fmf{plain,left=.3}{t3,t4}
\fmf{plain,right=.3}{t3,t4}
\fmf{plain,left=.2}{t4,b5}
\fmfdot{b2,b5,t3,t4}
\end{fmfgraph}
\end{center}}}

\chapter{\label{cha:ForestFormula}The Epstein-Glaser Forest~Formula}

\begin{center}
\begin{minipage}[t][1.2\totalheight]{0.8\columnwidth}%
\begin{flushleft}
\emph{\textquoteleft{}What\textquoteright{}s the sandwich scenario,
Mo?\textquoteright{}}
\par\end{flushleft}

\emph{\textquoteleft{}Ham and cheese; ham and tomato; cheese and tomato.\textquoteright{}}

\emph{\textquoteleft{}And ham, cheese and tomato.\textquoteright{}}

\emph{\textquoteleft{}How did you know?\textquoteright{}}

\emph{\textquoteleft{}You\textquoteright{}ve never noticed how you
group sandwiches into Venn diagrams?\textquoteright{}}

\emph{\textquoteleft{}Do I?\textquoteright{}}

\begin{flushright}
\emph{David Mitchell: Ghostwritten}
\par\end{flushright}%
\end{minipage}
\par\end{center}

It was the principle of covariance, understood as the axiom that all
physically relevant concepts must have an analogue in (globally hyperbolic)
curved spacetime, which brought to light the more profound structures
of perturbative renormalization theory in the investigation undertaken
by Brunetti, Dütsch, Fredenhagen, Hollands, and Wald (see references
in the introduction).%
\footnote{The covariance principle was made precise in \citep{Brunetti2003}.
And we want to use this footnote to remark that despite its reputation
of being conceptually clear but {}``too far from reality'' to have
predictive power for experiments the algebraic approach and in particular
perturbative Algebraic Quantum Field Theory has lead to falsifiable
predictions in cosmology \citep{Dappiaggi2008}.%
} As already said in the introduction of this thesis one of the main
results of their program was the formulation of perturbative Algebraic
Quantum Field Theory (pAQFT), briefly introduced in Chapter~\ref{cha:Setting-pAQFT}.
In this last chapter we will show that the tools of pAQFT and in particular
the precise statement of Stora's main theorem of perturbative renormalization,
augmented by the results on analytic regularization we have gained
in the previous chapters will make it possible to solve the recursive
procedure of Epstein-Glaser renormalization and to prove a forest
formula in the sense of Zimmermann for Epstein-Glaser renormalization.
The result will be independent of the chosen representation and will
in particular be applicable in momentum and in position space. The
main theorem of renormalization, written in termwise form by using
Faà~di~Bruno's formula for the $n$-fold chain rule \citep{FaaDiBruno1855},
implies a recursion relation for the minimally subtracted counterterms
to an analytically regularized $\Sm$-matrix. This recursion relation
will be crucial for the proof of the forest formula. 

In 1982 Joni and Rota introduced a bialgebra related to Faà~di~Bruno's
formula \citep{JoniRota1982}. We will use this bialgebra to derive
(a summed up version of) the Connes-Kreimer Hopf algebra of graphs
directly from the main theorem of renormalization. However, in contrast
to the Connes-Kreimer approach the Feynman rules will emerge naturally
and are not assumed to be characters into the commutative ring of
Laurent series with scalar coefficients. The emergent Feynman rules
will rather produce linear maps between spaces of (local) functionals.
On the space of linear maps the construction induces two products,
a symmetrized tensor product, and a non-commutative product, which
is given as the composition of linear maps. Both products, as well
as the coproduct, need to be reflected in the Hopf algebra (of graphs)
in order to encode the algebraic structure of the recursive construction
of counterterms. By giving this derivation we will establish the relation
of the pAQFT formalism to the {}``Hopf algebra school'' which was
not present in the original pAQFT article (cf.~\citep[p.~45]{Brunetti2009}).

After some preliminary remarks on the differential calculus used in
this chapter, we will cite the main theorem of renormalization from
\citep{Brunetti2009} in the second section. The third section will
be devoted to the derivation of a forest formula for regularized Epstein-Glaser
renormalization from Stora's main theorem. The above described Hopf
algebra will be constructed in the fourth section of this chapter.

\section{\label{sec:PreliminariesDifferentialCalculus}Preliminaries on differential
calculus}

We take the elevator in the hierarchy of differential calculi one
floor up and want to regard functional derivatives of the $\Sm$-matrix,
regarded as a map between spaces of (local) functionals,\[
\Sm\equiv\exp_{\dT}:\,\fps{\mathcal{F}_{\loc}(\M)}{\hbar}\rightarrow\fps{\mathcal{F}(\M)}{\hbar}\,,\]
and of the renormalization group transformations $Z\in\R$ to be defined
below as maps,\[
Z:\fps{\mathcal{F}_{\loc}(\M)}{\hbar}\rightarrow\fps{\mathcal{F}_{\loc}(\M)}{\hbar}\,.\]
The $n$-fold derivative of $\Sm$ at the origin \[
\Sm^{\left(n\right)}(0)(F^{\otimes n})=\frac{d^{n}}{d\lambda^{n}}\Sm(\lambda F)\bigg|_{\lambda=0}\]
gives the $n$-fold time-ordered product, i.e. the $n$th coefficient
in the series expansion of $\Sm$, cf.~Equation~(\ref{eq:PerturbationSeries}).
The $n$-fold derivative of $Z\in\R$ gives the counterterm at order
$n$ of causal perturbation theory. We will equivalently use $\Sm^{\left(n\right)}=\Sm^{\left(n\right)}\big|_{0}\equiv\Sm^{\left(n\right)}(0)$,
and likewise for $Z$, wherever there is no risk of confusion.

The mathematically precise definition of such a differential calculus
is quite involved and a focus of research in analysis \citep{Hamilton1982,KrieglMichor1997,Neeb2005}
(taking the stairs here, might be very hard). However, it is enough
for our purposes to assume that a calculus can be defined in such
a way that the corresponding differential, \[
\frac{\delta}{\delta F}:\Sm\mapsto\Sm^{\left(1\right)}(F)\,,\]
fulfills the chain - and the Leibniz rule in the sense below. A calculus
fulfilling the chain rule was defined for locally convex spaces in
\citep{Neeb2005}. And as shown in \citep[Sec.~3.1]{Brunetti2009}
$\mathcal{F}(\M)$ can be endowed with a locally convex topology,
defined as the initial topology of the Hörmander topology on spaces
of distributions with conic wave front set.

Let $Z$ be differentiable at $F$, and $\Sm$ be differentiable at
$Z(F)$, then we want to assume that the derivative of their composition
is given by the chain rule\begin{equation}
\left(\Sm\circ Z\right)^{\left(1\right)}\!(F)=\left(\Sm^{\left(1\right)}\circ Z\right)\!(F)\cdot\left(Z^{\left(1\right)}(F)\right)\equiv\frac{\delta\Sm}{\delta F'}\bigg|_{Z(F)}\cdot\frac{\delta Z}{\delta F'}\bigg|_{F}\,,\label{eq:ChainRule}\end{equation}
where on the right hand side we have a composition of linear maps,
generally denoted by {}``~$\cdot$~'' in this chapter. For the
iteration of the chain rule and the proof of the $n$-fold chain rule
in Lemma~\ref{lem:FaaDiBrunoSetPartitionVersion} we will also need
that the derivative fulfills the Leibniz rule in the following sense,\begin{equation}
\left(Z\otimes Z\right)^{\left(1\right)}=Z^{\left(1\right)}\otimes Z+Z\otimes Z^{\left(1\right)}\,.\label{eq:LeibnizRule}\end{equation}

We call a map \[
\bld{\Psi}:\fps{\mathcal{F}_{\loc}(\M)}{\hbar}\rightarrow\fps{\mathcal{F}(\M)}{\hbar}\]
analytic (at $F$), if the $n$th functional derivative exists for
all $n\in\mathbb{N}$ as a totally symmetric, linear map \[
\bld{\Psi}^{\left(n\right)}(F):\fps{\mathcal{F}_{\loc}(\M)}{\hbar}^{\otimes n}\rightarrow\fps{\mathcal{F}(\M)}{\hbar}\,,\]
and\begin{equation}
\bld{\Psi}^{\left(n\right)}(F):\fps{\mathcal{F}_{\loc}(\M)}{\hbar}^{\otimes n}\rightarrow\fps{\mathcal{F}_{\loc}(\M)}{\hbar}\,,\quad\mbox{ if }\im(\bld{\Psi})\subset\fps{\mathcal{F}_{\loc}(\M)}{\hbar}.\label{eq:-n-foldDerivativeLocalMap}\end{equation}

\section{\label{sec:MainTheoremRenormalization}The Main Theorem of Renormalization}

An important insight in perturbative renormalization theory is the
fact that the freedom in the definition of the $\Sm$-matrix can be
described in terms of the Stückelberg-Petermann renormalization group
\citep{Stueckelberg1953}. Popineau and Stora termed this fact the
{}``main theorem of perturbative renormalization theory'' \citep{StoraPopineau1982}.
One can find it, although not under this name, already in the early
literature of renormalization theory \citep{Gell-MannLow1954,BogoliubovShirkov1959}.
Modern versions are included in \citep{Pinter2001,Grigore2001b}.
The precise statement and proof of this theorem in the algebraic approach
to perturbative QFT \citep{DuetschFredenhagen2004,Duetsch2005} made
it possible to show that the renormalization group of Stückelberg
and Petermann provides a common basis also to other renormalization
groups found in literature \citep{Brunetti2009}. We will give here
a minimalistic review of the basic definitions needed to formulate
the main theorem of renormalization in pAQFT. A more detailed summary,
including a sketch of the proof is contained in Section~4.1 of \citep{Brunetti2009}.

The $\Sm$-matrix\[
\begin{array}{rccl}
\Sm: & \fps{\mathcal{F}_{\loc}(\M)}{\hbar} & \rightarrow & \fps{\mathcal{F}(\M)}{\hbar}\\
 & F & \mapsto & \Sm(F)=\exp_{\dT}(F)\end{array}\]
is analytic at the origin, where its derivatives are given by the
$n$-fold time ordered products. However, $\Sm$ is not unique, but
needs to be defined perturbatively by renormalization. As shown in
\citep{Brunetti2009} the prerequisites needed for a definition of
$\Sm$ within causal perturbation theory can be expressed directly
in terms of properties of the $\Sm$-matrix itself. Let $A,B,F\in\fps{\mathcal{F}_{\loc}(\M)}{\hbar}$,
then $\Sm$ is required to fulfill the following conditions,
\begin{description}[topsep=1mm]
\item [{{[}C1{]}~Causality.}] $\Sm(A+B)=\Sm(A)\star\Sm(B)$, if $\supp(A)\gtrsim\supp(B)$.
\item [{{[}C2{]}~Starting~Element.}] $\Sm(0)=1:\E(\M)\rightarrow1\in\mathbb{C}$,\\
$\phantom{.}\hspace{18.5mm}\Sm^{\left(1\right)}(0)=\id:\fps{\mathcal{F}_{\loc}(\M)}{\hbar}\rightarrow\fps{\mathcal{F}_{\loc}(\M)}{\hbar}$.
\item [{{[}C3{]}~$\varphi$-Locality.}] The value of $\Sm(F)\in\fps{\mathcal{F}(\M)}{\hbar}$
at a given field configuration $\varphi_{0}$ depends only on the
Taylor expansion of $F\in\fps{\mathcal{F}_{\loc}(\M)}{\hbar}$ around
$\varphi_{0}$, \[
\Sm(F)(\varphi_{0})=\Sm(F_{\varphi_{0}}^{\left[N\right]})(\varphi_{0})+\mathcal{O}(\hbar^{N+1})\,,\]
where ${\displaystyle F_{\varphi_{0}}^{\left[N\right]}(\varphi)=\sum_{n=0}^{N}\frac{1}{n!}\left\langle F^{\left(n\right)}(\varphi_{0}),\left(\varphi-\varphi_{0}\right)^{\otimes n}\right\rangle }$
denotes the Taylor expansion of $F$ up to order $N$.
\item [{{[}C4{]}~Field~Independence.}] $\Sm$ depends only implicitly,
i.e. via the interaction $F$, on the field configuration, \[
\forall\psi:{\displaystyle \left\langle \frac{\delta\Sm(F)}{\delta\varphi},\psi\right\rangle =\Sm^{\left(1\right)}\Big|_{F}(\left\langle \frac{\delta F}{\delta\varphi},\psi\right\rangle })\,.\]

\end{description}
While {[}C1{]} and {[}C2{]} are directly related to the inductive
procedure of Epstein-Glaser, condition {[}C3{]} implies that only
finitely many terms will contribute if one cuts the perturbative expansion
of the $\Sm$-matrix at a given order in $\hbar$, see also the discussion
at the end of Chapter~\ref{cha:DimReg-PositionSpace}. This makes
it possible to regard also more general, and in particular non-polynomial
interactions $F\in\fps{\mathcal{F}_{\loc}(\M)}{\hbar}$ in pAQFT.
Furthermore {[}C3{]} implies together with the fourth condition {[}C4{]}
the Wick expansion formula for the time-ordered product of Epstein
and Glaser \citep{Epstein1973}. This is needed to reduce the problem
of renormalizing $\Sm$ to an extension problem for distributions.
See the discussion in \citep[Sec.~4.1]{Brunetti2009} and also \citep[Sec.~4.B]{Keller2009}.

The freedom in the definition of the $\Sm$-matrix is described by
the Stückelberg-Petermann renormalization group $\R$. In the framework
of perturbative Algebraic Quantum Field Theory $\R$ is the group
of analytic maps\[
Z:\fps{\mathcal{F}_{\loc}(\M)}{\hbar}\rightarrow\fps{\mathcal{F}_{\loc}(\M)}{\hbar}\,,\]
with composition as group operation, and $Z\in\R$ having the following
properties,
\begin{description}[topsep=1mm]
\item [{{[}RG1{]}}] $Z(0)=0$
\item [{{[}RG2{]}~Starting~Element.}] $Z^{\left(1\right)}(0)=\id$
\item [{{[}RG3{]}}] $Z=\id+\mathcal{O}(\hbar)$
\item [{{[}RG4{]}~Locality.}] Let $A,B,C\in\fps{\mathcal{F}_{\loc}(\M)}{\hbar}$
with $\supp(A)\cap\supp(C)=\emptyset$, then\[
Z(A+B+C)=Z(A+B)-Z(B)+Z(B+C)\]

\item [{{[}RG5{]}~$\varphi$-Locality.}] $Z(F)(\varphi_{0})=Z(F_{\varphi_{0}}^{\left[N\right]})(\varphi_{0})+\mathcal{O}(\hbar^{N+1})$
\item [{{[}RG6{]}~Field~Independence.}] $Z$ depends only implicitly
on the field $\varphi$, \[
\forall\varphi\in\E(\M):{\displaystyle \frac{\delta Z}{\delta\varphi}=0}\,.\]

\end{description}
With these definitions at hand, we can now formulate
\begin{thm}
[{Main Theorem of Renormalization, cf.~\citep[Thm.~4.1]{Brunetti2009}}]
\label{thm:MainTheoremRenormalization}Given two $\Sm$-matrices $\Sm$
and $\widehat{\Sm}$ satisfying the conditions Causality, Starting
Element, $\varphi$-Locality, and Field Independence, {[}C1{]}-{[}C4{]},
there exists a unique $Z\in\R$ such that\begin{equation}
\widehat{\Sm}=\Sm\circ Z\,.\label{eq:MainTheoremRenormalization}\end{equation}
Conversely, given an $\Sm$-matrix $\Sm$ satisfying {[}C1{]}-{[}C4{]}
and a $Z\in\R$, then (\ref{eq:MainTheoremRenormalization}) defines
a new $\Sm$-matrix satisfying conditions {[}C1{]}-{[}C4{]}.
\end{thm}
We will be interested in this chapter mainly in a special class of
scattering matrices, which we define now.
\begin{defn}
[{Analytically Regularized $\Sm$-matrix}]\label{def:AnalyticallyRegularizedSmatrix}Any
scattering matrix, $\Sm_{\kappa}$, which fulfills the conditions
{[}C1{]}-{[}C4{]} and depends analytically on an additional parameter
$\kappa\in\Omega\backslash\left\{ 0\right\} \subset\mathbb{C}$, such
that for all $n\in\mathbb{N}$, $n\geq2$, the $n$-fold functional
derivative, \[
\Sm_{\kappa}^{\left(n\right)}(0):\fps{\mathcal{F}_{\loc}(\M)}{\hbar}^{\otimes n}\rightarrow\fps{\mathcal{F}(\M)}{\hbar}\,,\]
is the analytic regularization of a time-ordered product outside the
large diagonal in the sense of Proposition~\ref{pro:RegularizedAmplitude},
we want to call an \emph{analytically regularized $\Sm$-matrix.}
\end{defn}
Observe that the definition implies that the second derivative $\Sm_{\kappa}^{\left(2\right)}(0)$
corresponds to an analytic regularization in the strict sense of Definition~\ref{def:Regularization}.
In the functional framework this implies\[
\pp(\Sm_{\kappa}^{\left(2\right)}(0)):\fps{\mathcal{F}_{\loc}(\M)}{\hbar}^{\otimes2}\rightarrow\fps{\mathcal{F}_{\loc}(\M)}{\hbar}\,,\]
cf.~Section~\ref{sec:MS-For-TProduct}. An example for such an analytically
regularized $\Sm$-matrix is the unique dimensionally regularized
$\Sm$-matrix $\Sm_{\mu,\zeta}$ of Definition~\ref{def:DimRegSMatrix}.
This follows directly from its construction, since it was defined
using the methods of Epstein-Glaser renormalization. However, that
$\Sm_{\mu,\zeta}$ fulfills {[}C1{]}-{[}C4{]} is also readily seen
from its perturbative expansion (\ref{eq:PerturbationSeries}). And
we will take $\Sm_{\mu,\zeta}$ as an example, wherever it is necessary
to introduce a regularization in the discussion below.

By the above theorem, the Stückelberg-Petermann renormalization group
acts transitively on all $\Sm$-matrices fulfilling {[}C1{]}-{[}C4{]}.
Thus, if we want to find a finitely regularized $\Sm$-matrix $\Sm_{\mu,\zeta,\ren}$
which also fulfills {[}C1{]}-{[}C4{]} we will have to construct an
element $Z_{\mu,\zeta}$ of the Stückelberg-Petermann renormalization
group, such that \begin{equation}
\Sm_{\mu,\zeta,\ren}=\Sm_{\mu,\zeta}\circ Z_{\mu,\zeta}\,.\label{eq:FiniteRegSMatrixByRGTransform}\end{equation}
has a limit $\zeta\rightarrow0$ in the set of $\Sm$-matrices. That
is\begin{equation}
\Sm_{\mu,\ren}:=\lim_{\zeta\rightarrow0}\left(\Sm_{\mu,\zeta}\circ Z_{\mu,\zeta}\right)\label{eq:RenSMatrixByRGTransform}\end{equation}
exists in the sense of formal power series in $\hbar$ term by term
in the perturbative expansion; see also \citep[Sec.~5.2]{Brunetti2009}.
In Epstein-Glaser renormalization the construction of these local
counterterms, i.e., the perturbative definition of the map $Z$ has
to be done recursively, i.e., term by term in the perturbative expansion
starting with the counterterm $Z^{\left(2\right)}$ for $\Sm^{\left(2\right)}(0)$.
There will be a choice involved in each step of this recursion and
hence it is impossible to express it in an algorithm which computes,
say, the $n$'th counterterm. However, in the case we dispose of a
regularization, $\Sm_{\mu,\zeta}$, the second term, $\Sm_{\mu,\zeta}^{\left(2\right)}$,
is a regularization in the strict sense, and hence we have a preferred
choice for the local counterterm,\[
Z_{\mu,\zeta}^{\left(2\right)}=-\pp(\Sm_{\mu,\zeta}^{\left(2\right)}):\fps{\mathcal{F}_{\loc}(\M)}{\hbar}^{\otimes2}\rightarrow\fps{\mathcal{F}_{\loc}(\M)}{\hbar}\,.\]
which renormalizes the time ordered product, \[
\Sm_{\mu,\zeta,\ren}^{\left(2\right)}:=\rp(\Sm_{\mu,\zeta}^{\left(2\right)})=\left(1-\pp\right)\Sm_{\mu,\zeta}^{\left(2\right)}\,.\]
It will be show in the following section that this preferred choice
can be done at all orders of perturbation theory in a consistent way,
i.e., with local counterterms at all orders. This, in turn, makes
it possible to solve the recursive renormalization procedure of Epstein-Glaser
in quite the same way as it was done by Zimmermann in 1969 for BPH
in momentum space. We will derive from Equation~(\ref{eq:FiniteRegSMatrixByRGTransform})
a forest formula for Epstein-Glaser renormalization which solves the
recursive construction of counterterms to all orders in causal perturbation
theory. We want to remark that the choices at all orders are unique
in the minimal subtraction scheme, such that this leads to a recursive
procedure, which, in principle, can also be taught to a computer -
in contrast to the original Epstein-Glaser method.

The relation of the presented method to the modern formulation of
renormalization in terms of Hopf algebras will be given in Section~\ref{sec:MoreThanHopfAlgebra}.

\section{\label{sec:EGFF}A Forest Formula for Epstein-Glaser Renormalization}

Since we will stay in the functional framework throughout the derivation
of the forest formula, the result will be valid independent of the
chosen representation, in particular it holds for momentum space as
well as position space, whatever is the best suited representation
for the regularization. Furthermore, it is formulated without regard
to the graph expansion of the time-ordered product. Partitions will
take the place of graphs as the basic combinatorial objects. However,
analogous to the discussion in Section~\ref{sec:MS-For-TProduct},
the forest formula for Epstein-Glaser renormalization also holds in
a graph by graph manner, and then implies Zimmermann's forest formula
of \citep{Zimmermann1969} enhanced by his discussion on spurious
subtractions in \citep{Zimmermann1975}; see also Proposition~\ref{pro:RedundantProjections}
and the discussion thereafter. However, the combinatorial structures
used here will make the role of forests in Zimmermann's formula even
more transparent.

That a version of Zimmermann's forest formula should also exist in
position space was observed before. And the assertion is natural considering
the common origin of BPHZ and Epstein-Glaser renormalization. It has
been shown that Zimmermann's Taylor subtractions with respect to external
momenta of the graphs correspond to the $W$-projections in the Epstein-Glaser
framework \citep{Prange1999,Prange2000}. Gracia-Bondía and Lazzarini
gave a direct translation of this {}``Taylor surgery'' (GB) to position
space by considering a more general test function space for the {}``infrared
regulators'', i.e. the test functions $w_{\alpha}$ of Lemma~\ref{lem:W-Projection-Functions},
in fact they allowed the $w_{\alpha}$ to be distributions of the
Cesàro type \citep{Gracia-Bondia2003,GraciaBondiaLazzarini2003}.
A translation of the complete forest formula to position space was
given by Steinmann in the case of QED \citep{Steinmann2000}. However,
Steinmann's treatment was unsatisfactory in two points. First, Steinmann's
formulation involves the differentiation of (generalized) functions
at singular points. This was recognized by the author himself and
is due to the fact that the implicit regularization%
\footnote{I hope this is the only spot in the thesis where I use the word {}``regularization''
only in the sense of {}``making things well-defined''.%
} of the momentum space framework, namely the fact that Zimmermann
performs his manipulations on the integral kernel of the convolution
rather than the integral itself, has no counterpart in position space.
Momentum space convolution corresponds to the pointwise product (of
distributions) in position space and it is partly due to this implicit
regularization that momentum space integrals were introduced in perturbative
quantum field theory in the first place \citep{Bogoliubow1957}. Second,
Steinmann regards Quantum Electro Dynamics (QED). The fact that QED
has only one basic vertex of valence three implies that there are
no graphs with less lines but the same set of vertices so that the
spurious subtractions do not occur in QED and other theories {}``of
graphical $\varphi^{3}$-type''. Consequently, Steinmann's version
of the forest formula cannot be considered as a complete translation
of Zimmermann's forest formula (which treats general graphs in $\mathcal{G}$)
to position space. Observe that Zimmermann implements a preferred
choice for the extension at all orders in perturbation theory by performing
his Taylor subtractions always at zero external momentum. In order
to define this rigorously he has to introduce additional maps which
conceal part of the underlying pattern. However, as already remarked
above and as will be clear from the construction below, such a choice
of extension at all orders of perturbation theory is indispensable
for the solution of the recursive procedure of Bogoliubov, Parasiuk
and Hepp, or Epstein and Glaser, respectively. We start by exploring
the termwise structure of the main theorem (Theorem~\ref{thm:MainTheoremRenormalization}),
by applying the $n$-fold derivative to (\ref{eq:MainTheoremRenormalization}).
The Faà~di~Bruno formula arises naturally.

\subsection{Faà~di~Bruno's formula}

In 1855 Francesco Faà~di~Bruno proved a formula for the $n$-fold
chain rule \citep{FaaDiBruno1855}. And it is quite appealing that
this old formula, when applied to Equation~(\ref{eq:MainTheoremRenormalization})
gives a termwise version of the main theorem of perturbative renormalization
(Theorem~\ref{thm:MainTheoremRenormalization}). Considering the
time since its first proof, there are quite a few versions of Faà~di~Bruno's
formula in the literature today. However, in order to keep the relation
to causal perturbation theory and BPHZ renormalization visible at
all steps in our calculation, a set partition version of the form
found in \citep[p.~219]{Johnson2002} seems to be the most appropriate.
We prove here an adjusted version. But let us first give an easy definition,
mainly to fix notation.
\begin{defn}
[Partition, Blocks]\label{def:Partition}By a \emph{partition} $\Ptn$
of a finite set $V$ we mean any set of non-empty, disjoint subsets
$V_{i}\subset V$, $i\in I$, such that\[
V=\dot{\bigcup_{i\in I}}V_{i}\,,\qquad\mbox{that is, }\quad\Ptn=\left\{ V_{i}:i\in I\right\} ,\]
where $\dot{\cup}$ denotes disjoint union. We refer to the non-empty,
disjoint subsets $V_{i}$ as \emph{blocks} of $\Ptn$, and denote
the set of all partitions of $V$ by $\Part V$.
\end{defn}
We generally consider partitions of the set which corresponds to the
set of vertices, $V(\Gamma)$, in the graphical representation, and
in most cases it will be more convenient to regard instead the set
of numbers $\left\{ 1,\dots,n\right\} $. However, this implicit numbering
of vertices, is irrelevant for the derivation due to the symmetry
of the functional derivative briefly introduced in Section~\ref{sec:PreliminariesDifferentialCalculus}.
\begin{lem}
[Main Theorem - termwise]\label{lem:FaaDiBrunoSetPartitionVersion}Let
$\Sm:\fps{\mathcal{F}_{\loc}(\M)}{\hbar}\rightarrow\fps{\mathcal{F}(\M)}{\hbar}$
be an $\Sm$-matrix fulfilling conditions {[}C1{]}-{[}C4{]} and let
$Z:\fps{\mathcal{F}_{\loc}(\M)}{\hbar}\rightarrow\fps{\mathcal{F}_{\loc}(\M)}{\hbar}$
be an element of the Stückelberg-Petermann renormalization group,
$Z\in\R$. Then the $n$th term in the perturbative expansion of the
transformed $\Sm$-matrix, $\widehat{\Sm}=\Sm\circ Z$, is given by
\begin{equation}
\left(\Sm\circ Z\right)^{\left(n\right)}\!(0)=\sum_{\Ptn\in\Part\{1,\dots,n\}}\!\!\Sm^{\left(\left|\Ptn\right|\right)}(0)\cdot\left(\bigotimes_{I\in\Ptn}\left[Z^{\left(\left|I\right|\right)}(0)\right]\right),\label{eq:FaaDiBrunoPartitionVersion}\end{equation}
where the sum is over all partitions $\Ptn$ of the index set $\left\{ 1,\cdots,n\right\} $
into $\left|\Ptn\right|$ blocks, \linebreak[4]and ''~$\cdot$~''
denotes the composition of linear maps,\begin{equation}
\fps{\mathcal{F}_{\loc}(\M)}{\hbar}^{\otimes n}\xrightarrow{\bigotimes_{I\in\Ptn}Z^{\left(\left|I\right|\right)}(0)}\fps{\mathcal{F}_{\loc}(\M)}{\hbar}^{\otimes\left|\Ptn\right|}\xrightarrow[\hphantom{\bigotimes_{I\in\Ptn}Z^{\left(\left|I\right|\right)}(0)}]{\Sm^{\left(\left|\Ptn\right|\right)}(0)}\fps{\mathcal{F}(\M)}{\hbar}.\label{eq:CompositionDiagram}\end{equation}
\end{lem}
\begin{proof}
We prove (\ref{eq:FaaDiBrunoPartitionVersion}) by induction following
\citep{Johnson2002}. For $n=1$ we have,\[
\left(\Sm\circ Z\right)^{\left(1\right)}\!\bigg|_{F}=\Sm^{\left(1\right)}\bigg|_{Z(F)}\cdot Z^{\left(1\right)}\bigg|_{F}\,.\]
However, evaluating at $F=0$ gives an empty assertion ($\id=\id$)
due to the starting element conditions {[}C2{]}, {[}RG1{]}, and {[}RG2{]}.
The first non-trivial contribution is from the second derivative,
which we compute explicitly for illustration,\begin{align*}
\left(\Sm\circ Z\right)^{\left(2\right)}\!\bigg|_{F} & =\left(\Sm^{\left(1\right)}\bigg|_{Z(F)}\cdot Z^{\left(1\right)}\bigg|_{F}\right)^{\left(1\right)}\!\bigg|_{F}\\
 & =\Sm^{\left(2\right)}\bigg|_{Z(F)}\cdot Z^{\left(1\right)}\bigg|_{F}\otimes Z^{\left(1\right)}\bigg|_{F}+\Sm^{\left(1\right)}\bigg|_{Z(F)}\cdot Z^{\left(2\right)}\bigg|_{F}\,,\end{align*}
where the two terms correspond to the two partitions $\left\{ \left\{ 1\right\} ,\left\{ 2\right\} \right\} $
and $\left\{ \left\{ 1,2\right\} \right\} $ of $\left\{ 1,2\right\} $.
Evaluating at $F=0$ gives, again by using {[}C2{]}, {[}RG1{]}, and
{[}RG2{]}, \begin{align*}
\left(\Sm\circ Z\right)^{\left(2\right)}\!(0) & =\Sm^{\left(2\right)}(0)+Z^{\left(2\right)}(0)\,,\end{align*}
and $Z^{\left(2\right)}=\left(\Sm\circ Z\right)^{\left(2\right)}-\Sm^{\left(2\right)}$
is found to be the counterterm at second order.

For the induction step regard the derivative of (\ref{eq:FaaDiBrunoPartitionVersion}).
By (\ref{eq:ChainRule})/(\ref{eq:LeibnizRule}) we get, \begin{align}
\left[\Sm^{\left(\left|\Ptn\right|\right)}\!\Big|_{Z(F)}\!\left(\bigotimes_{I\in\Ptn}\left[Z^{\left(\left|I\right|\right)}\Big|_{F}\right]\right)\right]^{\left(1\right)} & =\Sm^{\left(\left|\Ptn\right|+1\right)}\Big|_{Z(F)}\!\left(Z^{\left(1\right)}\Big|_{F}\otimes\bigotimes_{I\in\Ptn}\left[Z^{\left(\left|I\right|\right)}\Big|_{F}\right]\right)\nonumber \\
 & \hspace{-13mm}+\sum_{I'\in\Ptn}\Sm^{\left(\left|\Ptn\right|\right)}\Big|_{Z(F)}\!\left(Z^{\left(\left|I'\right|+1\right)}\Big|_{F}\otimes\bigotimes_{I\in\Ptn\backslash\left\{ I'\right\} }\left[Z^{\left(\left|I\right|\right)}\Big|_{F}\right]\right).\label{eq:FaaDiBrunoPartitionVersion-ProofInductionStep}\end{align}
Any partition $\widetilde{\Ptn}$ of $\left\{ 1,\dots,n+1\right\} $
can be written in terms of a partition $\Ptn$ of\linebreak{}
 $\left\{ 1,\dots,n\right\} $ by either adjoining $\left\{ n+1\right\} $
as a block of its own, or by adding $\left\{ n+1\right\} $ to one
of the blocks in $\Ptn$, i.e., \[
\widetilde{\Ptn}=\Ptn\cup\left\{ \left\{ n+1\right\} \right\} \qquad\mbox{or}\qquad\widetilde{\Ptn}=\left(\Ptn\backslash\left\{ I'\right\} \right)\cup\left\{ I'\cup\left\{ n+1\right\} \right\} \,,\]
for some block $I'\in\Ptn$. Thus equation (\ref{eq:FaaDiBrunoPartitionVersion-ProofInductionStep})
contains all partitions of $\left\{ 1,\dots,n+1\right\} $ which can
be obtained from $\Ptn$. Evaluating at $F=0$ gives the result.
\end{proof}
Equation (\ref{eq:FaaDiBrunoPartitionVersion}) describes the action
of the St{\"u}ckelberg-Petermann group on time-ordered products.
This action followed directly from the main theorem by applying Fa{\`a}~di~Bruno's
formula. Since Equation~(\ref{eq:FaaDiBrunoPartitionVersion}) is
not the most cited version of Faà~di~Bruno's formula, we want to
show that it reduces to the more prevalent versions if we evaluate
$\left(\Sm\circ Z\right)^{\left(n\right)}\!(0)$ at the $n$-fold
tensor power of one and the same interaction functional, $F^{\otimes n}\in\fps{\mathcal{F}_{\loc}(\M)}{\hbar}^{\otimes n}$.
Due to the symmetry of the functional derivative partitions with identical
block sizes will give the same contribution to $\left(\Sm\circ Z\right)^{\left(n\right)}\!(0)(F^{\otimes n})$.
So the question is: How many of them are there? Depending on how one
chooses to sort these partitions, one gets the different versions
of Faà~di~Bruno's formula. As an example we give one of the derivations.

Let $\Ptn\in\Part\{1,\dots,n\}$ be a partition with $\left|\Ptn\right|=k$
blocks. Let $\left(l_{1},\dots,l_{k}\right)\in\mathbb{N}^{k}$ denote
the sizes of these blocks, $l_{1}+\cdots+l_{k}=n$, $l_{i}\geq1$.
There are ${n \choose l_{1},\dots,l_{k}}$ possibilities to distribute
$n$ elements among $k$ different blocks of specified size and order.%
\footnote{${n \choose l_{1},\dots,l_{k}}:=\frac{n!}{l_{1}!\cdots l_{k}!}$ denotes
the \emph{multinomial coefficient}, see, e.g., \citep{HarrisHirstMossinghoff2008}. %
} However, for a partition the order of the blocks is irrelevant. In
a sum over the multiindex $\left(l_{1},\dots,l_{k}\right)\in\mathbb{N}^{k}$
we thus have to divide by the number of permutations of $\left\{ l_{1},\dots,l_{k}\right\} $
to reduce it to a sum over all partitions. In total we get\begin{align}
 & \left(\Sm\circ Z\right)^{\left(n\right)}\!(F^{\otimes n})\label{eq:FaaDiBrunoCauchyProductFormula}\\
 & \qquad=\sum_{k=1}^{n}\frac{1}{k!}\Sm^{\left(k\right)}\cdot\sum_{\substack{l_{1}+\cdots+l_{k}=n\\
l_{i}\geq1}
}{n \choose l_{1},\dots,l_{k}}Z^{\left(l_{1}\right)}(F^{\otimes l_{1}})\otimes\cdots\otimes Z^{\left(l_{k}\right)}(F^{\otimes l_{k}})\,,\nonumber \end{align}
which was the starting point in \citep{FigueroaGraciaBondiaVarilly2005}
for the derivation of\begin{align}
 & \left(\Sm\circ Z\right)^{\left(n\right)}\!(F^{\otimes n})\label{eq:FaaDiBrunoSortedBlockSize}\\
 & \qquad=\sum_{k=1}^{n}\sum_{\lambda_{1},\dots,\lambda_{n}}\frac{n!}{\lambda_{1}!\cdots\lambda_{n}!}\Sm^{\left(k\right)}\!\cdot\left(\frac{Z^{\left(1\right)}}{1!}\right)^{\otimes\lambda_{1}}\negmedspace\otimes\cdots\otimes\negmedspace\left(\frac{Z^{\left(n\right)}}{n!}\right)^{\otimes\lambda_{n}}\!(F^{\otimes n}).\nonumber \end{align}
Here $\lambda_{l}\in\mathbb{N}_{0}$ denotes the number of blocks
of size $l$. Equation (\ref{eq:FaaDiBrunoSortedBlockSize}) is the
version, which is usually cited as Faà~di~Bruno's formula in the
literature, and often is the starting point for the introduction of
the Faà~di~Bruno bialgebra. We will learn more about this bialgebra
in Section~\ref{sec:MoreThanHopfAlgebra}.

\subsection{Minimal Subtraction}

Regard an analytically regularized $\Sm$-matrix, e.g., $\Sm_{\mu,\zeta}$.
Then Lemma~\ref{lem:FaaDiBrunoSetPartitionVersion} implies a recursion
relation for the counterterms $Z_{\mu,\zeta}^{\left(n\right)}$ in
the minimal subtraction renormalization scheme introduced in the previous
chapter.
\begin{cor}
[Recursion Relation for MS Counterterms]\label{cor:RecursionMSCounterterms}
In the minimal subtraction renormalization scheme (MS), a recursion
relation for the counterterms $Z_{\mu,\zeta}^{\left(n\right)}$ is
given by\begin{equation}
Z_{\mu,\zeta}^{\left(n\right)}=-\pp\sum_{\Ptn\in\Part\{1,\dots,n\}\backslash\left\{ \Ptn_{1}\right\} }\!\!\Sm_{\mu,\zeta}^{\left(\left|\Ptn\right|\right)}\cdot\left(\bigotimes_{I\in\Ptn}Z_{\mu,\zeta}^{\left(\left|I\right|\right)}\right).\label{eq:CountertermsMSRecursionFormula}\end{equation}
The counterterms are local, and all counterterms on the right hand
side are of lower order than $n$, since the only partition in $\Part\{1,\dots,n\}$
with a block containing $n$ elements is removed from the sum, $\Ptn_{1}=\left\{ \left\{ 1,\dots,n\right\} \right\} $.\end{cor}
\begin{proof}
Using the starting element condition {[}C2{]} we get from (\ref{eq:FaaDiBrunoPartitionVersion}),
\[
\left(\Sm_{\mu,\zeta}\circ Z_{\mu,\zeta}\right)^{\left(n\right)}=Z_{\mu,\zeta}^{\left(n\right)}+\sum_{\Ptn\in\Part\{1,\dots,n\}\backslash\left\{ \Ptn_{1}\right\} }\!\!\Sm_{\mu,\zeta}^{\left(\left|\Ptn\right|\right)}\cdot\left(\bigotimes_{I\in\Ptn}Z_{\mu,\zeta}^{\left(\left|I\right|\right)}\right).\]
By assumption $Z_{\mu,\zeta}\in\R$ is a renormalization group transformation
which renders the limit $\zeta\rightarrow0$ of the left hand side
finite, cf.~Equation~(\ref{eq:RenSMatrixByRGTransform}). Thus the
counterterm $Z_{\mu,\zeta}^{\left(n\right)}$ has to subtract at least
the principal part of the sum on the right hand side. In the minimal
subtraction scheme, $Z_{\mu,\zeta}^{\left(n\right)}$ is fixed by
the requirement to remove exactly the principal part, hence formula
(\ref{eq:CountertermsMSRecursionFormula}). The locality of $Z_{\mu,\zeta}^{\left(n\right)}$
is implied directly by the functional calculus, cf.~(\ref{eq:-n-foldDerivativeLocalMap}).
\end{proof}
The expert reader%
\footnote{I assume here that $\left\{ \mbox{readers}\right\} \neq\emptyset$.
If you have a proof, please tell me: \href{mailto:kai.johannes.keller@desy.de}{kai.johannes.keller@desy.de}%
} readily recognizes the similarity of (\ref{eq:CountertermsMSRecursionFormula})
to the recursive formula for the antipode in the Faà~di~Bruno bialgebra.
However, observe that there are two products involved in the recursion
for the counterterms. The tensor product $\otimes$ and the composition
of linear maps {}``~$\cdot$~''. What might be obvious for the
expert, namely that this is a structure which cannot be described
by a commutative Hopf algebra alone, will be derived {}``by foot''
in Section~\ref{sec:MoreThanHopfAlgebra}. However, let us first
give the derivation of a forest formula for the $n$-fold finitely
regularized time-ordered product $\Sm_{\mu,\zeta,\ren}^{\left(n\right)}$.
The forest formula will solve the inductive construction of the renormalization
group transformation $Z_{\mu,\zeta}$, which renders $\Sm_{\mu,\zeta,\ren}=\Sm_{\mu,\zeta}\circ Z_{\mu,\zeta}$
finite (in the sense of formal power series in $\hbar$) in the limit
$\zeta\rightarrow0$.

\subsection{Derivation of the forest formula}

We begin with the definition of a forest as it was given by Zimmermann.
We will then derive some relations to partitions which will motivate
the definition of an Epstein-Glaser forest and make it possible to
give a transparent proof of the forest formula.

Let us start with Zimmermann's definition, which was contained in
\citep{Zimmermann1969}. However, we incorporate directly his results
from \citep{Zimmermann1975} and will only consider subgraphs, which
are given by a subset of the set of vertices of a graph $\Gamma$
and all lines in $\Gamma$ connecting them, cf. Definition~\ref{def:Subgraph}.
Zimmermann called such subgraphs \emph{full vertex parts}. Observe
that this makes it possible to work directly with the set of vertices
instead of the set of general graphs or full vertex parts. However,
to keep the relation to the original definitions transparent, let
$\Gamma\in\mathcal{G}$ be a graph.

A \emph{$\Gamma$-forest} $U$ is a set of subgraphs $\gamma\subseteq\Gamma$,
such that any two elements $\gamma,\gamma'\in U$ are \emph{non-overlapping},
i.e.\begin{equation}
\mbox{either}\quad\gamma\subset\gamma'\quad\mbox{or}\quad\gamma'\subset\gamma\quad\mbox{or}\quad\gamma\cap\gamma'=\emptyset\,.\label{eq:ForestProperty}\end{equation}
The empty set is referred to as the empty forest. The notation $\gamma\cap\gamma'=\emptyset$
means that $V(\gamma)\cap V(\gamma')=\emptyset$, and it follows from
the definition of a subgraph (Definition~\ref{def:Subgraph}) that
then also the sets of edges are disjoint, $E(\gamma)\cap E(\gamma')=\emptyset$.
A graph $\gamma\in U$ is called \emph{maximal} if there is no other
graph in $U$ containing it. A $\Gamma$-forest $U$ is called \emph{maximal},
if there is no other $\Gamma$-forest containing it. A forest is called
\emph{restricted}, if it contains only {}``divergent graphs'', $\div(\gamma)\geq0$,
cf. (\ref{eq:DegreeOfDivergenceGraph}). Note that a (Feynman-) graph
$\Gamma$ has more than one maximal restricted forest, if and only
if it has \emph{overlapping divergences}, i.e., at least two divergent
subgraphs $\gamma,\gamma'\subset\Gamma$ for which (\ref{eq:ForestProperty})
does not hold. Given a forest $U$ of $\Gamma$ and a subgraph $G\in U$,
we define the set\[
U(G):=\left\{ \gamma\in U:\gamma\subseteq G\right\} ,\]
and note that $U(G)$ is a $G$-forest as well as a $\Gamma$-forest.
\begin{lem}
[Structure of Maximal Forests]\label{lem:StructureMaxForests}Let
$U$ be a maximal $\Gamma$-forest.
\begin{enumerate}
\item \label{enu:ContainsComplement}With any element $G\in U$, $U$ contains
also its vertex complement, $\Gamma\obs G\in U$.
\item \label{enu:MaxSubForest}For any graph $G\in U$ the set $U(G)=\left\{ \gamma\in U:\,\gamma\subseteq G\right\} $
is a maximal $G$-forest.
\item \label{enu:DisjointUnion}The forest $U\backslash\left\{ \Gamma\right\} $
is the disjoint union of two maximal forests. If $G\in U\backslash\left\{ \Gamma\right\} $
is a maximal element, then\[
U\backslash\left\{ \Gamma\right\} =U(G)\dot{\cup}U(\Gamma\obs G)\,.\]

\end{enumerate}
\end{lem}
\begin{proof}
(\ref{enu:ContainsComplement}). Let $G\in U$, then for any element
$\gamma\in U$, we have either $\gamma\subset G$ or $\gamma\subset\Gamma\obs G$.
Hence $U\cup\left\{ \Gamma\obs G\right\} $ is a forest, and by maximality
of $U$: $\Gamma\obs G\in U$.

(\ref{enu:MaxSubForest}). Let $U'$ be a $G$-forest properly containing
$U(G)$. Then there is a subgraph $\gamma'\subset G$ such that $\gamma'\notin U(G)$
does not overlap with any of the elements in $U(G)$. Since $G\in U$
it follows that $\gamma'\notin U$ is non-overlapping with any element
in $U$. Thus $U\dot{\cup}\left\{ \gamma'\right\} $ is a forest,
in contradiction with the maximality of $U$.

(\ref{enu:DisjointUnion}) follows from (\ref{enu:ContainsComplement})
and (\ref{enu:MaxSubForest}).
\end{proof}
The combinatorial result, which leads to the forests formula for Epstein-Glaser
renormalization, is that we can write (maximal) forests of the kind
defined above in terms of (complete) sets of partitions of the vertex
set.
\begin{lem}
[Partitions and Forests]\label{lem:MaximalForestsCompleteTotallyOrderedPartitions}The
set of partitions of $n$ elements,\linebreak[4] $\Part\{1,\dots,n\}$,
is a partially ordered set (poset). The partial order is defined by
saying that $\Ptn_{1}$ is finer than $\Ptn_{2}$ (and $\Ptn_{2}$
is coarser than $\Ptn_{1}$),\[
\Ptn_{1}\leq\Ptn_{2}\,,\]
if for any block $I\in\Ptn_{1}$ there is a block $J\in\Ptn_{2}$
containing $I$. $\left(\Part\{1,\dots,n\},\leq\right)$ is a complete
partial order (cpo) with finest element $\left\{ \left\{ 1\right\} ,\dots,\left\{ n\right\} \right\} $
and coarsest element $\left\{ \left\{ 1,\dots,n\right\} \right\} $.

Let $\Pto\subset\Part\{1,\dots,n\}$ be a totally ordered subset.
We call $\Pto$ maximal, if there is no totally ordered subset $\Pto'\subset\Part\{1,\dots,n\}$
containing $\Pto$. The union of any totally ordered subset $\Pto$
is a forest,\[
U(\Pto):=\bigcup_{\Ptn'\in\Pto}\Ptn'\,.\]
$U(\Pto)$ is a maximal forest, if $\Pto$ is maximal.\end{lem}
\begin{proof}
{}``$\leq$'' is reflexive ($\Ptn\leq\Ptn$), antisymmetric ($\Ptn_{1}\leq\Ptn_{2}\,\wedge\,\Ptn_{2}\leq\Ptn_{1}$
$\Rightarrow$ $\Ptn_{1}=\Ptn_{2}$), and transitive ($\Ptn_{1}\leq\Ptn_{2}\,\wedge\,\Ptn_{2}\leq\Ptn_{3}$
$\Rightarrow$ $\Ptn_{1}\leq\Ptn_{3}$), hence a partial order. For
$\left(\Part\{1,\dots,n\},\leq\right)$ to be a cpo, we have to show
that any pair $\Ptn_{1},\Ptn_{2}\in\Part\{1,\dots,n\}$ has a least
upper bound $\Ptn_{1}\sqcup\Ptn_{2}$ and a greatest lower bound $\Ptn_{1}\sqcap\Ptn_{2}$
in $\Part\{1,\dots,n\}$. The assertion is trivial, if $\Ptn_{1}$
and $\Ptn_{2}$ are related by {}``$\leq$'', hence let $\Ptn_{1}$
and $\Ptn_{2}$ not be related by {}``$\leq$''. $\Ptn_{1}\sqcup\Ptn_{2}$
is the partition where all overlapping blocks (cf.~(\ref{eq:ForestProperty}))
are replaced by their union; filled up with the larger blocks of either
$\Ptn_{1}$ or $\Ptn_{2}$. In $\Ptn_{1}\sqcap\Ptn_{2}$ overlapping
blocks are replaced by their intersection and filled up with the smaller
sets of either $\Ptn_{1}$ or $\Ptn_{2}$, see Figure~\ref{fig:PartitionBounds}.
Since $\Ptn_{1}\sqcup\Ptn_{2}$ and $\Ptn_{1}\sqcap\Ptn_{2}$ are
partitions of $\left\{ 1,\dots,n\right\} $, we infer that $\Part\{1,\dots,n\}$
is a cpo.

Regard a totally ordered subset $\Pto\subset\Part\{1,\dots,n\}$ and
let $\Ptn\in\Pto$, then by definition $I\cap J=\emptyset$ for all
$I,J\in\Ptn$. Let $\Ptn,\Ptn'\in\Pto$ be two different partitions,
$\Ptn'\neq\Ptn$, then we have either $I\cap J=\emptyset$, $I\subset J$,
or $J\subset I$ for any pair $\left(I,J\right)\in\Ptn\times\Ptn'$
since $\Pto$ is totally ordered. Thus $I$ and $J$ are non-overlapping,
and $U(\Pto)$ is a forest. Let $\Pto$ be maximal, then it contains
$\left\{ \left\{ 1,\dots,n\right\} \right\} $ and $\left\{ \left\{ 1\right\} ,\dots,\left\{ n\right\} \right\} $,
and with any partition $\Ptn\in\Pto$ it contains all partitions which
can be constructed out of $\Ptn$ by successively dividing any of
its blocks into a pair of disjoint subsets. Let $I=I_{1}\dot{\cup}I_{2}$
be an index set, then there is no partition $\Ptn_{I}$ of $I$ such
that $\left\{ I_{1},I_{2}\right\} \leq\Ptn_{I}\leq\left\{ I\right\} $
and $\left\{ I_{1},I_{2}\right\} \neq\Ptn_{I}\neq\left\{ I\right\} $.
Hence the maximal set $\Pto$ can be constructed out of $\left\{ \left\{ 1,\dots,n\right\} \right\} $
by the successive division procedure described above. Conversely,
let $U$ be a maximal forest, by Lemma~\ref{lem:StructureMaxForests},
$U$ can be constructed in exactly the same way.
\end{proof}
\addtocounter{figure}{1}

\begin{figure}[th]
\includegraphics[width=5cm]{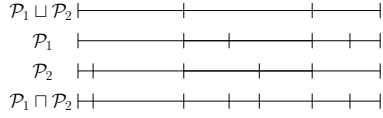}

\caption{\label{fig:PartitionBounds}Least upper and greatest lower bound of
a pair of partitions.}

\end{figure}

We will have to regard in the sequel unions of totally ordered sets
of partitions, which is non-trivial, since the set union of two totally
ordered sets of partitions will not be a totally ordered set of partitions,
just as the union of forests is not a forest in general.
\begin{defn}
[Position and Disjoint Union]\label{def:PositionDisjointUnionPto}Let
$\Pto\subset\Part\{1,\dots,n\}$ be a totally ordered set, then we
define the \emph{position} of any element $\Ptn\in\Pto$ by\[
\pos_{\Pto}(\Ptn):=\left|\left\{ \Ptn'\in\Pto:\Ptn\leq\Ptn'\right\} \right|.\]
It is easy to see that $\pos_{\Pto}(\Ptn_{c})=1$, if $\Ptn_{c}$
is the coarsest element of $\Pto$.

Let $\Pto(I)$ and $\Pto(J)$ be totally ordered subsets of $\Part I$
and $\Part J$, respectively. Any subset of \[
\left\{ \Ptn\dot{\cup}\Ptn':\Ptn\in\Pto(I)\mbox{ and }\Ptn'\in\Pto(J)\right\} ,\]
which is a totally ordered set of partitions of $I\dot{\cup}J$, we
call a \emph{disjoint union} of $\Pto(I)$ and $\Pto(J)$.
\end{defn}
Observe that there are forests which do not correspond to a totally
ordered subset of partitions, e.g., the empty forest or any forest
containing just one proper subset of $\left\{ 1,\dots,n\right\} $.
We now come to the definition of an Epstein-Glaser forest; a similar
definition for forests was considered in \citep{Figueroa2005} to
establish the relation to incidence Hopf algebras.
\begin{defn}
[Epstein-Glaser Forest]Let $V=\left\{ 1,\dots,n\right\} $ be a (vertex)
set. Then we call any totally ordered subset $\F$ of the set of partitions,
$\F\subset\Part\{1,\dots,n\}$, containing the finest partition $\Ptn_{n}:=\left\{ \left\{ 1\right\} ,\dots,\left\{ n\right\} \right\} $
of $V$, an \emph{Epstein-Glaser forest (EG forest)}, i.e., $\F$
has the form\[
\F=\left\{ \cdots\geq\Ptn_{n}\right\} ,\quad\Ptn_{n}=\left\{ \left\{ 1\right\} ,\dots,\left\{ n\right\} \right\} .\]
The EG forest containing only $\Ptn_{n}$ we denote by $\F_{n}:=\left\{ \Ptn_{n}\right\} $.
If an EG forest contains the coarsest partition $\Ptn_{1}:=\left\{ \left\{ 1,\dots,n\right\} \right\} $
we call it a \emph{full EG forest} \emph{(full forest)} and write
$\overline{\F}$. If an EG forest does not contain $\Ptn_{1}$, we
call it a \emph{normal EG forest (normal forest)} and write $\underline{\F}$.
For $n=1$ there is just one forest, the one with one vertex, and
we define this forest to be full. For $n>1$ there is a one to one
correspondence between full and normal forests, given by\[
\overline{\F}=\underline{\F}\cup\left\{ \Ptn_{1}\right\} \,.\]
An Epstein-Glaser forest $\F$ is called maximal, if $\F$ is maximal
as a totally ordered set of partitions.\end{defn}
\begin{cor}
\label{cor:NormalForestDisjointUnion}Any normal Epstein-Glaser forest
$\underline{\F}$ is a disjoint union of at least two full Epstein-Glaser
forests. If the coarsest partition in $\underline{\F}$ has $k$ elements,
then $\underline{\F}$ decomposes into $k$ full forests $\overline{\F}_{1},\dots,\overline{\F}_{k}$,
$k\geq2$, and we call $\underline{\F}$ $k$-fold connected. Conversely
one can say that $\underline{\F}$ is a disjoint union of $\overline{\F}_{1},\dots,\overline{\F}_{k}$
in the sense of Definition~\ref{def:PositionDisjointUnionPto},\begin{equation}
\underline{\F}=\dot{\bigcup_{i\in\left\{ 1,\dots,k\right\} }}\overline{\F}_{i}\,.\label{eq:NormalForestDecomposition}\end{equation}
If $I_{i}$ is the block of the coarsest partition in the full forest
$\overline{\F}_{i}$, then $\overline{\F}_{i}=\underline{\F}(I_{i})$.
In this sense the decomposition (\ref{eq:NormalForestDecomposition})
of $\underline{\F}$ is unique. See also Figure~\ref{fig:DecompositionAndDisjointUnion}.
\qed
\end{cor}
\begin{figure}
\includegraphics[height=2cm]{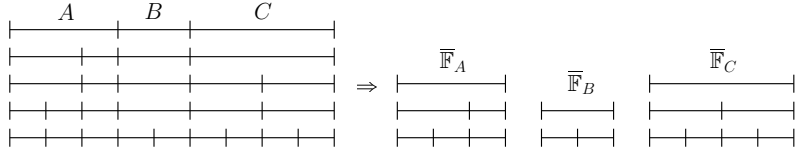}

\caption{\label{fig:DecompositionAndDisjointUnion}The depicted normal forest
of $A\dot{\cup}B\dot{\cup}C$ decomposes into tree full forests $\overline{\F}_{A}$,
$\overline{\F}_{B}$, and $\overline{\F}_{C}$. Conversely, given
the three full forests, any composition of them, which preserves the
order of partitions in each component separately, gives a different
disjoint union, i.e. a different normal forest of $A\dot{\cup}B\dot{\cup}C$.}

\end{figure}

So far for the combinatorial part, we now have to define the analytic
part, namely, the minimal subtractions in the blocks of a partition.
\begin{defn}
[MS Operator]\label{def:MS-OperatorPartitions}For any partition
$\Ptn\in\Part\{1,\dots,n\}$ define the minimal subtraction operator
(MS operator)\[
-T_{\Ptn}^{\MS}\Sm_{\mu,\zeta}^{\left(n\right)}:=\Sm_{\mu,\zeta}^{\left(\left|\Ptn\right|\right)}\cdot\left(\bigotimes_{I\in\Ptn}R_{\left|I\right|}\left(\Sm_{\mu,\zeta}^{\left(\left|I\right|\right)}\right)\right),\quad\mbox{where }R_{k}=\begin{cases}
\id & \mbox{if }k=1\\
-\pp & \mbox{if }k>1\,,\end{cases}\]
cf.~Definition~\ref{def:MS-OperatorSubsets}.
\end{defn}
This defines the operator $-T_{\Ptn}^{\MS}$ on the whole regularized
time-ordered product $\Sm_{\mu,\zeta}^{\left(n\right)}$. Observe,
however, that the above definition implies that we have chosen another
regularization parameter in each block $I$ of the partition; there
is one operator $R_{\left|I\right|}$ for each block. To consider
all partitions, which possibly contribute to the principal part, we
will regard the $n$-fold regularized time-ordered product as being
regularized in ${n \choose 2}$ different regularization parameters,
one for each pair of vertices. This is certainly possible regarding
the fact that we have {}``regularized the lines'' of each given
diagram, i.e. the propagators. At the stage of only one MS operator
$-T_{\Ptn}^{\MS}$ this consideration is not very important, since
the singularities in different blocks are independent, anyway. However,
the fact that we can choose the regularization parameters freely for
any pair of vertices becomes important as soon as we want to define
products of the MS operators applied to the same regularized time-ordered
product. Such products occur in the forest formula below, and we briefly
discuss one example in order to clarify this point. We choose the
position space representation for convenience. Let $\Ptn_{1}\leq\Ptn_{2}$
be different partitions in $\Part\{1,\dots,n\}$. Let $I_{1}\subset I_{2}$
be a pair of blocks, $I_{1}\in\Ptn_{1}$, $I_{2}\in\Ptn_{2}$, and
regard a special graph contribution, say $S_{\Gamma}^{\mu,\zeta}$,
to $\Sm_{\mu,\zeta}^{\left(n\right)}$ for simplicity, $\left|V(\Gamma)\right|=n$.
Let $\gamma$ be the full vertex part (in $\Gamma$) to the vertex
set $I_{1}$ and let $G$ be the full vertex part of $I_{2}$, then
$\gamma\subset G$. Regard the successive subtraction\[
s_{\Gamma\bbs G}^{\mu,\zeta}\, R_{I_{2}}\left(s_{G\bbs\gamma}^{\mu,\zeta}\, R_{I_{1}}s_{\gamma}^{\mu,\zeta}\right),\hspace{80mm}\]
and insert the definition from the previous chapter, cf.~(\ref{eq:ppSubgraphLaurentSeries}),\begin{align*}
 & =\sum_{n=-\infty}^{-1}\zeta^{n}\frac{1}{2\pi i}\oint_{C}d\xi\,\frac{1}{\xi^{n+1}}\sum_{k=-\infty}^{-1}\xi^{k}\frac{1}{2\pi i}\oint_{C'}d\xi'\,\frac{1}{{\xi'}^{k+1}}s_{\Gamma\bbs G}^{\mu,\zeta}\cdot s_{G\bbs\gamma}^{\mu,\xi}\cdot s_{\gamma}^{\mu,\xi'}\,,\end{align*}
where we had to introduce regularization parameters $\xi$ and $\xi'$
for the subgraphs in order to get independent subtractions in all
subgraphs. With these remarks concerning the regularization, we have
for the composition of MS operators corresponding to related partitions,
$\Ptn\geq\Ptn'$,\[
T_{\Ptn}^{\MS}T_{\Ptn'}^{\MS}\Sm_{\mu,\zeta}^{\left(n\right)}=\Sm_{\mu,\zeta}^{\left(\left|\Ptn\right|\right)}\cdot\left[\bigotimes_{I\in\Ptn}R_{\left|I\right|}\left[\Sm_{\mu,\zeta}^{\left(\left|I'/I\right|\right)}\cdot\left(\bigotimes_{I'\subset I}R_{\left|I'\right|}\left(\Sm_{\mu,\zeta}^{\left(\left|I'\right|\right)}\right)\right)\right]\right],\]
where $\left|I'/I\right|$ denotes the cardinality of $\left\{ I'\in\Ptn':I'\subset I\right\} $.
\begin{thm}
[Forest Formula for Epstein-Glaser Renormalization]\label{thm:EGFF}Let
$T_{\Ptn}^{\MS}$ be the minimal subtraction operator of Definition~\ref{def:MS-OperatorPartitions},
and let the product of two MS operators corresponding to related partitions
be defined as described above. Then\begin{equation}
\Sm_{\mu,\zeta,\ren}^{\left(n\right)}:=\sum_{\F\subset\Part\{1,\dots,n\}}\left(\oprod{\Ptn\in\F}-T_{\Ptn}^{\MS}\right)\Sm_{\mu,\zeta}^{\left(n\right)}\label{eq:EGFF}\end{equation}
gives a finite regularization of the $n$-fold regularized time-ordered
product $\Sm_{\mu,\zeta}^{\left(n\right)}\equiv\Time_{\mu,\zeta}^{n}$.
The sum is taken over all Epstein-Glaser forests. The product of the
operators is taken in the order prescribed by {}``$\geq$'', such
that the coarsest partition in $\F$ stands to the very left.\end{thm}
\begin{proof}
The forest formula implies an expression for the $n$th counterterm
in the renormalization scheme of analytic regularization and minimal
subtraction. We can split (\ref{eq:EGFF}) into a sum over full and
normal forests, \begin{align*}
\Sm_{\mu,\zeta,\ren}^{\left(n\right)} & =\sum_{\underline{\F}}\left(\oprod{\Ptn\in\underline{\F}}-T_{\Ptn}^{\MS}\right)\Sm_{\mu,\zeta}^{\left(n\right)}+\sum_{\overline{\F}}\left(\oprod{\Ptn\in\overline{\F}}-T_{\Ptn}^{\MS}\right)\Sm_{\mu,\zeta}^{\left(n\right)}\,.\end{align*}
Observe that for $n=1$ there are no normal forests, and the first
sum is empty. Since any forest in the second sum on the right hand
side contains the coarsest partition, ${\Ptn_{1}=\left\{ \left\{ 1,\dots,n\right\} \right\} }$,
we can factor out the corresponding $\MS$ operator and get from {[}C2{]},\begin{align*}
\Sm_{\mu,\zeta,\ren}^{\left(n\right)} & =\sum_{\underline{\F}}\left(\oprod{\Ptn\in\underline{\F}}-T_{\Ptn}^{\MS}\right)\Sm_{\mu,\zeta}^{\left(n\right)}+R_{n}\left[\sum_{\underline{\F}}\left(\oprod{\Ptn\in\underline{\F}}-T_{\Ptn}^{\MS}\right)\Sm_{\mu,\zeta}^{\left(n\right)}\right].\end{align*}
We show in the sequel, that\begin{equation}
C_{\mu,\zeta}^{\left(n\right)}:=R_{n}\left[\sum_{\underline{\F}}\left(\oprod{\Ptn\in\underline{\F}}-T_{\Ptn}^{\MS}\right)\Sm_{\mu,\zeta}^{\left(n\right)}\right]\equiv\sum_{\overline{\F}}\left(\oprod{\Ptn\in\overline{\F}}-T_{\Ptn}^{\MS}\right)\Sm_{\mu,\zeta}^{\left(n\right)},\label{eq:ProofLocalCountertermC}\end{equation}
is the local counterterm of regularized causal perturbation theory
in the minimal subtraction scheme. That is $C_{\mu,\zeta}^{\left(n\right)}$
fulfills the recursion relation of Corollary~\ref{cor:RecursionMSCounterterms}.
We proceed by induction.

For $n=1$ we have by {[}C2{]} and the definition of $R_{1}$,\[
C_{\mu,\zeta}^{\left(1\right)}=R_{1}\Sm_{\mu,\zeta}^{\left(1\right)}=\id=Z_{\mu,\zeta}^{\left(1\right)}\,.\]
Thus $C_{\mu,\zeta}^{\left(1\right)}=\id:\fps{\mathcal{F}_{\loc}(\M)}{\hbar}\rightarrow\fps{\mathcal{F}_{\loc}(\M)}{\hbar}$
is local, and we can assume that $C_{\mu,\zeta}^{\left(k\right)}=Z_{\mu,\zeta}^{\left(k\right)}$
for all $k<n$. For the induction step we have to show that\[
R_{n}\left[\sum_{\underline{\F}}\left(\oprod{\Ptn\in\underline{\F}}-T_{\Ptn}^{\MS}\right)\Sm_{\mu,\zeta}^{\left(n\right)}\right]=-\pp\sum_{\Ptn\in\Part\{1,\dots,n\}\backslash\left\{ \Ptn_{1}\right\} }\!\!\Sm_{\mu,\zeta}^{\left(\left|\Ptn\right|\right)}\cdot\left(\bigotimes_{I\in\Ptn}C_{\mu,\zeta}^{\left(\left|I\right|\right)}\right).\]
for $n\geq2$. By the definition of $R_{n}$ this is the case, if\begin{equation}
\sum_{\underline{\F}}\left(\oprod{\Ptn\in\underline{\F}}-T_{\Ptn}^{\MS}\right)\Sm_{\mu,\zeta}^{\left(n\right)}=\sum_{\Ptn\in\Part\{1,\dots,n\}\backslash\left\{ \Ptn_{1}\right\} }\!\!\Sm_{\mu,\zeta}^{\left(\left|\Ptn\right|\right)}\cdot\left(\bigotimes_{I\in\Ptn}C_{\mu,\zeta}^{\left(\left|I\right|\right)}\right).\label{eq:EGFF-Proof-A}\end{equation}
Regard the left hand side of this equation. Any normal forest is a
disjoint union of at least two full forests (Corollary~\ref{cor:NormalForestDisjointUnion}).
Hence we can write the sum over all normal forests as \begin{align}
\mbox{lhs} & =\sum_{k=2}^{n}\,\sum_{\substack{\Ptn'\in\Part\{1,\dots,n\}\\
\left|\Ptn'\right|=k}
}\,\sum_{\underline{\F}=\dot{\bigcup}_{I\in\Ptn'}\underline{\F}(I)}\left(\,\oprod{\Ptn\in\underline{\F}}-T_{\Ptn}^{\MS}\right)\Sm_{\mu,\zeta}^{\left(n\right)}\,.\nonumber \\
\intertext{The\, product\, splits\, and\, we\, get} & =\sum_{k=2}^{n}\sum_{\substack{\Ptn'\in\Part\{1,\dots,n\}\\
\left|\Ptn'\right|=k}
}\sum_{\underline{\F}=\dot{\bigcup}_{I\in\Ptn'}\underline{\F}(I)}\bigotimes_{I\in\Ptn'}\left(\,\,\,\,\oprod{\Ptn\in\underline{\F}(I)}-T_{\Ptn}^{\MS}\right)\Sm_{\mu,\zeta}^{\left(n\right)}\,.\label{eq:EGFF-Proof-B}\end{align}
The sum over all normal, i.e., multiply connected forests $\underline{\F}$
can be performed by summing over all full forests $\underline{\F}(I)$
in its connected components, cf.~Corollary~\ref{cor:NormalForestDisjointUnion}.
And we have to perform the sum in these components, in order to get
a well-defined expression for (\ref{eq:EGFF-Proof-B}),\[
\mbox{lhs}=\sum_{k=2}^{n}\sum_{\substack{\Ptn'\in\Part\{1,\dots,n\}\\
\left|\Ptn'\right|=k}
}\bigotimes_{I\in\Ptn'}\sum_{\underline{\F}(I)}\left(\,\,\,\,\oprod{\Ptn\in\underline{\F}(I)}-T_{\Ptn}^{\MS}\right)\Sm_{\mu,\zeta}^{\left(n\right)}\,.\]
Inserting the inductive assumption, $C_{\mu,\zeta}^{\left(k\right)}=Z_{\mu,\zeta}^{\left(k\right)}$
$\forall k<n$, gives the desired result ($\mbox{lhs}=\mbox{rhs}$).
\end{proof}

\begin{cor}
Let $\F_{1},\dots,\F_{c}$ be the maximal forests of the vertex set
$\left\{ 1,\dots,n\right\} $. Then we can write (\ref{eq:EGFF})
equivalently as\begin{equation}
\Sm_{\mu,\zeta,\ren}^{\left(n\right)}=\sum_{\substack{\emptyset\neq\left\{ i_{1},\dots,i_{\nu}\right\} \\
\subset\left\{ 1,\dots,c\right\} }
}\left(-1\right)^{\nu+1}\left(\,\,\,\,\oprod{\Ptn\in\F_{i_{1}}\cap\cdots\cap\F_{i_{\nu}}}\left(1-T_{\Ptn}^{\MS}\right)\right)\Sm_{\mu,\zeta}^{\left(n\right)}\,.\label{eq:EGFF-1-T}\end{equation}
\end{cor}
\begin{proof}
This is merely another way of summing up the contributions to (\ref{eq:EGFF}).
Multiplying out the factors $\left(1-T_{\Ptn}^{\MS}\right)$ gives
the result, cf. \citep[Thm.~3.3]{Zimmermann1969}.
\end{proof}
{}
\begin{cor}
[EG Forest Formula in terms of graphs]Let $\Gamma$ be a graph with
$n$ vertices, $\left|V(\Gamma)\right|=n$. For any partition $\Ptn$
of $V(\Gamma)$ the action of the MS operator on the level of graphs
is given by \[
-T_{\Ptn}^{\MS}\left\langle s_{\Gamma}^{\mu,\zeta}\right\rangle :=\left\langle s_{\Gamma/\Ptn}^{\mu,\zeta}\right\rangle \left(\bigotimes_{I\in\Ptn}R_{\left|I\right|}\left\langle s_{\gamma_{I}}^{\mu,\zeta}\right\rangle \right),\quad\mbox{where }R_{k}=\begin{cases}
\id & \mbox{if }k=1\\
-\pp & \mbox{if }k>1\,,\end{cases}\]
where $\Gamma/\Ptn$ is the graph with the blocks $I\in\Ptn$ as vertices
and as lines all lines in $\Gamma$ which connect different blocks
of $\Ptn$. For each block $I\in\Ptn$, the graph $\gamma_{I}$ is
the full vertex part of $I$. $\left\langle s_{\Gamma}^{\mu,\zeta}\right\rangle $
denotes the analytically regularized amplitude in any representation
(momentum or position space). Then the limit\[
\left\langle s_{\Gamma,\ren}^{\mu}\right\rangle =\lim_{\zeta\rightarrow0}\sum_{\F\subset\Part\{1,\dots,n\}}\left(\oprod{\Ptn\in\F}-T_{\Ptn}^{\MS}\right)\left\langle s_{\Gamma}^{\mu,\zeta}\right\rangle \]
is well-defined and gives a UV finite, i.e., renormalized amplitude.\end{cor}
\begin{proof}
The MS operators are tensor products of the corresponding operators
on sets (Definition~\ref{def:MS-OperatorSubsets}) and thus linear.
Hence the corollary is a direct consequence of the discussion given
in Section~\ref{sec:MS-For-TProduct} and the above Theorem~\ref{thm:EGFF}.
\end{proof}
{}
\begin{cor}
[Prepared Amplitude]\label{cor:PreparedAmplitude}Let $\Gamma$ be
a graph with $n$ vertices, $n>1$, i.e., $\left\langle S_{\Gamma}^{\mu,\zeta},\bld{\delta}^{\alpha}\right\rangle $
with $\alpha\in\mathbb{N}^{n}$ and $\left|\alpha\right|=2\left|E(\Gamma)\right|$
is a contribution to the $n$-fold, regularized time ordered product
\[
\Sm_{\mu,\zeta}^{\left(n\right)}:\fps{\mathcal{F}_{\loc}(\M)}{\hbar}^{\otimes n}\rightarrow\fps{\mathcal{F}(\M)}{\hbar}.\]
Then the prepared amplitude to $\Gamma$ is given by,\begin{equation}
S_{\Gamma,\prep}^{\mu,\zeta}:=\sum_{\substack{\emptyset\neq\left\{ i_{1},\dots,i_{\nu}\right\} \\
\subset\left\{ 1,\dots,c\right\} }
}\left(-1\right)^{\nu+1}\left(\,\,\,\,\oprod{\Ptn\in\F_{i_{1}}\cap\cdots\cap\F_{i_{\nu}}\backslash\left\{ \Ptn_{1}\right\} }\left(1-T_{\Ptn}^{\MS}\right)\right)S_{\Gamma}^{\mu,\zeta}\,.\qed\label{eq:TransInvPreparedAmplitude}\end{equation}

\end{cor}

\section{\label{sec:MoreThanHopfAlgebra}More than Hopf Algebra}

The investigation of the combinatorial structure of perturbative renormalization
theory is a vivid field of research in mathematical physics and for
the characterization of the underlying pattern Hopf algebras, and
in particular the \linebreak[4]Faà~di~Bruno bialgebra introduced
by Joni and Rota in 1982 \citep{JoniRota1982}, became more and more
important in recent years \citep{Figueroa2005}. We will show in this
last section, that there is a more intricate pattern underlying the
combinatorial structure of renormalization than is described by a
Hopf algebra. However, the relation to the Hopf algebra of graphs
originally encountered by Connes and Kreimer in BPHZ \citep{Kreimer1998,Connes2000,Connes2001}
and later by Gracia-Bondía, Lazzarini, and Pinter in Epstein-Glaser
renormalization \citep{Gracia-Bondia2000,Pinter2000b} will become
transparent. The attractive feature of our derivation is that we can
understand the emerging Hopf algebraic structure as a direct consequence
of the main theorem of renormalization (Theorem~\ref{thm:MainTheoremRenormalization}).
The elements of this Hopf algebra will be differential operators which
give the time-ordered products, when acting on a $\Sm$-matrix, and
local counterterms when acting on an element of the Stückelberg-Petermann
renormalization group. Hence, in a graphical representation they correspond
to sums of graphs with the same set of vertices. The Hopf algebra
structure for individual graphs is regained by linearity. The regularized
Feynman rules and the renormalization map will appear naturally as
soon as one specializes to an analytically regularized $\Sm$-matrix.
This is a major difference to the Connes-Kreimer approach, where the
Feynman rules had to be put by hand as characters into the commutative
ring of Laurent series. We will show, however, that the recursion
formula for minimally subtracted counterterms (\ref{eq:CountertermsMSRecursionFormula}),
which was seen to be a direct consequence of the main theorem, cannot
be described within the commutative Hopf algebra of Connes and Kreimer
\citep{Connes2000,Connes2001}.

As a first step we will use the Faà~di~Bruno Hopf algebra introduced
by Joni and Rota to derive the commutative, non-cocommutative Hopf
algebra of graphs described briefly above. It will be clear from the
given presentation that one needs an additional, non-commutative,
composition product, to get back the recursion formula for the counterterms
(Corollary~\ref{cor:RecursionMSCounterterms}). We will sketch in
the last section how this composition product can be implemented into
the Hopf algebra to describe algebraically the construction of counterterms
in pAQFT. An interpretation of the maps in terms of graph operations
will be given in the end.

\subsection{The Hopf Algebra}

As shown by Joni and Rota,\linebreak[1] Faà~di~Bruno's formula
for the chain rule gives rise to a natural bialgebra structure, which
the authors called the Faà~di~Bruno bialgebra \citep{JoniRota1982}.
In modern presentations it is often introduced as a bialgebra of the
coefficients in (\ref{eq:FaaDiBrunoSortedBlockSize}), cf.~\citep{Figueroa2005}.
However, to keep the correspondence to causal perturbation theory
transparent also in this last section of the present thesis we regard
instead directly the partition version given in Lemma~\ref{lem:FaaDiBrunoSetPartitionVersion}.
Apart from that we follow essentially the steps of \citep{Figueroa2005}
as far as the Hopf algebra structure is concerned.

Regard Faà~di~Bruno's formula (\ref{eq:FaaDiBrunoPartitionVersion})
in the termwise form of the main theorem of renormalization (Lemma~\ref{lem:FaaDiBrunoSetPartitionVersion}).
We denote the coefficients by \[
a_{n}(\Sm):=\Sm^{\left(n\right)}(0)\qquad\mbox{and}\qquad a_{n}(Z):=Z^{\left(n\right)}(0)\]
and get\begin{equation}
a_{n}(\Sm\circ Z)=\sum_{\Ptn\in\Part\{1,\dots,n\}}a_{\left|\Ptn\right|}(\Sm)\cdot\bigodot_{I\in\Ptn}a_{\left|I\right|}(Z)\,.\label{eq:FaaDiBrunoCoefficients}\end{equation}
We want to make the symmetry of the functional derivative explicit
here and replaced the tensor product in (\ref{eq:FaaDiBrunoPartitionVersion})
by the symmetrized tensor product\[
\bigodot_{i=1}^{k}A_{i}:=\frac{1}{k!}\sum_{\sigma\in\Perm(k)}\bigotimes_{i=1}^{k}A_{\sigma(i)}\,,\]
where $\Perm(k)$ denotes the group of permutations of $k$ elements.
Observe that besides this commutative product ($\odot$), there is
a second, non-commutative product in formula (\ref{eq:FaaDiBrunoCoefficients}).
Namely, the composition of linear maps\begin{equation}
\linc:a_{\left|\Ptn\right|}(\Sm)\otimes\bigodot_{I\in\Ptn}a_{\left|I\right|}(Z)\mapsto a_{\left|\Ptn\right|}(\Sm)\cdot\bigodot_{I\in\Ptn}a_{\left|I\right|}(Z)\,,\label{eq:CompositionLinearMaps}\end{equation}
cf.~(\ref{eq:CompositionDiagram}). This second product is absent
if we regard the coefficients in (\ref{eq:FaaDiBrunoCoefficients})
as scalars. And we will sketch how to implement this additional non-commutative
product into the Hopf algebra in the next section. Let us first regard
the commutative part.

Regard the coefficients $a_{k}$ in (\ref{eq:FaaDiBrunoCoefficients})
as differential operators\begin{equation}
a_{k}:\Sm\mapsto a_{k}(\Sm)\equiv\Sm^{\left(k\right)}(0)\quad\mbox{and}\quad a_{k}:Z\mapsto a_{k}(Z)\equiv Z^{\left(k\right)}(0)\,.\label{eq:FdBCoefficientsDifferentialOperators}\end{equation}
The operators will produce multi-linear maps on local functionals;
with local image in the case they act on $Z\in\R$,\[
a_{k}(Z):\fps{\mathcal{F}_{\loc}(\M)}{\hbar}^{\otimes k}\rightarrow\fps{\mathcal{F}_{\loc}(\M)}{\hbar}\,,\]
and with possibly non-local image if they act of $\Sm$ fulfilling
{[}C1{]}-{[}C4{]},\[
a_{k}(\Sm):\fps{\mathcal{F}_{\loc}(\M)}{\hbar}^{\otimes k}\rightarrow\fps{\mathcal{F}(\M)}{\hbar}\,.\]
Denote by $\Halg$ the space of these coefficients $a_{k}$. Since
the $Z^{\left(k\right)}$ and $\Sm^{\left(k\right)}$ (evaluation
at zero understood) are linear maps on tensor products of local functionals,
$\Halg$ carries a natural $\mathbb{C}$-vector space structure induced
by the $\mathbb{C}$-vector space structure on local functionals.
As already remarked above, the symmetry of the functional derivative
induces a commutative product on $\Halg$,\[
\begin{array}{rccl}
M_{\odot}: & \Halg\otimes\Halg & \rightarrow & \Halg\\
 & (a_{k}\otimes a_{l}) & \mapsto & a_{k}\odot a_{l}:=\tfrac{1}{2}\left[a_{k}\otimes a_{l}+a_{l}\otimes a_{k}\right],\end{array}\]
where we set\begin{equation}
\left(a_{k}\odot a_{l}\right)(Z):=a_{k}(Z)\odot a_{l}(Z)\,,\label{eq:CummutativeHproduct}\end{equation}
and likewise for $\Sm$. We regard $\Halg$ as the free, commutative
algebra generated by the $a_{k}$, and commit the usual abuse of notation
by using the same symbol for the (symmetrized) tensor product of linear
maps on the right hand side of (\ref{eq:CummutativeHproduct}) and
the commutative product of the algebra $\left(\Halg,\odot\right)$
on the left. It will be clear from the context, where we mean which.
$\left(\Halg,\odot\right)$ is a unital algebra with unit \[
\openone:Z\mapsto\openone(Z):=\id\,,\qquad\openone:\Sm\mapsto\openone(\Sm):=\id\,,\]
where $\id:\fps{\mathcal{F}_{\loc}(\M)}{\hbar}\rightarrow\fps{\mathcal{F}_{\loc}(\M)}{\hbar}$
denotes the identity map on the space of local functionals. We denote
the corresponding unit map by\[
\begin{array}{rccl}
e: & \mathbb{C} & \rightarrow & \Halg\\
 & \alpha & \mapsto & \alpha\,\openone\,.\end{array}\]
We set $\Halg\odot\openone=\Halg$. Joni and Rota interpreted Faà~di~Bruno's
formula as a coproduct rule for the coefficients $a_{k}$ , and we
can do exactly the same thing here, by defining the coproduct, $\Delta:\Halg\rightarrow\Halg\otimes\Halg$,
as\begin{equation}
\left(\Delta a_{n}\right):=\sum_{\Ptn\in\Part\{1,\dots,n\}}a_{\left|\Ptn\right|}\otimes\left(\bigodot_{I\in\Ptn}a_{\left|I\right|}\right).\label{eq:FaaDiBrunoCoproduct}\end{equation}
It is obvious from (\ref{eq:FaaDiBrunoCoefficients}) that this coproduct
is induced by\[
\left(\Delta a_{n}\right)(\Sm\otimes Z)=a_{n}(\Sm\circ Z)\,,\]
and we break with the tradition of flipping the arguments of $\Delta a_{n}$
here,%
\footnote{see, e.g. page 2 of \citep{FigueroaGraciaBondiaVarilly2005}.%
} since, in contrast to the chain rule for functions or formal power
series with scalar coefficients, we have a composition of linear maps
with a prescribed order rather than a commutative product on the right
hand side of Faà~di~Bruno's formula, cf.~(\ref{eq:CompositionLinearMaps}).
Consequently we have the linear part of the coproduct on the left
hand side of $\otimes$. Equipping $\Halg$ with the counit, defined
on generators \[
\begin{array}{rccl}
\overline{e}: & \Halg & \rightarrow & \mathbb{C}\\
 & a_{n} & \mapsto & \overline{e}(a_{n}):=\begin{cases}
1 & \mbox{if }a_{n}=\openone\\
0 & \mbox{else}\,,\end{cases}\end{array}\]
gives the usual commutative, non-cocommutative Faà~di~Bruno bialgebra,
however, now interpreted in terms of the functional differential operators
(\ref{eq:FdBCoefficientsDifferentialOperators}).

It will be helpful for the construction of the antipode to first discuss
the natural gradings on $\Halg$.%
\footnote{Two articles by Kastler were very helpful in learning about the relevance
of grading and other Hopf related topics \citep{Kastler2000,Kastler2004}.
They are probably not the standard references to be cited at this
point, however, they contain explicit proofs of the results from Hopf
algebra theory needed here. See also more standard literature like
\citep{Sweedler1969,Abe1977}. Since we will not make any connection
to more advanced structures in algebraic geometry, the given references
will fully suffice for the discussion in this section. %
} As any tensor algebra, $\Halg$ is graded by the number of factors,\[
\begin{array}{rccl}
\deg^{\odot}: & \Halg & \rightarrow & \mathbb{N}\\
 & \bigodot_{i=1}^{k}a_{l_{i}} & \mapsto & \deg^{\odot}(a_{l_{1}}\odot\cdots\odot a_{l_{k}}):=k\,.\end{array}\]
That is, $\Halg$ can be written as the direct sum \[
\Halg=\bigoplus_{k=0}^{\infty}\Halg^{\odot k}\quad,\quad\Halg^{\odot k}=\left\{ a\in\Halg:\deg^{\odot}(a)=k\right\} ,\quad\Halg^{\odot0}=\mathbb{C}\,.\]
Subordinate to this tensor algebra grading is a naturally induced
grading of the individual $a_{l}\in\Halg^{\odot1}$ given by the order
of the derivative (minus one), \[
\deg^{v}(a_{l}):=l-1\,,\]
and we have\[
\Halg^{\odot1}=\bigoplus_{n=0}^{\infty}\Halg_{n}^{\odot1}\quad,\quad\Halg_{n}^{\odot1}=\left\{ a\in\Halg^{\odot1}:\deg^{v}(a)=n\right\} .\]

Observe that $\deg^{v}+1$ corresponds to the order of derivative
of $\Sm$ (or $Z$) at zero. This determines the number of interaction
functionals in the argument of the corresponding linear map\[
a_{n}(\Sm)=\Sm^{\left(n\right)}(0):\fps{\mathcal{F}_{\loc}(\M)}{\hbar}^{\otimes n}\rightarrow\fps{\mathcal{F}(\M)}{\hbar}\,,\quad n=\deg^{v}+1\,,\]
and hence $\deg^{v}+1$ is the number of vertices of the graphs contributing
to the graph expansion of $\Sm^{\left(n\right)}(0)\equiv\Time_{n}$,
cf.~Equation~(\ref{eq:TimeOrderedGraphStructure}).

This grading by vertex number can be extended to the tensor product
$\Halg^{\odot k}$ by setting\[
\deg^{v}(a_{k}\odot a_{l}):=\deg^{v}(a_{k})+\deg^{v}(a_{l})\,,\quad\mbox{i.e.,}\quad\deg^{v}(a_{l_{1}}\odot\cdots\odot a_{l_{k}})=\sum_{i=1}^{k}l_{i}-k\,.\]
With this definition also $\deg^{v}$ is an algebra grading of the
algebra $\Halg$, and we have\[
\Halg=\bigoplus_{n=0}^{\infty}\Halg_{n}\quad,\quad\Halg_{n}=\left\{ a\in\Halg:\deg^{v}(a)=n\right\} \,.\]
Furthermore the vertex grading $\deg^{v}$ is compatible with the
coproduct (\ref{eq:FaaDiBrunoCoproduct}),\[
\deg^{v}(\Delta a_{n})=\left|\Ptn\right|-1+\sum_{I\in\Ptn}\left(\left|I\right|-1\right)=n-1\quad\forall\Ptn\in\Part\{1,\dots,n\}\,,\]
and hence $\left(\Halg,\odot,e,\Delta,\overline{e}\right)$ is graded
as a bialgebra. The starting element conditions {[}C2{]} and {[}RG2{]},
\[
a_{1}(\Sm)=\Sm^{\left(1\right)}(0)=\id\quad\mbox{and}\quad a_{1}(Z)=Z^{\left(1\right)}(0)=\id\]
imply that $a_{1}=\openone$ and hence $\Halg$ is $\mathbb{N}_{0}$-graded
connected as a bialgebra, i.e., $\Halg_{0}=\mathbb{C}$ (we implicitly
identify $\openone=\openone\odot\openone$ here). It is a well-known
fact of Hopf algebra theory that any $\mathbb{N}_{0}$-graded connected
bialgebra possesses an antipode \citep[Prop.~2.7]{Kastler2000},\[
\ap:\Halg\rightarrow\Halg\,,\]
and thus $\Halg$ is a Hopf algebra. We will now derive a recursion
formula for this antipode. By definition, the antipode $\ap$ of a
Hopf algebra is the inverse of the identity with respect to the induced
convolution product on the Hopf algebra automorphisms $\Aut(\Halg)$.
The convolution product on $\Aut(\Halg)$ is induced by the product
and coproduct on $\Halg$, we denote it by\[
\phi\oast\psi:=M_{\odot}\circ\left(\phi\otimes\psi\right)\circ\Delta\,,\quad\phi,\psi\in\Aut(\Halg)\,.\]
It is a standard computation to prove that \[
e\circ\overline{e}:\Halg\rightarrow\Halg\]
defines a unit in the algebra $\left(\Aut(\Halg),\oast\right)$. A
similar computation will be done below for the second product on the
Hopf algebra, so we leave it out here. The antipode of an $\mathbb{N}_{0}$-graded
connected bialgebra can then be constructed directly from its defining
condition,\[
\id_{\Halg}\oast\ap=e\circ\overline{e}\,.\]
We have\begin{align*}
M_{\odot}\circ\left(\id_{\Halg}\otimes\ap\right)\circ\Delta(a_{n}) & =e(\overline{e}(a_{n}))\\
\sum_{\Ptn\in\Part\{1,\dots,n\}}a_{\left|\Ptn\right|}\odot\left(\bigodot_{I\in\Ptn}\ap(a_{\left|I\right|})\right) & =\begin{cases}
\openone & \mbox{if }n=1\\
0 & \mbox{else}\,,\end{cases}\end{align*}
since, in particular, $\ap$ is an algebra homomorphism. From the
case $n=1$ we get $\ap(a_{1})=\openone$, and since there is only
one partition of $\left\{ 1,\dots,n\right\} $ with one block, $\Ptn_{1}=\left\{ \left\{ 1,\dots,n\right\} \right\} $,
we infer by connectedness of $\Halg$, i.e., by $a_{1}=\openone$
that\begin{equation}
\ap(a_{n})=-\sum_{\Ptn\in\Part\{1,\dots,n\}\backslash\left\{ P_{1}\right\} }a_{\left|\Ptn\right|}\odot\left(\bigodot_{I\in\Ptn}\ap(a_{\left|I\right|})\right).\label{eq:CKantipode}\end{equation}
Observe the similarity to the recursion for the counterterms in Corollary~\ref{cor:RecursionMSCounterterms}.
However, observe also that the composition structure of (\ref{eq:CountertermsMSRecursionFormula})
is completely absent in (\ref{eq:CKantipode}). This is no problem,
if we regard the Feynman rules as characters of $\left(\Halg,\odot\right)$
into a commutative ring of Laurent series with scalar coefficients,
as it was done in \citep{Connes2000,Connes2001}. We want to emphasize
the relation of the Hopf algebra $\left(\Halg,\odot,e,\Delta,\overline{e},\ap\right)$
to the Connes-Kreimer Hopf algebra of graphs. The elements of $\Halg$
are differential operators whose order is determined by the vertex
grading $\deg^{v}$. By implementing the graph expansion (\ref{eq:TimeOrderedGraphStructure}),
the elements $a_{k}\in\Halg$ can be interpreted as sums over graphs
with the same set of vertices. By linearity of the maps and the fact
that only finitely many graphs contribute to the perturbative expansion
at a given order of $\hbar$, {[}C3{]}, we can break the Hopf algebra
structure down to the level of graphs. However, since the structure
for the algebraic construction of counterterms is not complete yet,
we will give a more detailed account of this interpretation only at
the end of the next section. 

The fact that we found the Hopf algebra structure in the sums of graphs
is in accordance with the results of Brouder and Frabetti, who found
in different examples (including gauge theories) that Connes and Kreimer's
Hopf algebra structure is preserved when one sums up the graph contributions
at certain orders or perturbation theory \citep{BrouderFrabetti2000,BrouderFrabetti2001,Frabetti2007},
see also \citep{Suijlekom2007c}. Brouder and Frabetti, in collaboration
with Krattenthaler and Menous, respectively, also observed the relation
to the Faà~di~Bruno Hopf algebra \citep{BrouderFrabettiKra2006,BrouderFrabettiMenous2009},
however, the relation to the main theorem of perturbative renormalization
as proven in \citep{DuetschFredenhagen2004,Duetsch2005,Brunetti2009}
was, to the best of my knowledge, unobserved before.

We will now incorporate the non-commutative composition structure
into the commutative Hopf algebra $\left(\Halg,\odot,e,\Delta,\overline{e},\ap\right)$
constructed above.

\subsection{Algebraic Construction of Counterterms}

In contrast to the Connes-Kreimer approach to renormalization, in
our approach the Feynman rules are naturally induced as evaluation
maps of the differential operators $a_{n}\in\Halg$. The basic evaluation
operator, which gives the $n$-fold time-ordered product, and thus
corresponds to the Feynman rules is given by \[
\feyn:a_{n}\mapsto a_{n}(\Sm)=\Sm^{\left(n\right)}\equiv\Time^{n}\,.\]
The image of $\feyn$ is a multi-linear map between spaces of (local)
functionals. On linear maps there are two natural products, one is
the (symmetrized) tensor product discussed above, and the other is
the composition. Composition of linear maps is a non-commutative operation
in the generic case, and as a consequence it is impossible to derive
the action of the counterterms $Z^{\left(n\right)}$ on the time-ordered
products $\Sm^{\left(n\right)}$ described by Lemma~\ref{lem:FaaDiBrunoSetPartitionVersion}
from the commutative Hopf algebra $\left(\Halg,\odot,e,\Delta,\overline{e},\ap\right)$
alone. Regard once again the expression given in the lemma,

\[
\left(\Sm\circ Z\right)^{\left(n\right)}\!(0)=\sum_{\Ptn\in\Part\{1,\dots,n\}}\!\!\Sm^{\left(\left|\Ptn\right|\right)}(0)\cdot\left(\bigodot_{I\in\Ptn}\left[Z^{\left(\left|I\right|\right)}(0)\right]\right),\]
and observe that we need both, the commutative product $\odot$ and
the \linebreak[4]non-commutative composition {}``~$\cdot$~'',
as well as the coproduct $\Delta$ for defining the action of $\R$
on itself and on the set of $\Sm$-matrices term by term in an algebraic
fashion.

We want to incorporate the composition as an additional product in
the commutative, non-cocommutative Hopf algebra constructed above,
\[
\left(\Halg,\odot,e,\Delta,\overline{e},\ap\right)\]
with generators $a_{k}\in\Halg$, $k\in\mathbb{N}$. We define a map\[
\linc:\Halg\otimes\Halg\rightarrow\Halg\,,\]
induced by the composition of linear maps (\ref{eq:CompositionLinearMaps}),
\[
\linc(a_{k}\otimes\bigodot_{i=1}^{k}a_{l_{i}})=a_{k}\ocomp\bigodot_{i=1}^{k}a_{l_{i}}\,,\qquad\mbox{with}\qquad\left(a_{k}\ocomp\bigodot_{i=1}^{k}a_{l_{i}}\right)(Z):=a_{k}(Z)\cdot\bigodot_{i=1}^{k}a_{l_{i}}(Z)\,.\]
Observe that the application on $\Sm$ is ill-defined in the generic
case, since the derivatives $\Sm^{\left(n\right)}$ do not have local
images for $n>1$. The composition product $\ocomp$ can be seen as
(the dual of) the termwise group action of the Stückelberg-Petermann
group on itself. We implement compatibility with the vertex grading
by defining \begin{align*}
\deg^{v}(a_{k}\ocomp\bigodot_{i=1}^{k}a_{l_{i}}) & :=\deg^{v}(a_{k})+\deg^{v}(\bigodot_{i=1}^{k}a_{l_{i}})\\
 & =k-1+\sum_{i=1}^{k}\left(l_{i}-1\right)=\sum_{i=1}^{k}l_{i}-1,\end{align*}
and with the Hopf algebra unit \[
\openone\ocomp a_{k}=a_{k}\ocomp\openone=a_{k}\,,\]
in accordance with {[}RG2{]}. The product is also compatible with
the coproduct, in the sense that we can define another convolution
on $\Aut(\Halg)$,\[
\phi\bullet_{\linc}\psi:=\linc\circ\left(\phi\otimes\psi\right)\circ\Delta\,.\]
We check that $e\circ\overline{e}$ is a (both sided) identity for
this convolution,\begin{align*}
\left(e\circ\overline{e}\bullet_{\linc}\psi\right)(a_{k}) & =\sum_{\Ptn\in\Part\left\{ 1,\dots,k\right\} }e\circ\overline{e}\left(a_{\left|\Ptn\right|}\right)\ocomp\psi(\bigodot_{I\in\Ptn}a_{\left|I\right|})\\
 & =\sum_{\Ptn\in\Part\left\{ 1,\dots,k\right\} }\delta_{\left|\Ptn\right|,1}\openone\ocomp\psi(\bigodot_{I\in\Ptn}a_{\left|I\right|})\\
 & =\psi(a_{k})\,,\end{align*}
where $\delta_{\left|\Ptn\right|,1}$ is the Kronecker-$\delta$.
Observe that $\left|\Ptn\right|=1$ implies $\Ptn=\left\{ \left\{ 1,\dots,k\right\} \right\} $.
Conversely we get\begin{align*}
\left(\phi\bullet_{\linc}e\circ\overline{e}\right)(a_{k}) & =\sum_{\Ptn\in\Part\left\{ 1,\dots,k\right\} }\phi\left(a_{\left|\Ptn\right|}\right)\ocomp\bigodot_{I\in\Ptn}e\circ\overline{e}(a_{\left|I\right|})\\
 & =\sum_{\Ptn\in\Part\left\{ 1,\dots,k\right\} }\phi\left(a_{\left|\Ptn\right|}\right)\ocomp\bigodot_{I\in\Ptn}\delta_{\left|I\right|,1}\\
 & =\phi(a_{k})\,.\end{align*}
Here the result is obtained since $\left|I\right|=1$ $\forall I\in\Ptn$
implies $\Ptn=\left\{ \left\{ 1\right\} ,\dots,\left\{ k\right\} \right\} $.
We define the $Z$-Feynman rules,\[
\feyn_{Z}(a_{k}):=a_{k}(Z)\equiv Z^{\left(k\right)}(0):\fps{\mathcal{F}_{\loc}(\M)}{\hbar}^{\otimes k}\rightarrow\fps{\mathcal{F}_{\loc}(\M)}{\hbar}\,.\]
These are algebra homomorphisms with respect to both algebra products,\[
\feyn_{Z}:\left(\Halg,\odot\right)\rightarrow\left(\Lin,\odot\right)\qquad\mbox{and}\qquad\feyn_{Z}:\left(\Halg,\ocomp\right)\rightarrow\left(\Lin,\cdot\right)\,,\]
where we denoted by $\Lin$ the space of multi-linear maps between
spaces of local functionals. We get at order $n$ of causal perturbation
theory the finite renormalizations, the changes of renormalization
scheme, or the action of the Stückelberg-Petermann renormalization
group on itself,\[
\left(\feyn_{Z_{1}}\bullet_{\linc}\feyn_{Z_{2}}\right)(a_{n})=\sum_{\Ptn\in\Part\{1,\dots,n\}}Z_{1}^{\left(\left|\Ptn\right|\right)}\cdot\left(\bigodot_{I\in\Ptn}Z_{2}^{\left(\left|I\right|\right)}\right).\]

Let us compute the right sided antipode of the convolution $\bullet_{\linc}$,\begin{align*}
\id_{\Halg}\bullet_{\linc}\ap_{\linc}(a_{n}) & =e\circ\overline{e}(a_{n})\\
\sum_{\Ptn\in\Part\left\{ 1,\dots,n\right\} }a_{\left|\Ptn\right|}\ocomp\ap_{\linc}(\bigodot_{I\in\Ptn}a_{\left|I\right|}) & =\delta_{n,1}\end{align*}
From $n=1$ we get $\ap_{\linc}(a_{1})=\openone$. Furthermore we
get from the connectedness of the Hopf algebra, i.e., from $a_{1}=\openone$,\[
\ap_{\linc}(a_{n})=-\sum_{\Ptn\in\Part\left\{ 1,\dots,n\right\} \backslash\left\{ \Ptn_{1}\right\} }a_{\left|\Ptn\right|}\ocomp\ap_{\linc}(\bigodot_{I\in\Ptn}a_{\left|I\right|})\,,\]
where, as before, $\Ptn_{1}=\left\{ \left\{ 1,\dots,n\right\} \right\} $.
Observe the similarity to the result of Connes-Kreimer in the commutative
case, i.e., the recursion formula (\ref{eq:CKantipode}) for the antipode
$\ap$. And observe the difference, the composition product $\ocomp$.
The augmented Hopf algebra\[
\left(\Halg,\odot,e,\Delta,\overline{e},\ocomp,\ap_{\linc}\right)\]
constructed above can be interpreted as the algebraic dual of the
St{\"u}ckelberg-Petermann renormalization group. We want to discuss
now the action of this algebraic dual on the (regularized) time-ordered
products, and its relation to the original formulation of the Connes-Kreimer
theory of renormalization.

Given the nice circumstance that we have a preferred renormalization
prescription at all orders of perturbation theory as it is provided
by any analytic regularization of the $\Sm$-matrix combined with
minimal subtraction. Then we define the regularized Feynman rules
as\[
\feyn_{\mu,\zeta}(a_{k}):=a_{k}(\Sm_{\mu,\zeta})\equiv\Sm_{\mu,\zeta}^{\left(k\right)}:\fps{\mathcal{F}_{\loc}(\M)}{\hbar}^{\otimes k}\rightarrow\fps{\mathcal{F}(\M)}{\hbar}\,.\]
Since the derivatives of $\Sm_{\mu,\zeta}$ have non-local images
in the generic case, we have that, in contrast to the $Z$-Feynman
rules above, the regularized Feynman rules cannot be iterated. This
is reflected in the algebraic setting by the fact that the regularized
Feynman rules are algebra homomorphisms with respect to $\odot$,
but not with respect to $\ocomp$, $ $\[
\feyn_{\mu,\zeta}(a_{k}\odot a_{l})=\feyn_{\mu,\zeta}(a_{k})\odot\feyn_{\mu,\zeta}(a_{l})\,.\]
Following the idea of Kreimer \citep{Kreimer1999}, we define\[
\ap_{\linc,\MS}^{\feyn_{\mu,\zeta}}:=R\circ\feyn_{\mu,\zeta}\circ\ap_{\linc}\,,\]
where $R$ denotes minimal subtraction, i.e., the {}``renormalization
map'',\[
R(\Sm_{\mu,\zeta}^{\left(k\right)})=\begin{cases}
\id & \mbox{for }k=1\\
\pp(\Sm_{\mu,\zeta}^{\left(k\right)}) & \mbox{for }k>1\,.\end{cases}\]
Assuming that $R$ is only applied to prepared time-ordered products,
we have that $R(\Sm_{\mu,\zeta}^{\left(k\right)})$ is a multi-linear
map from local functionals to local functionals, and then $R\circ\feyn_{\mu,\zeta}$
defines an algebra homomorphism of $\left(\Halg,\odot,\ocomp\right)$
with respect to both products\[
R\circ\feyn_{\mu,\zeta}\left(a_{k}\ocomp\bigodot_{i=1}^{k}a_{l_{i}}\right)=R(\Sm_{\mu,\zeta}^{\left(k\right)})\cdot\left(\bigodot_{i=1}^{k}R(\Sm_{\mu,\zeta}^{\left(l_{i}\right)})\right).\]
That $R$ is applied only to prepared time-ordered products is guaranteed
by the recursive definition of the antipode. In particular $R$ itself
is a homomorphism of the symmetrized tensor product $\odot$, which
makes the Rota-Baxter argument of \citep{Kreimer1999} redundant in
the presented framework (see Remark~\ref{rem:RotaBaxter} below).
We infer that $\ap_{\linc,\MS}^{\feyn,\mu,\zeta}$ is an algebra homomorphism\[
\ap_{\linc,\MS}^{\feyn,\mu,\zeta}:\left(\Halg,\odot\right)\rightarrow\left(\Lin,\odot\right)\,.\]
We have that $\ap_{\linc,\MS}^{\feyn_{\mu,\zeta}}(a_{1})=\id$, the
identity on local functionals, and for $n>1$, \[
\ap_{\linc,\MS}^{\feyn_{\mu,\zeta}}(a_{n})=-R\sum_{\Ptn\in\Part\left\{ 1,\dots,n\right\} \backslash\left\{ \Ptn_{1}\right\} }a_{\left|\Ptn\right|}(\Sm_{\mu,\zeta})\cdot\ap_{\linc,\MS}^{\feyn_{\mu,\zeta}}(\bigodot_{I\in\Ptn}a_{\left|I\right|})\,.\]
This is just the recursion formula we got for the counterterms (\ref{eq:CountertermsMSRecursionFormula})
by applying the minimal subtraction condition to the main theorem
of renormalization in the form of Lemma~\ref{lem:FaaDiBrunoSetPartitionVersion}.
Observe that one loses the information on the product $\ocomp$ if
one regards the amplitudes, or regularized time-ordered products as
elements in a commutative ring of Laurent series, only. However, similar
as in Connes-Kreimer theory of renormalization, we can define the
algebra homomorphism \begin{equation}
\ap_{\linc,\ren}^{\feyn_{\mu,\zeta}}=\feyn_{\mu,\zeta}\:\bullet_{\linc}\:\ap_{\linc,\MS}^{\feyn_{\mu,\zeta}}\,,\label{eq:FiniteAntipode}\end{equation}
which gives the finitely regularized $n$-fold time ordered product,
when applied to a generator $a_{n}\in\Halg$,\begin{equation}
\ap_{\linc,\ren}^{\feyn_{\mu,\zeta}}(a_{n})=\Sm_{\mu,\zeta,\ren}^{\left(n\right)}.\label{eq:FiniteAntipodeTimeOrdProd}\end{equation}
However, although (\ref{eq:FiniteAntipode}) gives the correct result
in the general case, i.e., for arbitrary local interactions, this
is merely a compact notation for the successive subtraction of counterterms
in the sense of BPH rather than a forest formula in the sense of Zimmermann
\citep{Zimmermann1969,Zimmermann1975} or Theorem~\ref{thm:EGFF}.
Solving the recursion would be equivalent to giving a closed formula
for the character $\ap_{\linc,\ren}^{\feyn_{\mu,\zeta}}$ in (\ref{eq:FiniteAntipodeTimeOrdProd}).
This was done in Section~\ref{sec:EGFF}, although not in this abstract
algebraic setting. After the remark on Rota-Baxter maps, we give a
graphical interpretation of the maps constructed above.

\vspace{-1mm}
\begin{rem}
[Rota-Baxter Maps]\label{rem:RotaBaxter}For the reader, less familiar
with the {}``Hopf algebra school'' in renormalization theory, we
probably have to remark here that a linear map fulfilling the algebraic
relation\begin{equation}
R(a)R(b)=R\left[R(a)b\right]+R\left[aR(b)\right]-R(ab)\,,\label{eq:RotaBaxter}\end{equation}
is called a \emph{Rota-Baxter map of weight one}. Examples of such
maps are the projections in a \emph{Birkhoff sum}, i.e., an algebra
which splits into a direct sum of algebras, both closed under multiplication,\[
A=A^{+}\oplus A^{-}.\]
The Laurent series with scalar coefficients are elements of a Birkhoff
sum. It is straight forward to show the above claim that any linear
projection in a Birkhoff sum, which projects to one of its components
is Rota-Baxter of weight one. Denote by $a^{+}+a^{-}\in A^{+}\oplus A^{-}$
the elements of a Birkhoff sum, then by linearity and the fact that
the components are closed under multiplication we have\begin{align*}
\left[\left(a^{+}+a^{-}\right)\left(b^{+}+b^{-}\right)\right]^{-} & =\left[a^{-}b^{+}+a^{+}b^{-}\right]^{-}+a^{-}b^{-}\\
 & =\left[a^{-}b^{+}+a^{+}b^{-}+2a^{-}b^{-}\right]^{-}-a^{-}b^{-}\\
 & =\left[a^{-}b\right]^{-}+\left[ab^{-}\right]^{-}-a^{-}b^{-}.\end{align*}
Following a suggestion by Brouder, Kreimer used this property for
the projection to the principal part of a Laurent series in order
to solve the {}``multiplicativity constraints'', see~\citep[Sec.~3]{Kreimer1999}.
Kreimer's {}``multiplicativity constraints'' were originally formulated
to have the map \[
\pp\circ\phi\circ\ap:\left(\Halg,\odot\right)\rightarrow\mathcal{L}^{-}\,,\]
defined as a character of the commutative Hopf algebra to the ring
of scalar Laurent series $\mathcal{L}$. $\phi:\Halg\rightarrow\Lrt$
denotes the character which induces Kreimer's (regularized) Feynman
rules, $\ap$ is the antipode of the commutative Hopf algebra, and
$\pp:\Lrt^{+}\oplus\Lrt^{-}\rightarrow\Lrt^{-}$ is the projection
to the principal part. Kreimer was able to show that the Rota-Baxter
condition implies the multiplicativity of $R\circ\phi\circ\ap$ (Prop.~2,
loc. cit.). More advanced topics related to Rota-Baxter algebras in
Physics and Mathematics partially induced by Kreimer's observation
can be found, e.g., in \citep{Ebrahimi-FardGuo2007}.

Observe, however, that in the framework advocated here all maps, and
in particular $R$, are homomorphisms of the commutative algebra $\left(\Halg,\odot\right)$,
so that a Rota-Baxter argument is not necessary, since any algebra
homomorphism, trivially, is a Rota-Baxter map of weight one.
\end{rem}

\subsection{Graphs}

The role of the additional composition product $\ocomp$ and antipode
$\ap_{\linc}$ in \[
\left(\Halg,\odot,e,\Delta,\overline{e},\ap,\ocomp,\ap_{\linc}\right)\,.\]
may become clearer if we break them down to the graph level. This
graphwise interpretation is regained, if we insert for $\Sm$ its
perturbative expansion (\ref{eq:PerturbationSeries}), and regard
the corresponding operations on the level of the graph contributions.
This can be done since all maps involved are linear and since by condition
{[}C3{]} all sums are finite at each fixed order of $\hbar$. Remember
that we only regard full vertex parts as subgraphs.

In accordance with the structure of $\left(\Halg,\odot,e\right)$,
we regard the abstract algebra of graphs with disjoint union $\dot{\cup}$
as product and the empty set $\emptyset$ as unit. Let $\Gamma\in\mathcal{G}$
be a graph with $n$ vertices, hence a contribution to $a_{n}(\Sm$).
Let $\Part^{c}V(\Gamma)$ be the set of all connected partitions of
the vertex set of $\Gamma$. By connected partition we mean a partition
$\Ptn$ whose blocks $I\in\Ptn$ give rise to connected full vertex
parts $\gamma_{I}\subset\Gamma$. We can restrict to connected partitions,
since the principal parts of the regularized amplitudes corresponding
to disconnected graphs vanish. Denote by $\Gamma/\Ptn$ the graph,
which has the blocks $I\in\Ptn$ as vertices and as lines the lines
in $\Gamma$, which connect different blocks of $\Ptn$. Let $\gamma_{I}$,
$I\in\Ptn$, be the full vertex part of the block $I\in\Ptn$. Then
we can write the coproduct on the level of graphs as\[
\Delta\Gamma=\sum_{\Ptn\in\Part^{c}V(\Gamma)}\Gamma/\Ptn\otimes\dot{\bigcup_{I\in\Ptn}}\gamma_{I}\,.\]
Applying the map $\id_{\Halg}\otimes\ap_{\linc}$ gives\[
\left(\id_{\Halg}\otimes\ap_{\linc}\right)\circ\Delta\Gamma=\sum_{\Ptn\in\Part^{c}V(\Gamma)}\Gamma/\Ptn\otimes\dot{\bigcup_{I\in\Ptn}}\ap_{\linc}(\gamma_{I})\,,\]
which corresponds to (recursively) computing the counterterms for
the connected subgraphs $\gamma_{I}$, e.g., in DimReg+MS. Applying
the composition product $\linc$, inserts the counterterms $\ap_{\linc}(\gamma_{I})$
at the vertices of $\Gamma/\Ptn$ to give one contribution to the
{}``renormalized graph'',\[
\id_{\Halg}\bullet_{\linc}\ap_{\linc}=\sum_{\Ptn\in\Part^{c}V(\Gamma)}\Gamma/\Ptn\ocomp\dot{\bigcup_{I\in\Ptn}}\ap_{\linc}(\gamma_{I})\,.\]
It is one of the results of the present thesis that this procedure,
performed on the level of analytically regularized amplitudes in arbitrary
representation (momentum or position space) leads to finitely regularized,
i.e., UV convergent integrals, and local counterterms in all orders
of perturbation theory.

\end{fmffile}

\chapter*{Conclusion}

The investigation undertaken for this thesis has shown that the methods
of dimensional regularization and minimal subtraction can consistently
be implemented into causal perturbation theory in the framework of
perturbative Algebraic Quantum Field Theory (pAQFT). This enriches
the framework by a renormalization technique, which has a preferred
extension at all orders of causal perturbation theory, given the fact
one disposes of an analytic regularization of the $\Sm$-matrix. A
concrete form of such a regularization was given in Minkowski position
space in terms of the dimensionally regularized scattering matrix,
$\Sm_{\mu,\zeta}$. It was proven that the incorporation of an analytically
regularized $\Sm$-matrix makes it possible to solve the Epstein-Glaser
induction at all orders in perturbation theory, and the result was
given in terms of the Epstein-Glaser forest formula (Theorem~\ref{thm:EGFF}).
This result was derived directly from the main theorem of renormalization
and was given in a form which is independent of the chosen representation.
In particular the derived forest formula is valid in both, momentum
and position space, whatever space its better suited for the concrete
calculation.

Besides this forest formula, I gave a direct derivation of the Hopf
algebra of Feynman graphs from the main theorem of perturbative renormalization.
This Hopf algebra first occurred in the work of Kreimer and Connes-Kreimer
in their analysis of BPHZ renormalization \citep{Kreimer1998,Connes2000,Connes2001}
and was later found also in causal perturbation theory \citep{Gracia-Bondia2000,Pinter2000b}.
In the present thesis the Hopf algebra of graphs was derived in a
summed up form, i.e. the elements can be regarded as sums over all
graphs with the same set of vertices. This is in accordance with the
findings of \citep{BrouderFrabetti2000,BrouderFrabetti2001,Frabetti2007,Suijlekom2007c}.
The reduced Hopf algebra of Pinter with only full vertex parts can
be derived by linearity, however, the pure BPHZ subgraphs of the Connes-Kreimer
approach do not emerge here. This is in accordance with the proof
of Zimmermann that pure BPHZ graphs do not contribute to the renormalized
amplitude \citep{Zimmermann1975}. Although the Hopf algebra of graphs
emerged as a commutative, non-cocommutative Hopf algebra, we could
show that it is necessary to augment it with a non-commutative product
stemming from the composition of linear maps in order to get the recursion
relation for the pAQFT counterterms, which has been derived independently
from the main theorem. The recursion relation for the counterterms
is described algebraically as the antipode of the convolution induced
by the coproduct and the additional non-commutative composition product.
A main difference to the Connes-Kreimer theory of renormalization
is that in the algebraic setting described here the Feynman rules
emerge naturally from the theory and are not assumed to give Laurent
series with scalar coefficients. I see applications of these results
in three main areas of current research in physics, mathematical physics,
and mathematics.

First, physics. The forest formula was proven directly for the time-ordered
products. However, we also gave a {}``graph form'' of the formula,
which could be relevant for concrete computations. I want to emphasize,
that the combinatorial pattern underlying the Epstein-Glaser forest
formula is much simpler than the one underlying Zimmermann's original
version, which is still used in modern computations. Spurious subtractions
do not occur in the EG forest formula. Although the spurious subtractions
do not play a role in QED calculations, they do occur in Quantum Chromo
Dynamics (QCD), since this theory has a four valent vertex. Besides
this, there is a second simplification in the forest formula proven
in this work. I showed that one can replace the Zimmermann forests
of full vertex parts by totally ordered sets of partitions of the
vertex set. This simplifies a lot the intricate combinatorics of Zimmermann's
forest formula, and might make it possible to implement the advocated
method in an algorithm. This, in turn, is certainly relevant for the
computation of higher order contributions to the perturbative expansion
in high energy physics phenomenology. Since the method was proven
for any analytic regularization, also gauge symmetries should be preserved
if one chooses a regularization which preserves these symmetries.
This assertion is affirmed by the concrete computations of \citep{Falk2009},
however, the case of gauge theories was not discussed in the present
thesis.

Second, mathematical physics. The covariant formulation of perturbative
Algebraic Quantum Field Theory makes the formalism applicable also
in curved, globally hyperbolic spacetimes. Although the construction
of the regularized $\Sm$-matrix was done in the present thesis for
Minkowski space, the results of the last chapter, in particular the
recursion relation for the counterterms (Corollary~\ref{cor:RecursionMSCounterterms})
and the Epstein-Glaser forest formula (Theorem~\ref{thm:EGFF}) were
derived in the more general, covariant framework. Thus they can be
applied directly, given the fact that one disposes of an analytically
regularized $\Sm$-matrix. Considering the convenient properties dimensional
regularization has in flat spacetime when it comes to gauge theories,
one may want to have a similar concept in curved spacetime. However,
in the construction of the dimensionally regularized $\Sm$-matrix
in \linebreak[4]Minkowski spacetime we made explicit use of translation
invariance, and the choice of relative coordinates was made using
the graph cohomology. As shown in \citep{Brunetti2000} the wave front
set condition on local functionals can be understood as a microlocal
remnant of translation invariance. However, one has to understand
better the interplay of this microlocal condition with the graph cohomology
in order to give a direct translation of the results.

Third, for mathematics. The analysis of algebraic structures is an
active field of research in pure mathematics which is of interest
in its own right. The Hopf algebra found by Connes and Kreimer in
perturbative renormalization theory affected this research on a profound
basis, and the relation to the main theorem of perturbative renormalization
and the framework of pAQFT which was established in this thesis could
possibly be a new seed for research in this field. I showed that it
is necessary to incorporate an additional composition structure into
the Hopf algebra to have an interpretation for the antipode in terms
of minimally subtracted counterterms in pAQFT, with the merit of having
naturally emerging Feynman rules. There are many more intriguing questions
about the connection of perturbative quantum field theory to pure
mathematics. Questions about the role of multiple zeta values, graph
polynomials, shuffle and stuffle products and the like in (algebraic)
quantum field theory and causal perturbation theory. Such relations
were established on the level of examples in the pioneering works
of Connes and Kreimer \citep{Connes2000,Connes2001} and Bloch, Esnault,
and Kreimer \citep{BlochEsnaultKreimer2006}. One might hope that
the tools developed in this thesis contribute to further investigation
of the suggested relations, and to a better understanding of the relation
of the framework of perturbative algebraic quantum field theory to
the more abstract algebraic setting of Connes and Marcolli \citep{Connes2004a,Connes2004b}.

\begin{appendix}

\chapter{\label{app:ModifiedBesselFunctions}Solutions of the Modified Bessel
equation}

In this Appendix we briefly review the solution theory of the (modified)
Bessel equation. The interested reader may want to refer to \citep{SpainSmith1970}
(e.g.) for a more detailed discussion of the topic. The modified Bessel
equation%
\footnote{The modified Bessel equation (\ref{eq:ModifiedBesselEquation}) is
related to Bessel's differential equation by the coordinate transform
$x\mapsto ix$.%
}\begin{equation}
\frac{d^{2}y}{dx^{2}}+\frac{1}{x}\frac{dy}{dx}-\frac{1}{x^{2}}\left(x^{2}+\nu^{2}\right)y=0\,,\quad\Re(\nu)\geq0\label{eq:ModifiedBesselEquation}\end{equation}
is a second order ordinary differential equation with a \emph{regular
singular point} at the origin. That is, (at least one of) the coefficients
$p(x)$ of $\frac{dy}{dx}$ and $q(x)$ of $y$ are singular at $x=0$,
but $x\, p(x)$ and $x^{2}q(x)$ are regular in a neighborhood of
zero. Let \[
x\, p(x)=\sum_{n=0}^{\infty}p_{n}x^{n}=1\qquad\mbox{and}\qquad x^{2}q(x)=\sum_{n=0}^{\infty}q_{n}x^{n}=-\left(x^{2}+\nu^{2}\right)\]
be the corresponding Taylor expansions. A differential equation with
a regular singular point at $0$ is solved with the ansatz \[
y=x^{\alpha}\sum_{s=0}^{\infty}c_{s}x^{s}\,,\quad c_{0}\neq0\,.\]
In order $x^{\alpha}$ one finds the \emph{indicial equation}\[
\alpha\left(\alpha-1\right)+p_{0}\alpha+q_{0}=0\,,\]
whose roots are called the \emph{exponents} of the differential equation.
In the case of the (modified) Bessel equation (\ref{eq:ModifiedBesselEquation}),
we evidently have\[
\alpha\left(\alpha-1\right)+\alpha-\nu^{2}=\left(\alpha+\nu\right)\left(\alpha-\nu\right)=0\,,\]
hence (\ref{eq:ModifiedBesselEquation}) has exponents $\alpha\in\left\{ \pm\nu\right\} $.
It is a straight forward calculation to see that, in the case $\nu\in\mathbb{C}\backslash\mathbb{N}_{0}$,
we have $c_{1}=0$ and \begin{equation}
c_{s}=\frac{c_{s-2}}{\left(s+\alpha+\nu\right)\left(s+\alpha-\nu\right)}\,,\quad\alpha\in\left\{ \pm\nu\right\} ,\quad\nu\in\mathbb{C}\backslash\mathbb{N}_{0}.\label{eq:ModifiedBesselRecursion}\end{equation}
This leads to the linear independent set of solutions $\left\{ I_{-\nu},I_{\nu}\right\} $,
where\begin{equation}
I_{\nu}(x)=\sum_{s=0}^{\infty}\frac{1}{s!\,\Gamma(\nu+s+1)}\left(\frac{x}{2}\right)^{\nu+2s}\,,\quad\nu\in\mathbb{C}\backslash\mathbb{N}_{0}.\label{eq:ModifiedBesselFirstKind}\end{equation}
The functions $I_{\nu}$ and $I_{-\nu}$ are called modified Bessel
function of first kind, they are related to the Bessel functions of
first kind $J_{\nu}$ by\[
I_{\nu}=i^{-\nu}J_{\nu}(ix)\,.\]

Observe that in the case of integer order, $n\in\mathbb{N}_{0}$,
$I_{n}=I_{-n}$, as can also be seen from the recursion relations
(\ref{eq:ModifiedBesselRecursion}). This is a general feature of
(second order) ordinary differential equations with regular singular
point. Problems occur, when their exponents differ by an integer value.

A linearly independent set of solutions of (\ref{eq:ModifiedBesselEquation})
for arbitrary order can be constructed in the following way. For non-integer
order, $\nu\in\mathbb{C}\backslash\mathbb{N}_{0}$, we replace $I_{-\nu}$
by \begin{equation}
K_{\nu}:=\frac{\pi}{2\sin(\nu\pi)}\left[I_{-\nu}-I_{\nu}\right]\,,\quad\nu\in\mathbb{C}\backslash\mathbb{N}_{0}\,,\label{eq:BesselFunctionsSecondKind}\end{equation}
evidently giving a linearly independent set of solutions $\left\{ I_{\nu},K_{\nu}\right\} $
for non-integer order. The limit \[
K_{n}:=\lim_{\nu\rightarrow n}K_{\nu}\]
exists for all integers $n\in\mathbb{N}_{0}$, and $\left\{ I_{n},K_{n}\right\} $
is a complete, linear independent set of solutions of (\ref{eq:ModifiedBesselEquation}).
That $\left\{ I_{n},K_{n}\right\} $ is linearly independent, also
in the case $n\in\mathbb{N}_{0}$ can be seen by the following argument.
Just setting $\nu=n$ clearly results in the situation $\frac{0}{0}$,
hence we apply l'Hôspital's rule to compute the limit \[
\lim_{\nu\rightarrow n}K_{\nu}=\frac{\pi}{2}\left[\frac{\left(\partial_{\nu}I_{-\nu}\right)-\partial_{\nu}I_{\nu}}{\pi\cos(\nu\pi)}\right]_{\nu=n}=\tfrac{\left(-1\right)^{n}}{2}\left[\left(\partial_{\nu}I_{-\nu}\right)-\partial_{\nu}I_{\nu}\right]_{\nu=n}\]
We have seen that $I_{\nu}(x)\sim x^{\nu}f_{\nu}(x^{2})$, with entire
analytic functions $f_{\nu}$, hence the derivatives of $I_{\nu}$
with respect to $\nu$ will introduce logarithmic terms, $\partial_{\nu}\left(x^{\nu}\right)=\ln(x)\cdot x^{\nu}$,
which do not cancel, since $f_{\nu}\neq f_{-\nu}$ for all $\nu\neq0$.

The introduction of logarithmic terms in the limit $\nu\rightarrow n$
is not a feature of the special choice of $K_{\nu}$ but merely a
consequence of the singularity structure of the (modified) Bessel
differential equation at the origin (see e.g. \citep[Sec. 1.6]{SpainSmith1970}).

\end{appendix}

\begin{center}
\textbf{\Huge Danke.}
\par\end{center}{\Huge \par}

\vspace{15mm}

\thispagestyle{empty}

\noindent Mein größter Dank gilt dem Betreuer meiner Dissertation,
Klaus Fredenhagen. Ohne seinen Überblick und sein tiefes Verständnis
wären die Ergebnisse dieser Arbeit wohl kaum zu Stande gekommen. Neben
den vielen Dingen, die ich hier als Dankesgründe anführen könnte möchte
ich ihm vor allem für seine Geduld in den gemeinsamen Diskussionen
und während der gesamten Bearbeitungszeit danken. Ich hätte mir keinen
besseren Betreuer für meine Doktorarbeit wünschen können.

\noindent Ich möchte mich bei der gesamten Hamburger AQFT-Gruppe
bedanken, die mir immer mit Rat und Tat zur Seite stand. Vor allem
die ungezwungenen Diskussionen beim Kaffee werden mir fehlen. Allen
voran möchte ich Thomas Hack danken, {}``Sag mal Thomas, Du kennst
Dich doch aus mit ...'' I'd like to express my gratidude to our PostDocs
Nicola Pinamonti, Claudio Dappiaggi, and Pedro \linebreak[3]Lauridsen
Ribeiro for many helpful discussions with and without coffee and cookies.
Bei Katarzyna Rejzner möchte ich mich für die hilfreichen Hinweise
in der Endphase der Arbeit bedanken, und ich begrüße natürlich alle
neuen Mitglieder der Arbeitsgruppe.

\noindent Ich danke Christian Fleischhack f{\"u}r den {}``Privatunterricht''
in Funktionalanalysis.

\noindent Ich möchte meinem guten Freund und Weggefährten Christian
Bogner für viele hilfreiche Diskussionen und seine moralische Unterstützung
danken.

\noindent I'd like to thank Romeo Brunetti and Michael Dütsch, who
clarified more than one misunderstanding on my side.

\noindent Many thanks go to Alessandra Frabetti, Christian Brouder,
and José Gracia-Bondía for the discussions in Cargèse. I am grateful
to José for directing my attention to the almost forgotten proof by
Zimmermann and for kindly accepting to be a referee of this thesis.
I want to thank Alessandra for the invitation to Lyon, many discussions,
manifestations, and the great hospitality I was offered there.

\noindent Además quiero decirles un gran Gracias a los organizadores
de la escuela de verano en Villa de Leyva. Gracias a Andrés Reyes
por su amistad, la invitación y su apoyo. Hasta la próxima! Gracias
a Sylvie Paycha y Francis Brown por las discussiones sobre los aspectos
matemáticos de las integrales y las diagramas de Feynman.

\noindent I'd like to thank Christoph Bergbauer for his useful remarks
and corrections.

\noindent Ich möchte meiner Freundin Corinna für so vieles danken,
mir fehlen nur die Worte. \hspace{50mm}\textifsymbol{"0C} \textifsymbol{"1C} \textifsymbol[ifwea]{"11}

\noindent Ich danke meiner Familie, meiner Schwester Marion und meinen
Eltern, auf deren Unterstützung ich auch über die Entfernung immer
bauen konnte.

\chapter*{References}

\nocite{CaswellKennedy1982}

\begin{spacing}{1.23}

\bibliographystyle{kaialpha}
\begin{btSect}{Literatur_Physik}
\btPrintCited
\end{btSect}

\end{spacing}
\end{document}